\documentclass{LMCS}

\def\dOi{11(3:11)2015}
\lmcsheading%
{\dOi}
{1--43}
{}
{}
{Jul.~\phantom02, 2014}
{Sep.~17, 2015}
{}

\ACMCCS{[{\bf Theory of computation}]: Logic---Verification by model
  checking\,/\,Abstraction; [{\bf Software and its Engineering}]:
  Software organization and properties---Software functional
  properties---Formal methods---Model checking}

\usepackage{amsmath}
\usepackage{amssymb}
\usepackage{algorithm}
\usepackage[noend]{algorithmic}
\usepackage{subfig}
\usepackage{hyperref}
\usepackage[english]{babel}
\usepackage{todonotes}
\usepackage{tikz}
\usetikzlibrary{arrows, shapes, shapes.multipart, calc, positioning, 
shapes.misc, 
decorations}

\newcommand{\err}{\mathit{err}}

\begin{document}

\title[Abstract Model Repair]{Abstract Model Repair\rsuper*}
\author[G.~Chatzieleftheriou]{George~Chatzieleftheriou\lowercase{$^a$}}
\address{$^a$Department of Informatics, Aristotle University of Thessaloniki, 
54124 Thessaloniki, Greece}
\email{gchatzie@csd.auth.gr}

\author[B.~Bonakdarpour]{Borzoo~Bonakdarpour\lowercase{$^b$}}
\address{$^b$Department of Computing and Software, McMaster University, 
1280 Main Street West, Hamilton, ON L8S 4L7, Canada}
\email{borzoo@mcmaster.ca}

\author[P.~Katsaros]{Panagiotis~Katsaros\lowercase{$^c$}}
\address{$^c$Department of Informatics, Aristotle University of Thessaloniki, 
54124 Thessaloniki, Greece}
\email{katsaros@csd.auth.gr}

\author[S.~A.~Smolka]{Scott~A.~Smolka\lowercase{$^d$}}
\address{$^d$Department of Computer Science, Stony Brook University, 
 Stony Brook, NY 11794-4400, USA}
\email{sas@cs.sunysb.edu}

\keywords{Model Repair, Model Checking, Abstraction Refinement}
\titlecomment{{\lsuper*}A preliminary version of the paper has appeared in~\cite{GBSK12}}

\begin{abstract}
	Given a Kripke structure $M$ and CTL formula $\phi$, 
	where $M$ does not satisfy $\phi$, the problem of \emph{Model Repair} 
	is to obtain a new model $M'$ such that $M'$ satisfies $\phi$.  
	Moreover, the changes made to $M$ to derive $M'$ should be minimum 
	with respect to all such $M'$.  As in model checking, 
	\emph{state explosion} can make it virtually impossible to carry out 
	model repair on models with infinite or even large state spaces.  
	In this paper, we present a framework for model repair that uses 
	\emph{abstraction refinement} to tackle state explosion.  Our framework
	aims to repair Kripke Structure models based on a Kripke Modal 
Transition System abstraction and a 3-valued semantics for CTL.  We 
introduce 
	an abstract-model-repair algorithm for which we prove soundness and
	semi-completeness, and we study its complexity class.  Moreover, a prototype 
	implementation is presented to illustrate the practical utility of abstract-model-repair on an 
	Automatic Door Opener system model and a model of the Andrew File System 1 protocol.  
\end{abstract}

\maketitle

\section{Introduction}
\label{sec:intro}	

Given a model $M$ and temporal-logic formula $\phi$,
\emph{model checking}~\cite{CES09} is the problem of determining whether or not
$M \models \phi$.  When this is not the case, a model checker will typically 
provide a \emph{counterexample} in the form of an execution path along which 
$\phi$ is violated.  The user should then process the counterexample manually 
to correct $M$.

An extended version of the model-checking problem is that of 
\emph{model repair}: given a model $M$ and temporal-logic formula
$\phi$, where $M \not\models \phi$, obtain a new model $M'$,
such that $M' \models \phi$. The problem of Model Repair for Kripke structures 
and Computation Tree Logic (CTL)~\cite{EH85} properties was first introduced 
in~\cite{BEGL99}.

\emph{State explosion} is a well known limitation of automated formal
methods, such as model checking and model repair, which impedes their
application to systems having large or even infinite state spaces.
Different techniques have been developed to cope with this problem.
In the case of model checking, 
\emph{abstraction}~\cite{CGL94,LGSBBP95,GS97,DGG97,GHJ01} is used to 
create a smaller, more abstract version $\hat{M}$ of the initial concrete
model $M$, and model checking is performed on this smaller model.
For this technique to work as advertised, it should be the case that
if $\hat{M} \models \phi$ then $M \models \phi$.

Motivated by the success of abstraction-based model checking, we present
in this paper a new framework for Model Repair that uses
\emph{abstraction refinement} to tackle state explosion.  The resulting
\emph{Abstract Model Repair} (AMR) methodology makes it possible to
repair models with large state spaces, and to speed-up the repair
process through the use of smaller abstract models.  The major
contributions of our work are as follows:

\begin{itemize}

\item
We provide an AMR framework that uses Kripke structures (KSs) for the concrete 
model $M$, Kripke Modal Transition Systems (KMTSs) for the abstract model 
$\hat{M}$, and a 3-valued semantics for interpreting CTL over 
KMTSs~\cite{HJS01}.  An iterative refinement of the abstract 
KMTS model takes place whenever the result of the 3-valued CTL model-checking 
problem is undefined.  If the refinement process terminates with a KMTS that  
violates the CTL property, this property is also falsified by the concrete 
KS $M$.  Then, the repair process for the refined KMTS is initiated.  

\item
We strengthen the Model Repair problem by additionally taking into
account the following \emph{minimality} criterion (refer to the
definition of Model Repair above): the changes made to $M$ to derive
$M'$ should be minimum with respect to all $M'$ satisfying $\phi$.
To handle the minimality constraint, we define a metric space over KSs
that quantifies the structural differences between them.

\item
We introduce an Abstract Model Repair algorithm for KMTSs, which 
takes into account the aforementioned minimality criterion.

\item 
We prove the soundness of the Abstract Model Repair algorithm for the full CTL and 
the completeness for a major fragment of it.  Moreover, the algorithm's complexity 
is analyzed with respect to the abstract KMTS model size, which can be
much smaller than the concrete KS.     

\item
We illustrate the utility of our approach through a prototype implementation used to 
repair a flawed Automatic Door Opener system~\cite{BK08} and the Andrew File System 1 
protocol. Our experimental results show significant improvement in efficiency 
compared to a concrete model repair solution.    
\end{itemize}

\noindent\emph{Organization. } \ The rest of this paper is organized as follows.
Sections~\ref{sec:mc} and~\ref{sec:abstr} introduce KSs, KMTSs,
as well as abstraction and refinement based on a 3-valued semantics
for CTL.  Section~\ref{sec:mrp} defines a metric space for 
KSs and formally defines the problem of Model Repair.
Section~\ref{sec:absmrp} presents our framework for
Abstract Model Repair, while Section~\ref{sec:alg} introduces the
abstract-model-repair algorithm for KMTSs and discusses its soundness, 
completeness and complexity properties.  Section~\ref{sec:exp} presents the experimental 
evaluation of our method through its application to the Andrew File 
System 1 protocol (AFS1).  Section~\ref{sec:relwork} considers 
related work, while Section~\ref{sec:concl} concludes with a review of the overall
approach and pinpoints directions for future work.

\section{Kripke Modal Transition Systems}
\label{sec:mc}	\enlargethispage{\baselineskip}

Let $AP$ be a set of {\em atomic propositions}.  Also, let $Lit$ be the set of
{\em literals}:
\[
Lit = AP \; \cup \; \{ \neg p \mid p \in AP\}
\]

\begin{defi}
\label{def:ks}
A {\em Kripke Structure} (KS) is a quadruple $M = (S, S_{0}, R, L)$,
where:
\begin{enumerate}
\item $S$ is a finite set of {\em states}.
\item $S_{0}\subseteq S$ is the set of {\em initial states}.
\item $R\subseteq S \times S$ is a {\em transition relation} that must
  be total, i.e., $$\forall s \in S: \exists s' \in S:R(s,s').$$
\item $L: S \rightarrow 2^{Lit}$ is a state {\em labeling function},
  such that $$\forall s \in S: \forall p \in AP: p \in L(s)
  \Leftrightarrow \neg p \notin L(s).\eqno{\qEd}$$ \qedhere
\end{enumerate} 
\end{defi}

\noindent The fourth condition in Def.~\ref{def:ks} ensures that any
atomic proposition $p \in AP$ has one and only one truth value at any
state.\\

\noindent \emph{Example.} We use the Automatic Door Opener system (ADO)
of~\cite{BK08} as a running example throughout the paper.  The system,
given as a KS in Fig~\ref{fig:ado_system}, requires a three-digit code 
$(p_{0},p_{1},p_{2})$ to open a door, allowing for one and only one wrong digit 
to be entered at most twice.  Variable $\mathit{err}$ counts the number of 
errors, and an alarm is rung if its value exceeds two.  For the purposes of our 
paper, we use a simpler version of the ADO system, given as the KS $M$ in
Fig.~\ref{fig:ado_initial}, where the set of atomic propositions is
$AP = \{q\}$ and $q \equiv (open = true)$.

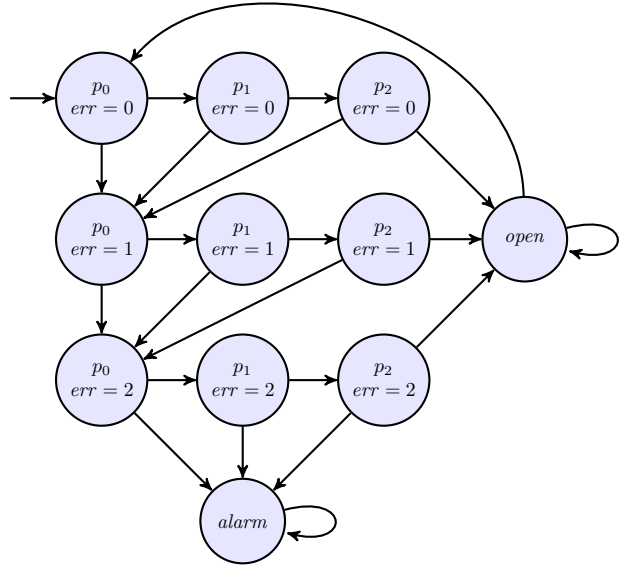
\begin{figure}[t]
\centering
\begin{tikzpicture}[->,>=stealth',auto,node 
distance=2.5cm, scale=0.75, thick, main node/.style={scale=0.75, minimum size = 
1.5cm, align=center,circle,fill=blue!10,draw, font=\small}]

  \node[main node] (1) {$p_0$ \\ $\err =0$};
  \node[main node] (2) [below of=1] {$p_0$ \\ $\err =1$};
  \node[main node] (3) [below of=2] {$p_0$ \\ $\err =2$};
  \node[main node] (4) [right of=1] {$p_1$ \\ $\err =0$};
  \node[main node] (5) [below of=4] {$p_1$ \\ $\err =1$};
  \node[main node] (6) [below of=5] {$p_1$ \\ $\err =2$};
  \node[main node] (7) [right of=4] {$p_2$ \\ $\err =0$};
  \node[main node] (8) [below of=7] {$p_2$ \\ $\err =1$};
  \node[main node] (9) [below of=8] {$p_2$ \\ $\err =2$};
  \node[main node] (10) [right of=8] {$\mathit{open}$};
  \node[main node] (11) [below of=6] {$\mathit{alarm}$};

  \draw[->] ([xshift=-.8cm]1.west) --  (1.west);

   \path
     (1) edge (4)
     (4) edge (7)
     (2) edge (5)
     (5) edge (8)
     (3) edge (6)
     (6) edge (9)
     (7) edge (10)
     (8) edge (10)
     (9) edge (10)
     (3) edge (11)
     (6) edge (11)
     (9) edge (11)

     (1) edge (2)
     (4) edge (2)
     (7) edge (2)

     (2) edge (3)
     (5) edge (3)
     (8) edge (3)

     (10) edge [bend right=70] node {} (1)
     
     (10) edge [loop right] (10)
     (11) edge [loop right] (11);
\end{tikzpicture}
\caption{The Automatic Door Opener (ADO) System.}
\label{fig:ado_system}
\end{figure}

\begin{defi}
\label{def:kmts}
A {\em Kripke Modal Transition System} (KMTS) is a 5-tuple 
$\hat{M} = (\hat{S}, \hat{S_{0}},$ $R_{must}, R_{may}, \hat{L})$, where:

\begin{enumerate}
\item $\hat{S}$ is a finite set of \emph{states}.
\item $\hat{S_{0}}\subseteq \hat{S}$ is the set of \emph{initial states}.
\item $R_{must} \subseteq \hat{S} \times \hat{S}$ and $R_{may} \subseteq
  \hat{S} \times \hat{S}$ are \emph{transition relations} such that
  $R_{must} \subseteq R_{may}$. 
\item $\hat{L}: \hat{S} \rightarrow 2^{Lit}$ is a state-labeling
  such that $\forall \hat{s} \in \hat{S}$, $\forall p \in AP$, $\hat{s}$
  is labeled by {\em at most} one of $p$ and $\neg p$.\qed
  \end{enumerate}
\end{defi}

\noindent A KMTS has two types of transitions: \emph{must-transitions}, which
exhibit \emph{necessary} behavior, and \emph{may-transitions}, which
exhibit \emph{possible} behavior. Must-transitions are also may-transitions.  
The ``at most one'' condition in the
fourth part of Def.~\ref{def:kmts} makes it possible for the truth
value of an atomic proposition at a given state to be {\em unknown}.
This relaxation of truth values in conjunction with the existence of 
may-transitions in a KMTS constitutes a \emph{partial modeling}
formalism.

Verifying a CTL formula $\phi$ over a KMTS may result in an undefined
outcome ($\bot$).  We use the \emph{3-valued semantics}~\cite{HJS01} of a
CTL formula $\phi$ at a state $\hat{s}$ of KMTS $\hat{M}$.  

\begin{defi}
\label{def:ctl3_semantics}
{\bf \cite{HJS01}} Let $\hat{M} = (\hat{S}, \hat{S_{0}}, R_{must}, R_{may}, 
\hat{L})$ 
be a KMTS.  The 3-valued semantics of a CTL formula 
$\phi$ at a state $\hat{s}$ of $\hat{M}$, denoted as 
$(\hat{M},\hat{s}) \models^{3} \phi$, is defined inductively as follows: 

\begin{itemize}
\item If $\phi = \mathit{false}$
\begin{itemize}
\item $[(\hat{M},\hat{s}) \models^{3} \phi ] = \mathit{false}$
\end{itemize}
\item If $\phi = \mathit{true}$
\begin{itemize}
\item $[(\hat{M},\hat{s}) \models^{3} \phi ] = \mathit{true}$
\end{itemize}
\item If $\phi = p$ where $p \in AP$
\begin{itemize}
\item $[(\hat{M},\hat{s}) \models^{3} \phi ] = \mathit{true}$, 
iff $p \in \hat{L}(\hat{s})$.  
\item $[(\hat{M},\hat{s}) \models^{3} \phi ] = \mathit{false}$, 
iff $\neg p \in \hat{L}(\hat{s})$.
\item $[(\hat{M},\hat{s}) \models^{3} \phi ] = \bot$, 
otherwise.  
\end{itemize}
\item If $\phi = \neg\phi_{1}$
\begin{itemize}
\item $[(\hat{M},\hat{s}) \models^{3} \phi ] = \mathit{true}$, 
iff $[(\hat{M},\hat{s}) \models^{3} \phi_{1} ] = \mathit{false}$.  
\item $[(\hat{M},\hat{s}) \models^{3} \phi ] = \mathit{false}$, 
iff $[(\hat{M},\hat{s}) \models^{3} \phi_{1} ] = \mathit{true}$.
\item $[(\hat{M},\hat{s}) \models^{3} \phi ] = \bot$, 
otherwise.  
\end{itemize}
\item If $\phi = \phi_{1} \, \vee \, \phi_{2}$
\begin{itemize}
\item $[(\hat{M},\hat{s}) \models^{3} \phi ] = \mathit{true}$, 
iff $[(\hat{M},\hat{s}) \models^{3} \phi_{1}] = \mathit{true}$ or 
$[(\hat{M},\hat{s}) \models^{3} \phi_{2}] = \mathit{true}$.    
\item $[(\hat{M},\hat{s}) \models^{3} \phi ] = \mathit{false}$,
iff $[(\hat{M},\hat{s}) \models^{3} \phi_{1}] = \mathit{false}$ and 
$[(\hat{M},\hat{s}) \models^{3} \phi_{2}] = \mathit{false}$.  
\item $[(\hat{M},\hat{s}) \models^{3} \phi ] = \bot$, 
otherwise.  
\end{itemize}
\item If $\phi = \phi_{1} \, \wedge \, \phi_{2}$
\begin{itemize}
\item $[(\hat{M},\hat{s}) \models^{3} \phi ] = \mathit{true}$, 
iff $[(\hat{M},\hat{s}) \models^{3} \phi_{1}] = \mathit{true}$ and 
$[(\hat{M},\hat{s}) \models^{3} \phi_{2}] = \mathit{true}$.    
\item $[(\hat{M},\hat{s}) \models^{3} \phi ] = \mathit{false}$,
iff $[(\hat{M},\hat{s}) \models^{3} \phi_{1}] = \mathit{false}$ or 
$[(\hat{M},\hat{s}) \models^{3} \phi_{2}] = \mathit{false}$.  
\item $[(\hat{M},\hat{s}) \models^{3} \phi ] = \bot$, 
otherwise.  
\end{itemize}
\item If $\phi = AX\phi_{1}$
\begin{itemize}
\item $[(\hat{M},\hat{s}) \models^{3} \phi ] = \mathit{true}$, 
iff for all $\hat{s}_{i}$ such that 
$(\hat{s},\hat{s}_{i}) \in R_{may}$, 
$[(\hat{M},\hat{s}_{i}) \models^{3} \phi_{1}] = \mathit{true}$.      
\item $[(\hat{M},\hat{s}) \models^{3} \phi ] = \mathit{false}$,
iff there exists some $\hat{s}_{i}$ such that 
$(\hat{s},\hat{s}_{i}) \in R_{must}$ and  
$[(\hat{M},\hat{s}_{i}) \models^{3} \phi_{1}] = \mathit{false}$.  
\item $[(\hat{M},\hat{s}) \models^{3} \phi ] = \bot$,
otherwise.  
\end{itemize} 
\item If $\phi = EX\phi_{1}$
\begin{itemize}
\item $[(\hat{M},\hat{s}) \models^{3} \phi ] = \mathit{true}$, 
iff there exists $\hat{s}_{i}$ such that 
$(\hat{s},\hat{s}_{i}) \in R_{must}$ and  
$[(\hat{M},\hat{s}_{i}) \models^{3} \phi_{1}] = \mathit{true}$.      
\item $[(\hat{M},\hat{s}) \models^{3} \phi ] = \mathit{false}$,
iff for all $\hat{s}_{i}$ such that 
$(\hat{s},\hat{s}_{i}) \in R_{may}$, 
$[(\hat{M},\hat{s}_{i}) \models^{3} \phi_{1}] = \mathit{false}$.  
\item $[(\hat{M},\hat{s}) \models^{3} \phi ] = \bot$,
otherwise.  
\end{itemize} 
\item If $\phi = AG\phi_{1}$
\begin{itemize}
\item $[(\hat{M},\hat{s}) \models^{3} \phi ] = \mathit{true}$, 
iff for all may-paths 
$\pi_{may} = [\hat{s},\hat{s}_{1},\hat{s}_{2},...]$ and for all 
$\hat{s}_{i} \in \pi_{may}$ it holds that 
$[(\hat{M},\hat{s}_{i}) \models^{3} \phi_{1}] = \mathit{true}$.  
\item $[(\hat{M},\hat{s}) \models^{3} \phi ] = \mathit{false}$,
iff there exists some must-path  
$\pi_{must} = [\hat{s},\hat{s}_{1},\hat{s}_{2},...]$, 
such that for some $\hat{s}_{i} \in \pi_{must}$,  
$[(\hat{M},\hat{s}_{i}) \models^{3} \phi_{1}] = \mathit{false}$.  
\item $[(\hat{M},\hat{s}) \models^{3} \phi ] = \bot$,
otherwise.  
\end{itemize}
\item If $\phi = EG\phi_{1}$
\begin{itemize}
\item $[(\hat{M},\hat{s}) \models^{3} \phi ] = \mathit{true}$, 
iff there exists some must-path 
$\pi_{must} = [\hat{s},\hat{s}_{1},\hat{s}_{2},...]$, 
such that for all $\hat{s}_{i} \in \pi_{must}$,  
$[(\hat{M},\hat{s}_{i}) \models^{3} \phi_{1}] = \mathit{true}$.  
\item $[(\hat{M},\hat{s}) \models^{3} \phi ] = \mathit{false}$,
iff for all may-paths   
$\pi_{may} = [\hat{s},\hat{s}_{1},\hat{s}_{2},...]$, there is some 
$\hat{s}_{i} \in \pi_{may}$ such that   
$[(\hat{M},\hat{s}_{i}) \models^{3} \phi_{1}] = \mathit{false}$.  
\item $[(\hat{M},\hat{s}) \models^{3} \phi ] = \bot$,
otherwise.  
\end{itemize}
\item If $\phi = AF\phi_{1}$
\begin{itemize}
\item $[(\hat{M},\hat{s}) \models^{3} \phi ] = \mathit{true}$, 
iff for all may-paths 
$\pi_{may} = [\hat{s},\hat{s}_{1},\hat{s}_{2},...]$, there is a 
$\hat{s}_{i} \in \pi_{may}$ such that  
$[(\hat{M},\hat{s}_{i}) \models^{3} \phi_{1}] = \mathit{true}$.  
\item $[(\hat{M},\hat{s}) \models^{3} \phi ] = \mathit{false}$,
iff there exists some must-path  
$\pi_{must} = [\hat{s},\hat{s}_{1},\hat{s}_{2},...]$, 
such that for all $\hat{s}_{i} \in \pi_{must}$,  
$[(\hat{M},\hat{s}_{i}) \models^{3} \phi_{1}] = \mathit{false}$.  
\item $[(\hat{M},\hat{s}) \models^{3} \phi ] = \bot$,
otherwise.  
\end{itemize}
\item If $\phi = EF\phi_{1}$
\begin{itemize}
\item $[(\hat{M},\hat{s}) \models^{3} \phi ] = \mathit{true}$, 
iff there exists some must-path  
$\pi_{must} = [\hat{s},\hat{s}_{1},\hat{s}_{2},...]$, such 
that there is some $\hat{s}_{i} \in \pi_{must}$ for which   
$[(\hat{M},\hat{s}_{i}) \models^{3} \phi_{1}] = \mathit{true}$.  
\item $[(\hat{M},\hat{s}) \models^{3} \phi ] = \mathit{false}$,
iff for all may-paths   
$\pi_{may} = [\hat{s},\hat{s}_{1},\hat{s}_{2},...]$ and for all 
$\hat{s}_{i} \in \pi_{may}$, 
$[(\hat{M},\hat{s}_{i}) \models^{3} \phi_{1}] = \mathit{false}$.  
\item $[(\hat{M},\hat{s}) \models^{3} \phi ] = \bot$,
otherwise.  
\end{itemize}
\item If $\phi = A(\phi_{1} \, U \, \phi_{2})$
\begin{itemize}
\item $[(\hat{M},\hat{s}) \models^{3} \phi ] = \mathit{true}$, 
iff for all may-paths 
$\pi_{may} = [\hat{s},\hat{s}_{1},\hat{s}_{2},...]$, there is 
$\hat{s}_{i} \in \pi_{may}$ such that 
$[(\hat{M},\hat{s}_{i}) \models^{3} \phi_{2}] = \mathit{true}$
and $\forall j < i: [(\hat{M},\hat{s}_{j}) \models^{3} \phi_{1}] 
= true$.  
\item $[(\hat{M},\hat{s}) \models^{3} \phi ] = \mathit{false}$,
iff there exists some must-path  
$\pi_{must} = [\hat{s},\hat{s}_{1},\hat{s}_{2},...]$, 
such that 
\begin{itemize}
\item[i.] for all $0\leq k< |\pi_{must}|:$\\
$(\forall j < k : [(\hat{M},\hat{s}_{j}) \models^{3} \phi_{1}] \neq \mathit{false}) 
\Rightarrow ([(\hat{M},\hat{s}_{k}) \models^{3} \phi_{2}] = \mathit{false})$
\item[ii.] 
$(\text{for all } 0 \leq k < |\pi_{must}|:[(\hat{M},\hat{s}_{k}) \models^{3} \phi_{2}] \neq \mathit{false}) 
\Rightarrow |\pi_{must}| = \infty$
\end{itemize}  
\item $[(\hat{M},\hat{s}) \models^{3} \phi ] = \bot$,
otherwise.  
\end{itemize}
\item If $\phi = E(\phi_{1}U\phi_{2})$
\begin{itemize}
\item $[(\hat{M},\hat{s}) \models^{3} \phi ] = \mathit{true}$, 
iff there exists some must-path 
$\pi_{must} = [\hat{s},\hat{s}_{1},\hat{s}_{2},...]$
such that there is a $\hat{s}_{i} \in \pi_{must}$ with   
$[(\hat{M},\hat{s}_{i}) \models^{3} \phi_{2}] = \mathit{true}$ 
and for all 
$j < i, [(\hat{M},\hat{s}_{j}) \models^{3} 
\phi_{1}] = \mathit{true}$.  
\item $[(\hat{M},\hat{s}) \models^{3} \phi ] = \mathit{false}$,
iff for all may-paths   
$\pi_{may} = [\hat{s},\hat{s}_{1},\hat{s}_{2},...]$ 
\begin{itemize}
\item[i.] for all $0 \leq k < |\pi_{may}|:$\\ 
$(\forall j < k : 
[(\hat{M},\hat{s}_{j}) \models^{3} \phi_{1}] \neq \mathit{false}) 
\Rightarrow ([(\hat{M},\hat{s}_{k}) \models^{3} \phi_{2}] = \mathit{false})$
\item[ii.] $(\text{for all } 0 \leq k < |\pi_{may}| : 
[(\hat{M},\hat{s}_{k}) \models^{3} \phi_{2}] \neq \mathit{false}) 
\Rightarrow |\pi_{may}| = \infty$
\end{itemize}  
\item $[(\hat{M},\hat{s}) \models^{3} \phi ] = \bot$,
otherwise.  \qed
\end{itemize}
\end{itemize}  
\end{defi}\enlargethispage{\baselineskip}

\noindent From the 3-valued CTL semantics, it follows that must-transitions are used to 
check the truth of existential CTL properties, while may-transitions
are used to check the truth of universal CTL properties.  This works inversely 
for checking the refutation of CTL properties.  In what follows, we use 
$\models$ instead of $\models^{3}$ in order to refer to the 3-valued 
satisfaction relation.

\section{Abstraction and Refinement for 3-Valued CTL}
\label{sec:abstr}

\subsection{Abstraction}

\emph{Abstraction} is a state-space reduction technique that produces a
smaller abstract model from an initial {\em concrete} model, so that the
result of model checking a property $\phi$ in the abstract model is preserved
in the concrete model.  This can be achieved if the abstract model is built 
with certain requirements~\cite{CGL94,GHJ01}.

\begin{defi}
\label{def:abs_kmts}
Given a KS $M = (S, S_{0}, R, L)$ and a pair of total functions 
$(\alpha : S \rightarrow \hat{S}, \gamma : \hat{S}
\rightarrow 2^{S})$ such that $$\forall s \in S:
\forall\hat{s} \in \hat{S}: (\alpha(s) = \hat{s} \Leftrightarrow s \in 
\gamma(\hat{s}))$$ the KMTS 
$\alpha(M) = (\hat{S}, \hat{S_{0}}, R_{must}, R_{may}, 
\hat{L})$ is defined as follows:
\begin{enumerate}
\item $\hat{s} \in \hat{S_{0}}$ iff $\exists s \in 
\gamma(\hat{s})$ such that $s \in S_{0}$

\item $lit \in \hat{L}(\hat{s})$ only if $\forall s \in 
\gamma(\hat{s}): lit \in L(s)$

\item $R_{must} = \left\{(\hat{s_{1}},\hat{s_{2}}) \mid 
\forall s_{1} \in \gamma(\hat{s_{1}}): \exists s_{2} \in 
\gamma(\hat{s_{2}}): (s_{1},s_{2}) \in R\right\}$

\item $R_{may} = \left\{(\hat{s_{1}},\hat{s_{2}}) \mid 
\exists s_{1} \in \gamma(\hat{s_{1}}): \exists s_{2} \in 
\gamma(\hat{s_{2}}): (s_{1},s_{2}) \in R\right\}$\qed
\end{enumerate}
\end{defi}

For a given KS $M$ and a pair of abstraction and concretization functions 
$\alpha$ and $\gamma$, Def.~\ref{def:abs_kmts} introduces the KMTS $\alpha(M)$ 
defined over the set $\hat{S}$ of \emph{abstract states}.  In our AMR 
framework, we view $M$ as the \emph{concrete model} and the KMTS $\alpha(M)$ as 
the \emph{abstract model}.  Any two concrete states $s_{1}$ and $s_{2}$ 
of $M$ are abstracted by $\alpha$ to a state $\hat{s}$ of $\alpha(M)$ 
if and only if $s_{1}$, $s_{2}$ are elements of the set $\gamma(\hat{s})$ (see 
Fig~\ref{fig:abstract_concrete}).  A state of $\alpha(M)$ is initial \emph{if 
and only if} at least one of its concrete states is initial as well.  An atomic 
proposition in an abstract state is true (respectively, false), \emph{only if} 
it is also true (respectively, false) in all of its concrete states. This means 
that the value of an atomic proposition may be unknown at a state of 
$\alpha(M)$. A must-transition from $\hat{s_{1}}$ to $\hat{s_{2}}$ of 
$\alpha(M)$ exists, 
if and only if there are transitions from all states of $\gamma(\hat{s_{1}})$ 
to at least one state of $\gamma(\hat{s_{2}})$ $(\forall\exists-condition)$.  
Respectively, a may-transition from $\hat{s_{1}}$ to $\hat{s_{2}}$ of 
$\alpha(M)$ exists, 
if and only if there is at least one transition from some state of 
$\gamma(\hat{s_{1}})$ 
to some state of $\gamma(\hat{s_{2}})$ $(\exists\exists-condition)$.  

\begin{figure}[t]
\centering
\includegraphics[width=38mm]{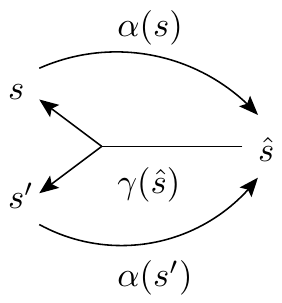} 
\caption{Abstraction and Concretization.}
\label{fig:abstract_concrete}
\end{figure}

\begin{defi}\enlargethispage{\baselineskip}
\label{def:concretize_kmts}
Given a pair of total functions $(\alpha : S \rightarrow \hat{S}, \gamma : 
\hat{S} \rightarrow 2^{S})$ such that $$\forall s \in S: \forall
\hat{s} \in \hat{S}: (\alpha(s) = \hat{s} \Leftrightarrow s \in 
\gamma(\hat{s}))$$ and a KMTS 
$\hat{M} = (\hat{S}, \hat{S_{0}}, R_{must}, R_{may}, \hat{L})$, 
the set of KSs  $\gamma(\hat{M}) = \{M \mid M = (S, S_{0}, R, L)\}$ 
is defined such that for all $M \in \gamma(\hat{M})$ the following conditions 
hold:  
\begin{enumerate}
\item $s \in S_{0}$ iff $\alpha(s) \in \hat{S_{0}}$
\item $lit \in L(s)$ if $lit \in \hat{L}(\alpha(s))$
\item $(s_{1},s_{2}) \in R$ iff
\begin{itemize}{}
\item $\exists s_{1}^{\prime} \in \gamma(\alpha(s_{1})): 
\exists s_{2}^{\prime} \in \gamma(\alpha(s_{2})) : 
(\alpha(s_{1}),\alpha(s_{2})) \in R_{may}$ and,
\item $\forall s_{1}^{\prime} \in \gamma(\alpha(s_{1})): 
\exists s_{2}^{\prime} \in \gamma(\alpha(s_{2})) : 
(\alpha(s_{1}),\alpha(s_{2})) \in R_{must}$ \qed
\end{itemize}{}
\end{enumerate}
\end{defi}

\noindent For a given KMTS $\hat{M}$ and a pair of abstraction and concretization 
functions $\alpha$ and $\gamma$, Def.~\ref{def:concretize_kmts} introduces 
a set $\gamma(\hat{M})$ of \emph{concrete} KSs.  A state $s$ of a KS 
$M \in \gamma(\hat{M})$ is initial if its abstract state $\alpha(s)$ is 
also initial.  An atomic proposition in a concrete state $s$ is true 
(respectively, false) if it is also true (respectively, false) in its abstract 
state $\alpha(s)$. A transition from a concrete state $s_{1}$ to another 
concrete state $s_{2}$ exists, if and only if 
\begin{itemize}
\item{}
there are concrete states 
$s_{1}^{\prime} \in \gamma(\alpha(s_{1}))$ and
$s_{2}^{\prime} \in \gamma(\alpha(s_{2}))$, where  
$(\alpha(s_{1}),\alpha(s_{2})) \in R_{may}$, and 

\item{} 
there is at least one concrete state 
$s_{2}^{\prime} \in \gamma(\alpha(s_{2}))$ such that for all 
$s_{1}^{\prime} \in \gamma(\alpha(s_{1}))$ it holds that 
$(\alpha(s_{1}),\alpha(s_{2})) \in R_{must}$.  
\end{itemize}

\paragraph{Abstract Interpretation.} 
A pair of abstraction and concretization functions can 
be defined within an \emph{Abstract Interpretation}~\cite{CC77,CC79} framework. 
Abstract interpretation is a theory for a set of abstraction techniques, for 
which important 
properties for the model checking problem have been proved~\cite{DGG97,D96}.    

\begin{defi}
\label{def:mixsimul}
\emph{~\cite{DGG97,GJ02}}
Let $M = (S, S_{0}, R, L)$ be a concrete KS and 
$\hat{M}$ = $(\hat{S}, \hat{S_{0}}, R_{must},$  $R_{may}, \hat{L})$ 
be an abstract KMTS.  A relation $H \subseteq S \times \hat{S}$ 
for $M$ and $\hat{M}$ is called a \emph{mixed simulation},
when $H(s,\hat{s})$ implies:  
\begin{itemize}
\item $\hat{L}(\hat{s}) \subseteq L(s)$

\item if $r = (s,s^{\prime}) \in R$, then there is exists $\hat{s}^{\prime} \in 
\hat{S}$ such that $r_{may} = (\hat{s},\hat{s}^{\prime}) \in R_{may}$ 
and $(s^{\prime},\hat{s}^{\prime}) \in H$.

\item if $r_{must} = (\hat{s},\hat{s}^{\prime}) \in R_{must}$, 
then there exists $s^{\prime} \in S$ such that 
$r = (s,s^{\prime}) \in R$ and $(s^{\prime},\hat{s}^{\prime}) 
\in H$.\qed
\end{itemize}
\end{defi}

\noindent The abstraction function $\alpha$ of Def.~\ref{def:abs_kmts} is a 
mixed simulation for the KS $M$ and its abstract KMTS $\alpha(M)$.

\begin{thm}
\label{theor:preserv}
\emph{\cite{GJ02}}
Let $H \subseteq S \times \hat{S}$ be a mixed simulation 
from a KS $M = (S, S_{0}, R, L)$ to a KMTS $\hat{M} = 
(\hat{S}, \hat{S_{0}}, R_{must}, R_{may}, \hat{L})$.  
Then, for every CTL formula $\phi$ and every $(s,\hat{s}) 
\in H$ it holds that
\[
[(\hat{M},\hat{s}) \models \phi] \neq \bot \Rightarrow 
[(M,s) \models \phi] = [(\hat{M},\hat{s}) \models \phi] 
\]
\end{thm}
\noindent Theorem~\ref{theor:preserv} ensures that if a CTL formula 
$\phi$ has a definite truth value (i.e., true or false) in 
the abstract KMTS, then it has the same truth value in the concrete KS.  When 
we get $\bot$ from the 3-valued model checking of a CTL formula $\phi$, the 
result of model checking property $\phi$ on the corresponding KS can be 
either true or false.\\

\noindent \emph{Example.} An abstract KMTS $\hat{M}$ is presented in
Fig.~\ref{fig:ado_initial}, where all the states labeled by $q$
are grouped together, as are all states labeled by $\neg q$.

\begin{figure}[t]
\centering
\subfloat[The KS and initial KMTS.]
{\label{fig:ado_initial}\begin{tikzpicture}[->,>=stealth',auto,node 
distance=2cm, scale=0.75, thick, main node/.style={scale=0.75, minimum size = 
1cm, align=center,circle,fill=blue!10,draw}, abs node/.style={scale=0.75, 
minimum size = 1cm, align=center,rectangle,fill=blue!10,draw}]

\begin{scope}
  \node[main node] (1) {$s_0$ \\ $\neg q$};
  \node[main node] (2) [below of=1] {$s_3$ \\ $\neg q$};
  \node[main node] (3) [below of=2] {$s_6$ \\ $\neg q$};
  \node[main node] (4) [right of=1] {$s_1$ \\ $\neg q$};
  \node[main node] (5) [below of=4] {$s_4$ \\ $\neg q$};
  \node[main node] (6) [below of=5] {$s_7$ \\ $\neg q$};
  \node[main node] (7) [right of=4] {$s_2$ \\ $\neg q$};
  \node[main node] (8) [below of=7] {$s_5$ \\ $\neg q$};
  \node[main node] (9) [below of=8] {$s_8$ \\ $\neg q$};
  \node[main node] (10) [right of=8] {$s_{10}$ \\ $q$};
  \node[main node] (11) [below of=6] {$s_9$ \\ $\neg q$};

   \path
     (1) edge (4)
     (4) edge (7)
     (2) edge (5)
     (5) edge (8)
     (3) edge (6)
     (6) edge (9)
     (7) edge (10)
     (8) edge (10)
     (9) edge (10)
     (3) edge (11)
     (6) edge (11)
     (9) edge (11)

     (1) edge (2)
     (4) edge (2)
     (7) edge (2)

     (2) edge (3)
     (5) edge (3)
     (8) edge (3)

     (10) edge [bend right=70] node {} (1)
     
     (10) edge [loop right] (10)
     (11) edge [loop below] (11);
     \draw[->] ([xshift=-.8cm]1.west) --  (1.west);
     
     \draw[dotted] (-1, 1) rectangle (5, -7);
     \draw[dotted] (5.2, -1.2) rectangle (6.7, -2.8);
     \node[font = \small] at (3, -7.5) {$M$};
\end{scope}

\begin{scope}[xshift=9cm]
  \node[abs node] (12) {$\hat{s}_0$ \\ $\neg q$};
  \node[abs node] (13) [right of=12] {$\hat{s}_1$ \\ $q$};

     \path
     (12) edge [loop above] (12)
     (13) edge [loop above] (13)
     (13) edge [bend left=15] (12)
     (12) edge [bend left=15, dashed] (13);

     \draw[->] ([xshift=-.8cm]12.west) --  (12.west);

     \node[font = \small] at (1, -1.5) {$\alpha(M$)};

     \draw (-.5, -3) rectangle (4.2, -4.5);
     \node[font = \footnotesize] at (2.5, -3.5) {must-transition};
     \node[font = \footnotesize] at (2.5, -4) {may-transition};
     \draw[->] (0, -3.5) -- (.8, -3.5);
     \draw[->, dashed] (0, -4) -- (.8, -4);
  \end{scope}

\end{tikzpicture}}          

\subfloat[The KS and refined KMTS.]
{\label{fig:ado_refined}\begin{tikzpicture}[->,>=stealth',auto,node 
distance=2cm, scale=0.75, thick, main node/.style={scale=0.75, minimum size = 
1cm, align=center,circle,fill=blue!10,draw}, abs node/.style={scale=0.75, 
minimum size = 1cm, align=center,rectangle,fill=blue!10,draw}]

\begin{scope}
  \node[main node] (1) {$s_0$ \\ $\neg q$};
  \node[main node] (2) [below of=1] {$s_3$ \\ $\neg q$};
  \node[main node] (3) [below of=2] {$s_6$ \\ $\neg q$};
  \node[main node] (4) [right of=1] {$s_1$ \\ $\neg q$};
  \node[main node] (5) [below of=4] {$s_4$ \\ $\neg q$};
  \node[main node] (6) [below of=5] {$s_7$ \\ $\neg q$};
  \node[main node] (7) [right of=4] {$s_2$ \\ $\neg q$};
  \node[main node] (8) [below of=7] {$s_5$ \\ $\neg q$};
  \node[main node] (9) [below of=8] {$s_8$ \\ $\neg q$};
  \node[main node] (10) [right of=8] {$s_{10}$ \\ $q$};
  \node[main node] (11) [below of=6] {$s_9$ \\ $\neg q$};

   \path
     (1) edge (4)
     (4) edge (7)
     (2) edge (5)
     (5) edge (8)
     (3) edge (6)
     (6) edge (9)
     (7) edge (10)
     (8) edge (10)
     (9) edge (10)
     (3) edge (11)
     (6) edge (11)
     (9) edge (11)

     (1) edge (2)
     (4) edge (2)
     (7) edge (2)

     (2) edge (3)
     (5) edge (3)
     (8) edge (3)

     (10) edge [bend right=70] node {} (1)
     
     (10) edge [loop right] (10)
     (11) edge [loop below] (11);
     \draw[->] ([xshift=-.8cm]1.west) --  (1.west);
     
     \draw[dotted] (-1, 1) rectangle (2.8, -7);
     \draw[dotted] (3, 1) rectangle (5, -5);
     \draw[dotted] (5.2, -1.2) rectangle (6.7, -2.8);
     \node[font = \small] at (3, -7.5) {$M$};
\end{scope}

\begin{scope}[xshift=9cm]
  \node[abs node] (12) {$\hat{s}_{01}$ \\ $\neg q$};
  \node[abs node] (13) [below of=12] {$\hat{s}_{02}$ \\ $\neg q$};
  \node[abs node, yshift=1cm] (14) [right of=13] {$\hat{s}_1$ \\ $q$};

     \path
     (12) edge [loop above] (12)
     (14) edge [loop above] (14)
     (12) edge [bend left=15, dashed] (13)
     (14) edge (12)
     (13) edge (14)
     (13) edge [bend left=15] (12);

     \draw[->] ([xshift=-.8cm]12.west) --  (12.west);

     \node[font = \small] at (3, -2.5) {$\alpha_{\mathit{Refined}}(M$)};

     \draw (-.5, -3) rectangle (4.2, -4.5);
     \node[font = \footnotesize] at (2.5, -3.5) {must-transition};
     \node[font = \footnotesize] at (2.5, -4) {may-transition};
     \draw[->] (0, -3.5) -- (.8, -3.5);
     \draw[->, dashed] (0, -4) -- (.8, -4);
  \end{scope}

\end{tikzpicture}}
\caption{The KS and KMTSs for the ADO system.}
\label{fig:ado_ks_kmts}
\end{figure}
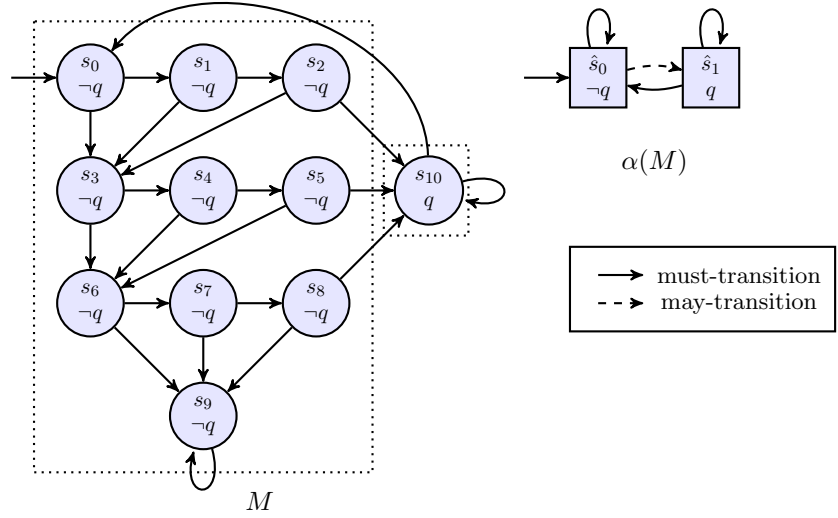
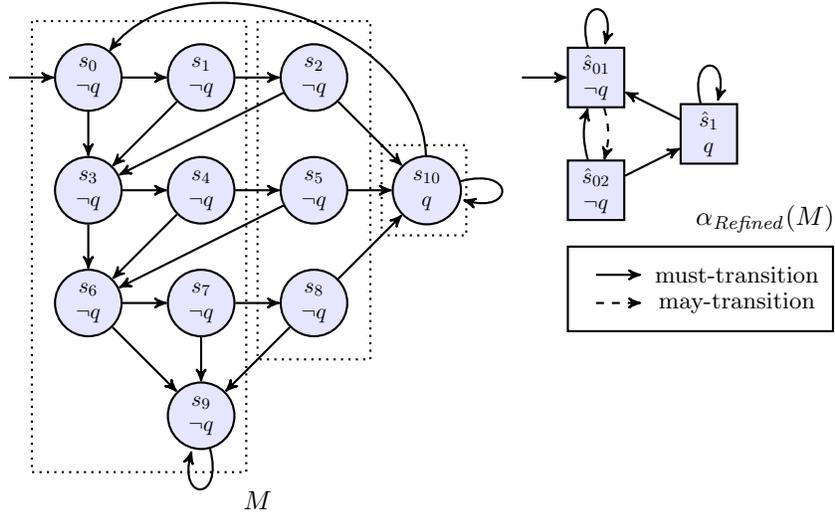

\subsection{Refinement}

When the outcome of verifying a CTL formula $\phi$ on an abstract
model using the 3-valued semantics is $\bot$, then a \emph{refinement}
step is needed to acquire a more \emph{precise} abstract model.  In the 
literature, there are refinement approaches for the 2-valued CTL 
semantics~\cite{CGJLV00,CPR05,CGR07}, as well as a number of techniques for the 
3-valued CTL model checking~\cite{GHJ01,SG04,SG07,GLLS07}.  The refinement 
technique that we adopt is an automated two-step process based 
on~\cite{CGJLV00,SG04}: 
\begin{enumerate}
\item Identify a \emph{failure state} in $\alpha(M)$ using the algorithms 
in~\cite{CGJLV00,SG04}; the cause of failure for a state $\hat{s}$ stems from an 
atomic proposition having an undefined value in $\hat{s}$, or from an outgoing 
may-transition from $\hat{s}$.   
\item Produce the abstract KMTS $\alpha_{\mathit{Refined}}(M)$, where 
$\alpha_{\mathit{Refined}}$ is a new abstraction function as in 
Def.~\ref{def:abs_kmts}, such that the identified failure state is refined into 
two states.  If the cause of failure is an undefined value of an atomic 
proposition in $\hat{s}$, then $\hat{s}$ is split into states $\hat{s}_{1}$ and 
$\hat{s}_{2}$, such that the atomic proposition is true in $\hat{s}_{1}$ and 
false in $\hat{s}_{2}$.  Otherwise, if the cause of failure is an outgoing 
may-transition from $\hat{s}$, then $\hat{s}$ is split into states $\hat{s}_{1}$ 
and $\hat{s}_{2}$, such that there is an outgoing must-transition from 
$\hat{s}_{1}$ and no outgoing may- or must-transition from $\hat{s}_{2}$.  
\end{enumerate}
The described refinement technique does not necessarily converge to an abstract 
KMTS with a definite model checking result.  A promising approach in order to 
overcome this restriction is by using a different type of abstract model, as 
in~\cite{SG04}, where the authors propose the use of Generalized KMTSs, which 
ensure monotonicity of the refinement process.\\

\noindent \emph{Example. } Consider the case where the ADO system requires a
mechanism for opening the door from any state with a direct action. This
could be an action done by an expert if an immediate opening of the door
is required.  This property can be expressed in CTL as $\phi = AGEXq$.  
Observe that in $\alpha(M)$ of Fig.~\ref{fig:ado_initial}, the absence of a 
must-transition from $\hat{s}_{0}$ to $\hat{s}_{1}$, where 
$[(\alpha(M),\hat{s}_{1}) \models q]
= true$, in conjunction with the existence of a may-transition from
$\hat{s}_{0}$ to $\hat{s}_{1}$, i.e. to a state where
$[(\alpha(M),\hat{s}_{1}) \models q] = true$, results in an undefined
model-checking outcome for $[(\alpha(M),\hat{s}_{0}) \models \phi]$.  
Notice that state $\hat{s}_{0}$ is the failure state, and the
may-transition from $\hat{s}_{0}$ to $\hat{s}_{1}$ is the cause of the
failure.  Consequently, $\hat{s}_{0}$ is refined into two states, 
$\hat{s}_{01}$ and $\hat{s}_{02}$, such that the former
has no transition to $\hat{s}_{1}$ and the latter has an 
outgoing must-transition to $\hat{s}_{1}$.  Thus, the may-transition which 
caused the undefined outcome is eliminated and for the refined KMTS 
$\alpha_{\mathit{Refined}}(M)$ it holds that 
$[\alpha_{\mathit{Refined}}(M),\hat{s}_{1}) \models \phi] = 
\mathit{false}$.  The initial KS and the refined KMTS 
$\alpha_{\mathit{Refined}}(M)$ are shown in 
Fig.~\ref{fig:ado_refined}.

\section{The Model Repair Problem}
\label{sec:mrp}

In this section, we formulate the problem of Model Repair.  A metric space over 
Kripke structures is defined to quantify their structural differences.  This 
allows us taking into account the \emph{minimality of changes} criterion
in Model Repair. 

\noindent Let $\pi$ be a function on the set of all functions $f: X \rightarrow Y$ such that:
\[
\pi(f) = \{(x, f(x)) \mid x \in X\}
\]

\noindent A \emph{restriction operator}  (denoted by $\upharpoonright$) for 
the domain of function $f$ is defined such that for $X_{1} \subseteq X$,
\[
f\upharpoonright_{X_{1}} = \{(x, f(x)) \mid x \in X_{1}\}
\] 
By $S^C$, we denote the complement of a set $S$. 

\begin{defi}
\label{def:metric_space}
For any two $M = (S,S_{0},R,L)$ and
$M^{\prime} = (S^{\prime},S^{\prime}_{0},R^{\prime},L^{\prime})$ in the set $K_{M}$
of all KSs, where
\begin{itemize}
\item[]
$S^{\prime} = (S \cup S_{\mathit{IN}}) - S_{\mathit{OUT}}$ 
for some $S_{\mathit{IN}} \subseteq S^{C}$, $S_{\mathit{OUT}} \subseteq S$,
\item[]
$R^{\prime} = (R \cup R_{\mathit{IN}}) - R_{\mathit{OUT}}$ 
for some $R_{\mathit{IN}} \subseteq R^{C}$, $R_{\mathit{OUT}} \subseteq R$,
\item[]
$L^{\prime} = S^{\prime} \rightarrow 2^{LIT}$, 
\end{itemize}
the {\em distance function} $d$ over $K_{M}$ is defined as follows:  
\[
d(M,M^{\prime}) = |S\,\Delta \, S^{\prime}| + |R \, \Delta \, R^{\prime}| + 
\frac{|\pi(L\upharpoonright_{S\cap S^{\prime}}) \,\Delta \, 
\pi(L^{\prime}\upharpoonright_{S\cap S^{\prime}})|}{2}
\]
with $A \, \Delta \, B$ representing the symmetric difference 
$(A-B)\cup(B-A)$.\qed
\end{defi}

\noindent For any two KSs defined over the same set of atomic propositions $AP$, 
function $d$ counts the number of differences $|S\,\Delta\, S^{\prime}|$ in the 
state spaces, the number of differences $|R\,\Delta\, R^{\prime}|$ in their 
transition relation and the number of common states with altered labeling. 

\begin{prop}
\label{prop:metric_space}
The ordered pair $(K_{M},d)$ is a metric space.
\end{prop}

\begin{proof}
We use the fact that the cardinality of the symmetric difference between any two sets is a distance metric.  
It holds that: 
\begin{enumerate}
\item $|S\Delta S^{\prime}| \geq 0$, $|R\Delta R^{\prime}| \geq 0$ and 
$|\pi(L\upharpoonright_{S\cap S^{\prime}})\Delta 
\pi(L^{\prime}\upharpoonright_{S\cap S^{\prime}})| \geq 0$ (non-negativity)
\item $|S\Delta S^{\prime}| = 0$ iff $S = S^{\prime}$, 
$|R\Delta R^{\prime}| = 0$ iff $R = R^{\prime}$ and $|\pi(L\upharpoonright_{S\cap S^{\prime}})|\Delta |\pi(L^{\prime}\upharpoonright_{S\cap S^{\prime}})| = 0$ iff $\pi(L\upharpoonright_{S\cap S^{\prime}}) = \pi(L^{\prime}\upharpoonright_{S\cap S^{\prime}})$ (identity of indiscernibles)
\item $|S\Delta S^{\prime}| = |S^{\prime}\Delta S|$, 
$|R\Delta R^{\prime}| = |R^{\prime}\Delta R|$ and 
$|\pi(L\upharpoonright_{S\cap S^{\prime}})\Delta \pi(L^{\prime}\upharpoonright_{S\cap S^{\prime}})| =\\ 
|\pi(L^{\prime}\upharpoonright_{S\cap S^{\prime}}) \Delta \pi(L\upharpoonright_{S\cap S^{\prime}})|$(symmetry)
\item $|S^{\prime}\Delta S^{\prime\prime}| \leq |S^{\prime}\Delta S| + |S \Delta S^{\prime\prime}|$,
$|R^{\prime}\Delta R^{\prime\prime}| \leq |R^{\prime}\Delta R| + |R \Delta R^{\prime\prime}|$, \\
$|\pi(L^{\prime}\upharpoonright_{S^{\prime}\cap S^{\prime\prime}})\Delta \pi(L^{\prime\prime}|_{S^{\prime}\cap S^{\prime\prime}})| \leq
|\pi(L^{\prime}\upharpoonright_{S^{\prime}\cap S})\Delta \pi(L\upharpoonright_{S^{\prime}\cap S})| + \\
|\pi(L\upharpoonright_{S\cap S^{\prime\prime}})\Delta \pi(L^{\prime\prime}|_{S\cap S^{\prime\prime}})|$ \\ (triangle inequality)
\end{enumerate}
We will prove that $d$ is a metric on $K_{M}$.  
Suppose $M, M^{\prime}, M^{\prime\prime} \in K_{M}$
\begin{itemize}
\item It easily follows from (1) that $d(M,M^{\prime}) \geq 0$ (non-negativity)
\item From (2), $d(M,M^{\prime}) = 0$ iff $M = M^{\prime}$ (identity of indiscernibles)
\item Adding the equations in (3), results in $d(M,M^{\prime}) = d(M^{\prime},M)$ (symmetry)
\item If we add the inequalities in (4), then we get 
$d(M^{\prime},M^{\prime\prime}) \leq d(M^{\prime},M) + d(M,M^{\prime\prime})$ 
(triangle inequality)
\end{itemize}
So, the proposition is true.   
\end{proof}

\begin{defi}
\label{def:metric_space_kmts}
For any two $\hat{M}$ = $(\hat{S}, \hat{S_{0}}, R_{must}, R_{may}, \hat{L})$ 
and 
$\hat{M}^{\prime}$ = $(\hat{S}^{\prime}, \hat{S_{0}}^{\prime}, R_{must}^{\prime},$  
$R_{may}^{\prime}, \hat{L}^{\prime})$ in the set $K_{\hat{M}}$ of all KMTSs, where 
\begin{itemize}
\item[]
$\hat{S}^{\prime} = (\hat{S} \cup \hat{S}_{\mathit{IN}}) - \hat{S}_{\mathit{OUT}}$ 
for some $\hat{S}_{\mathit{IN}} \subseteq \hat{S}^{C}$, $\hat{S}_{\mathit{OUT}} \subseteq \hat{S}$,
\item[]
$\hat{R}_{must}^{\prime} = (\hat{R}_{must} \cup \hat{R}_{\mathit{IN}}) - \hat{R}_{\mathit{OUT}}$ 
for some $\hat{R}_{\mathit{IN}} \subseteq \hat{R}_{must}^{C}$, $\hat{R}_{\mathit{OUT}} \subseteq \hat{R}_{must}$,
\item[]
$\hat{R}_{may}^{\prime} = (\hat{R}_{may} \cup \hat{R}_{\mathit{IN}}^{\prime}) - \hat{R}_{\mathit{OUT}}^{\prime}$ 
for some $\hat{R}_{\mathit{IN}}^{\prime} \subseteq \hat{R}_{may}^{C}$, $\hat{R}_{\mathit{OUT}}^{\prime} \subseteq \hat{R}_{may}$,
\item[]
$\hat{L}^{\prime} = \hat{S}^{\prime} \rightarrow 2^{LIT}$, 
\end{itemize}
the {\em distance function} $\hat{d}$ over $K_{\hat{M}}$ is defined as follows:   
\[
\begin{split}
\hat{d}(M,M^{\prime}) = |\hat{S} \, \Delta \, \hat{S}^{\prime}| + 
|\hat{R}_{must} \, \Delta \,  \hat{R}_{must}^{\prime}| + 
|(\hat{R}_{may} - \hat{R}_{must}) \, \Delta  \, (\hat{R}_{may}^{\prime} - 
\hat{R}_{must}^{\prime})| + \\ 
\frac{|\pi(\hat{L}\upharpoonright_{\hat{S}\cap \hat{S}^{\prime}}) \, \Delta  \, 
\pi(\hat{L}^{\prime}\upharpoonright_{\hat{S}\cap \hat{S}^{\prime}})|}{2}
\end{split}
\]
with $A \Delta B$ representing the symmetric difference $(A-B)\cup(B-A)$.
\end{defi}
\noindent
We note that $\hat{d}$ counts the differences between $\hat{R}_{may}^{\prime}$ and $\hat{R}_{may}$, and those between $\hat{R}_{must}^{\prime}$ and $\hat{R}_{must}$ separately, while avoiding to count the differences in the latter case twice (we remind that must-transitions are also included in $\hat{R}_{may}$). 

\begin{prop}
\label{prop:kmts_metric_space}
The ordered pair $(K_{\hat{M}},\hat{d})$ is a metric space.
\end{prop}

\begin{proof}
The proof is done in the same way as in Prop.~\ref{prop:metric_space}.  
\end{proof}

\begin{defi}
Given a KS $M$ and a CTL formula $\phi$ where $M \not\models \phi$, the 
Model Repair problem is to find a KS $M^{\prime}$, such that $M^{\prime} \models 
\phi$ and $d(M,M^{\prime})$ is minimum with respect to all such $M^{\prime}$.
\end{defi}

\noindent The Model Repair problem aims at modifying a KS such that the resulting KS satisfies a CTL formula that was violated before. The distance function $d$ of Def.~\ref{def:metric_space} features all the attractive properties of a distance metric. Given that no quantitative interpretation exists for predicates and logical operators in CTL, $d$ can be used in a model repair solution towards selecting minimum changes to the modified KS.

\section{The Abstract Model Repair Framework}
\label{sec:absmrp}

Our AMR framework integrates 3-valued
model checking, model refinement, and a new algorithm for selecting the
repair operations applied to the abstract model.  The goal
of this algorithm
is to apply the repair operations in a way, such that the number of
structural changes to the corresponding concrete model is
minimized.  The algorithm works based on a partial order relation 
over a set of basic repair operations for KMTSs.  This section describes
the steps involved in our AMR framework, the basic repair operations,
and the algorithm.

\subsection{The Abstract Model Repair Process}
\label{subsec:absmrp_approach}

\begin{figure}[t]
\centering
\includegraphics[width=120mm]{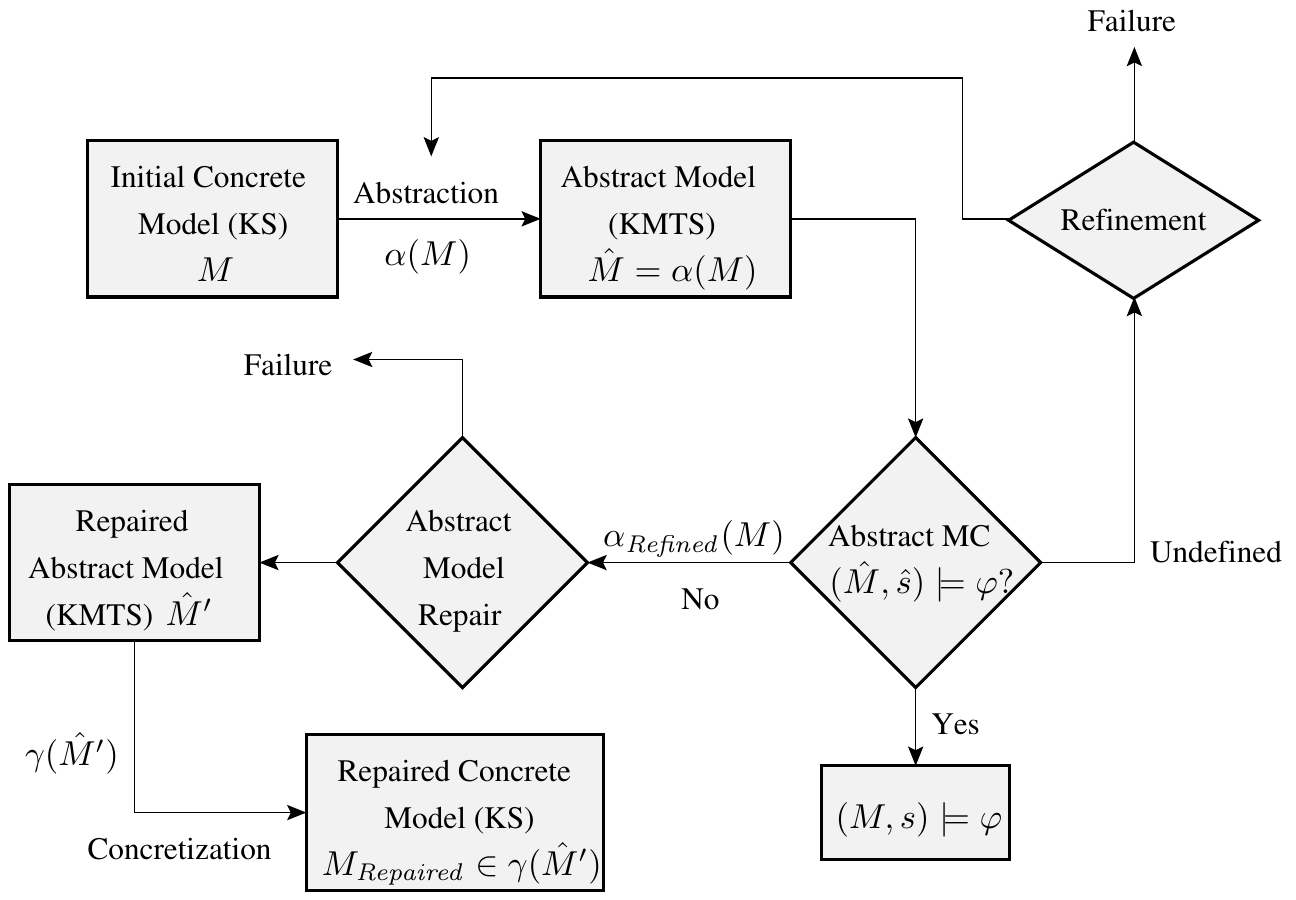} 
\caption{Abstract Model Repair Framework.}
\label{fig:abs_repair}
\end{figure} 

The process steps shown in Fig.~\ref{fig:abs_repair} rely 
on the KMTS abstraction of Def.~\ref{def:abs_kmts}.    
These are the following:
\begin{description}
\item[Step 1.] Given a KS $M$, a state $s$ of $M$, and 
a CTL property $\phi$, let us call $\hat{M}$
the KMTS obtained as in Def.~\ref{def:abs_kmts}.
\item[Step 2.] For state $\hat{s} = \alpha(s)$ of 
$\hat{M}$, we check whether 
$(\hat{M},\hat{s}) \models \phi$ by 3-valued model 
checking.  
\begin{description}
\item[Case 1.] If the result is \emph{true}, then,
according to Theorem~\ref{theor:preserv}, 
$(M,s) \models \phi$ and there is no need to 
repair $M$.     
\item[Case 2.] If the result is \emph{undefined}, 
then a refinement of $\hat{M}$ takes place, and:
\begin{description}
\item[Case 2.1.] If an $\hat{M}_{Refined}$ is found, 
the control is transferred to Step~2.  
\item[Case 2.2.] If a refined KMTS cannot be retrieved, the 
repair process terminates with a failure.   
\end{description}
\item[Case 3.] If the result is \emph{false},
then, from Theorem~\ref{theor:preserv},
$(M,s) \not\models \phi$ and the repair process
is enacted; the control is transferred to Step 3.
\end{description}
\item[Step 3.] 
The \emph{AbstractRepair} algorithm is called for 
the abstract KMTS ($\hat{M}_{Refined}$ or $\hat{M}$  
if no refinement has occurred), the state $\hat{s}$ and the 
property $\phi$.  
\begin{description}
\item[Case 1.] \emph{AbstractRepair} returns an 
$\hat{M}^{\prime}$ for which 
$(\hat{M}^{\prime},\hat{s}) \models \phi$.  
\item[Case 2.] \emph{AbstractRepair} fails to find an  
$\hat{M}^{\prime}$ for which the property holds true.   
\end{description}
\item[Step 4.] If \emph{AbstractRepair} returns an 
$\hat{M}^{\prime}$, then the process ends with selecting the
subset of KSs from $\gamma(\hat{M}^{\prime})$, with elements
whose distance $d$ from the KS $M$ is minimum 
with respect to all the KSs in $\gamma(\hat{M}^{\prime})$.
\end{description} 

\subsection{Basic Repair Operations}
\label{subsec:basic_ops}
We decompose the KMTS repair process into seven basic 
repair operations:  
\begin{description}
\item[AddMust] Adding a must-transition
\item[AddMay] Adding a may-transition
\item[RemoveMust] Removing a must-transition
\item[RemoveMay] Removing a may-transition
\item[ChangeLabel] Changing the labeling of a KMTS state
\item[AddState] Adding a new KMTS state
\item[RemoveState] Removing a disconnected KMTS state
\end{description}

\subsubsection{Adding a must-transition}

\begin{defi}[AddMust]
\label{def:AddMust}
For a given KMTS $\hat{M} = (\hat{S},\hat{S_{0}}, R_{must}, 
R_{may}, \hat{L})$ and 
$\hat{r}_{n} = (\hat{s}_{1},\hat{s}_{2}) \notin R_{must}$,  
$AddMust(\hat{M},\hat{r}_{n})$ is the KMTS 
$\hat{M^{\prime}} = (\hat{S}, \hat{S_{0}}, 
R_{must}^{\prime}, R_{may}^{\prime}, \hat{L})$
such that $R_{must}^{\prime} = R_{must} \cup \{\hat{r}_{n}\}$ and  
$R_{may}^{\prime} = R_{may} \cup \{\hat{r}_{n}\}$. \qed
\end{defi}

Since $R_{must} \subseteq R_{may}$, $\hat{r}_{n}$ must also 
be added to $R_{may}$, resulting in a new may-transition if 
$\hat{r}_{n} \notin R_{may}$.  Fig.~\ref{fig:AddMust} shows 
how the basic repair operation \emph{AddMust} modifies a 
given KMTS.  The newly added transitions are in bold.    

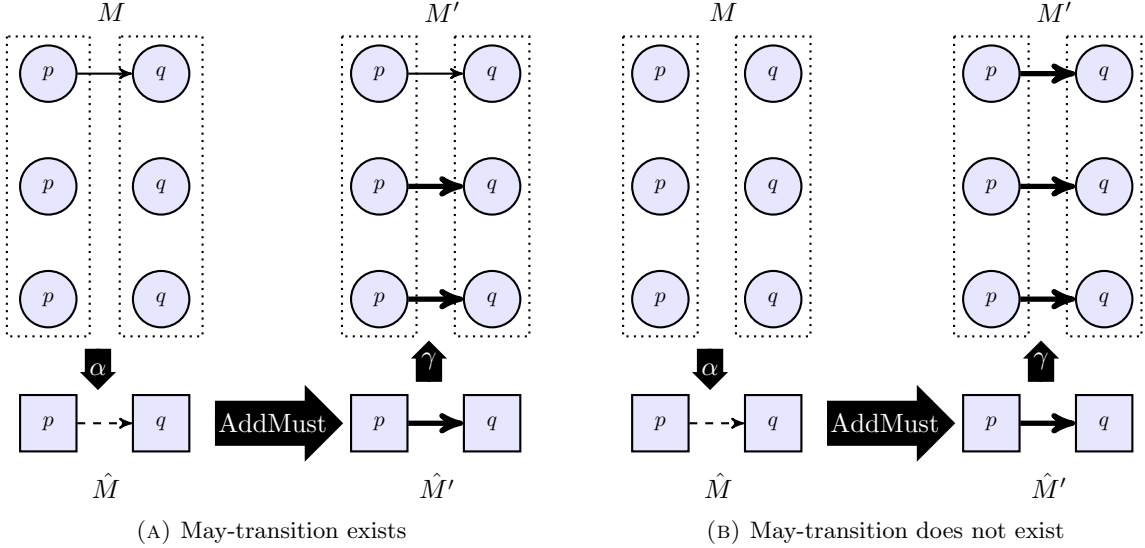
\begin{figure}[htb]
\centering
\subfloat[May-transition exists]{\begin{tikzpicture}[->,>=stealth',auto,node 
distance=2cm, scale=0.55, thick, main node/.style={scale=0.75, minimum size = 
1cm, align=center,circle,fill=blue!10,draw}, abs node/.style={scale=0.75, 
minimum size = 1cm, align=center,rectangle,fill=blue!10,draw}]

    \tikzstyle{bigarrows}=[line width=1.5mm,draw=black,-triangle 
90,postaction={draw, line width=6mm, shorten >=3mm, -}]

    \tikzstyle{smallarrows}=[line width=.5mm,draw=black,-triangle 
90,postaction={draw, line width=3.5mm, shorten >=2mm, -}]

\begin{scope}
  \node[main node] (1) {$p$};
  \node[main node] (2) [below of=1] {$p$};
  \node[main node] (3) [below of=2] {$p$};
  \node[main node] (4) [right of=1] {$q$};
  \node[main node] (5) [below of=4] {$q$};
  \node[main node] (6) [below of=5] {$q$};

   \path
     (1) edge (4);
     
     \draw[dotted] ([xshift=-3mm, yshift=9mm]1.west) rectangle ([xshift=3mm, 
yshift= -9mm]3.east);
     \draw[dotted] ([xshift=-3mm, yshift=9mm]4.west) rectangle ([xshift=3mm, 
yshift= -9mm]6.east);

   \node[font = \small] at ([xshift=-12mm, yshift=8mm]4.north) {$M$};

   \node[abs node] (7) at ([yshift=-23mm]3.south) {$p$};
   \node[abs node] (8) [right of=7] {$q$};
     
    \draw[->, dashed] (7) -- (8);

    \node[font = \small] at ([xshift=14mm, yshift=-8mm]7.south) {$\hat{M}$};

\draw [smallarrows] ([xshift=5mm, yshift=-12mm]3.east) -- ([xshift=5mm, 
yshift=-22mm]3.east) node[color = white, font=\small] at ([xshift=5mm, 
yshift=-17mm]3.east) {$\alpha$};

\draw [bigarrows] ([xshift=6mm]8.east) -- ([xshift=37mm]8.east) node[color = 
white, font=\small] at ([xshift=20mm]8.east) {AddMust};
   
     \end{scope}

 \begin{scope}[xshift=8cm]
  \node[main node] (1) {$p$};
  \node[main node] (2) [below of=1] {$p$};
  \node[main node] (3) [below of=2] {$p$};
  \node[main node] (4) [right of=1] {$q$};
  \node[main node] (5) [below of=4] {$q$};
  \node[main node] (6) [below of=5] {$q$};

   \path
     (1) edge (4)
     (2) edge [line width = .7mm] (5)
     (3) edge [line width = .7mm] (6);
     
     \draw[dotted] ([xshift=-2mm, yshift=9mm]1.west) rectangle ([xshift=2mm, 
yshift= -9mm]3.east);
     \draw[dotted] ([xshift=-2mm, yshift=9mm]4.west) rectangle ([xshift=2mm, 
yshift= -9mm]6.east);

   \node[font = \small] at ([xshift=-12mm, yshift=8mm]4.north) {$M'$};

   \node[abs node] (7) at ([yshift=-23mm]3.south) {$p$};
   \node[abs node] (8) [right of=7] {$q$};
     
    \draw[->, line width = .7mm] (7) -- (8);

    \node[font = \small] at ([xshift=14mm, yshift=-8mm]7.south) {$\hat{M'}$};

\draw [smallarrows] ([xshift=5mm, yshift=-20mm]3.east) -- ([xshift=5mm, 
yshift=-10mm]3.east) node[color = white, font=\small] at ([xshift=5mm, 
yshift=-15mm]3.east) {$\gamma$};
 \end{scope}

\end{tikzpicture}}
\hfill
\subfloat[May-transition does not exist]{\begin{tikzpicture}[->,>=stealth',auto,node 
distance=2cm, scale=0.55, thick, main node/.style={scale=0.75, minimum size = 
1cm, align=center,circle,fill=blue!10,draw}, abs node/.style={scale=0.75, 
minimum size = 1cm, align=center,rectangle,fill=blue!10,draw}]

    \tikzstyle{bigarrows}=[line width=1.5mm,draw=black,-triangle 
90,postaction={draw, line width=6mm, shorten >=3mm, -}]

    \tikzstyle{smallarrows}=[line width=.5mm,draw=black,-triangle 
90,postaction={draw, line width=3.5mm, shorten >=2mm, -}]

\begin{scope}
  \node[main node] (1) {$p$};
  \node[main node] (2) [below of=1] {$p$};
  \node[main node] (3) [below of=2] {$p$};
  \node[main node] (4) [right of=1] {$q$};
  \node[main node] (5) [below of=4] {$q$};
  \node[main node] (6) [below of=5] {$q$};

     \draw[dotted] ([xshift=-2mm, yshift=9mm]1.west) rectangle ([xshift=2mm, 
yshift= -9mm]3.east);
     \draw[dotted] ([xshift=-2mm, yshift=9mm]4.west) rectangle ([xshift=2mm, 
yshift= -9mm]6.east);

   \node[font = \small] at ([xshift=-12mm, yshift=8mm]4.north) {$M$};

   \node[abs node] (7) at ([yshift=-23mm]3.south) {$p$};
   \node[abs node] (8) [right of=7] {$q$};
     
    \draw[->, dashed] (7) -- (8);

    \node[font = \small] at ([xshift=14mm, yshift=-8mm]7.south) {$\hat{M}$};

\draw [smallarrows] ([xshift=5mm, yshift=-12mm]3.east) -- ([xshift=5mm, 
yshift=-22mm]3.east) node[color = white, font=\small] at ([xshift=5mm, 
yshift=-17mm]3.east) {$\alpha$};

\draw [bigarrows] ([xshift=6mm]8.east) -- ([xshift=37mm]8.east) node[color = 
white, font=\small] at ([xshift=20mm]8.east) {AddMust};
   
     \end{scope}

 \begin{scope}[xshift=8cm]
  \node[main node] (1) {$p$};
  \node[main node] (2) [below of=1] {$p$};
  \node[main node] (3) [below of=2] {$p$};
  \node[main node] (4) [right of=1] {$q$};
  \node[main node] (5) [below of=4] {$q$};
  \node[main node] (6) [below of=5] {$q$};

   \path
     (1) edge [line width = .7mm] (4)
     (2) edge [line width = .7mm] (5)
     (3) edge [line width = .7mm] (6);
     
     \draw[dotted] ([xshift=-2mm, yshift=9mm]1.west) rectangle ([xshift=2mm, 
yshift= -9mm]3.east);
     \draw[dotted] ([xshift=-2mm, yshift=9mm]4.west) rectangle ([xshift=2mm, 
yshift= -9mm]6.east);

   \node[font = \small] at ([xshift=-12mm, yshift=8mm]4.north) {$M'$};

   \node[abs node] (7) at ([yshift=-23mm]3.south) {$p$};
   \node[abs node] (8) [right of=7] {$q$};
     
    \draw[->, line width = .7mm] (7) -- (8);

    \node[font = \small] at ([xshift=14mm, yshift=-8mm]7.south) {$\hat{M}'$};

\draw [smallarrows] ([xshift=5mm, yshift=-20mm]3.east) -- ([xshift=5mm, 
yshift=-10mm]3.east) node[color = white, font=\small] at ([xshift=5mm, 
yshift=-15mm]3.east) {$\gamma$};
 \end{scope}

\end{tikzpicture}}
\caption{\emph{AddMust}: Adding a new must-transition}
\label{fig:AddMust}
\end{figure}

\begin{prop}
\label{prop:AddMust}
For any $\hat{M}^{\prime} = AddMust(\hat{M},\hat{r}_{n})$, it 
holds that $\hat{d}(\hat{M},\hat{M}^{\prime}) = 1$.\qed
\end{prop}

\begin{defi}
\label{def:add_must_ks}
Let $M = (S,S_{0},R,L)$ be a KS and let $\alpha(M) = (\hat{S},\hat{S_{0}}, 
R_{must}, R_{may}, \hat{L})$ be the abstract KMTS derived from $M$ 
as in Def.~\ref{def:abs_kmts}.  Also, let      
$\hat{M}^{\prime} = AddMust(\alpha(M),\hat{r}_{n})$ 
for some $\hat{r}_{n} = (\hat{s}_{1},\hat{s}_{2}) \notin R_{must}$.  
The set $K_{min} \subseteq \gamma(\hat{M}^{\prime})$ with all KSs, 
whose distance $d$ from $M$ is minimized is:
\begin{equation}
K_{min} = \{M^{\prime} \mid M^{\prime} = (S, S_{0}, R \cup R_{n}, L)\} 
\end{equation}
where $R_{n}$ is given for one $s_{2} \in \gamma(\hat{s}_{2})$ as follows:  
\begin{equation} \nonumber
	R_{n} = \bigcup_{s_{1} \in \gamma(\hat{s}_{1})} \{(s_{1}, s_{2}) \mid \nexists s \in \gamma(\hat{s}_{2}): (s_{1},s) \in R\} \eqno{\qEd}
\end{equation}
\end{defi}

\noindent Def.~\ref{def:add_must_ks} implies that when the 
\emph{AbstractRepair} algorithm applies \emph{AddMust} 
on the abstract KMTS $\hat{M}$, then a set of 
KSs is retrieved from the concretization of 
$\hat{M}^{\prime}$.  The same holds for all other 
basic repair operations and consequently, when 
\emph{AbstractRepair} finds a 
repaired KMTS, one or more KSs can be obtained 
for which property $\phi$ holds.     

\begin{prop}
\label{prop:add_must}
For all $M^{\prime} \in K_{min}$, it holds that
$1 \leq d(M,M^{\prime}) \leq \left|S\right|$.
\end{prop}
\begin{proof}
Recall that $$d(M,M^{\prime}) = |S\Delta S^{\prime}| + |R\Delta R^{\prime}| + 
\frac{|\pi(L\upharpoonright_{S\cap S^{\prime}})\Delta 
\pi(L^{\prime}\upharpoonright_{S\cap S^{\prime}})|}{2}$$
Since $|S\Delta S^{\prime}| = 0$ and $|\pi(L\upharpoonright_{S\cap 
S^{\prime}})\Delta \pi(L^{\prime}\upharpoonright_{S\cap S^{\prime}})| = 0$, 
$d(M,M^{\prime}) = |R\Delta R^{\prime}| = 
|R - R^{\prime}| + |R^{\prime} - R| = 0 + |R_{n}|$.
Since $|R_{n}| \geq 1$ and $|R_{n}| \leq |S|$, it is proved that 
$1 \leq d(M,M^{\prime}) \leq \left|S\right|$. 
\end{proof}
From Prop.~\ref{prop:add_must}, we conclude that 
a lower and upper bound exists for the distance 
between $M$ and any $M^{\prime} \in K_{min}$.

\subsubsection{Adding a may-transition}
\begin{defi}[AddMay]
\label{def:AddMay}
For a given KMTS $\hat{M} = (\hat{S},\hat{S_{0}}, R_{must}, 
R_{may}, \hat{L})$ and 
$\hat{r}_{n} = (\hat{s}_{1},\hat{s}_{2}) \notin R_{may}$, 
$AddMay(\hat{M},\hat{r}_{n})$ is the KMTS 
$\hat{M^{\prime}} = (\hat{S}, \hat{S_{0}}, 
R_{must}^{\prime}, R_{may}^{\prime}, \hat{L})$
such that $R_{must}^{\prime} = R_{must} \cup \{\hat{r}_{n}\}$ 
if $\left|S_{1}\right| = 1$ or
$R_{must}^{\prime} = R_{must}$ if $\left|S_{1}\right| > 1$ 
for $S_{1} = \{s_{1} \mid s_{1} \in \gamma(\hat{s}_{1})\}$ and 
$R_{may}^{\prime} = R_{may} \cup \{\hat{r}_{n}\}$. \qed
\end{defi}

From Def.~\ref{def:AddMay}, we conclude that there are 
two different cases in adding a new may-transition 
$\hat{r}_{n}$; adding also a must-transition or not.  
In fact, $\hat{r}_{n}$ is also a must-transition if and 
only if the set of the corresponding concrete states of 
$\hat{s}_{1}$ is a singleton.  Fig.~\ref{fig:AddMay} 
displays the two different cases of applying basic 
repair operation \emph{AddMay} to a KMTS.

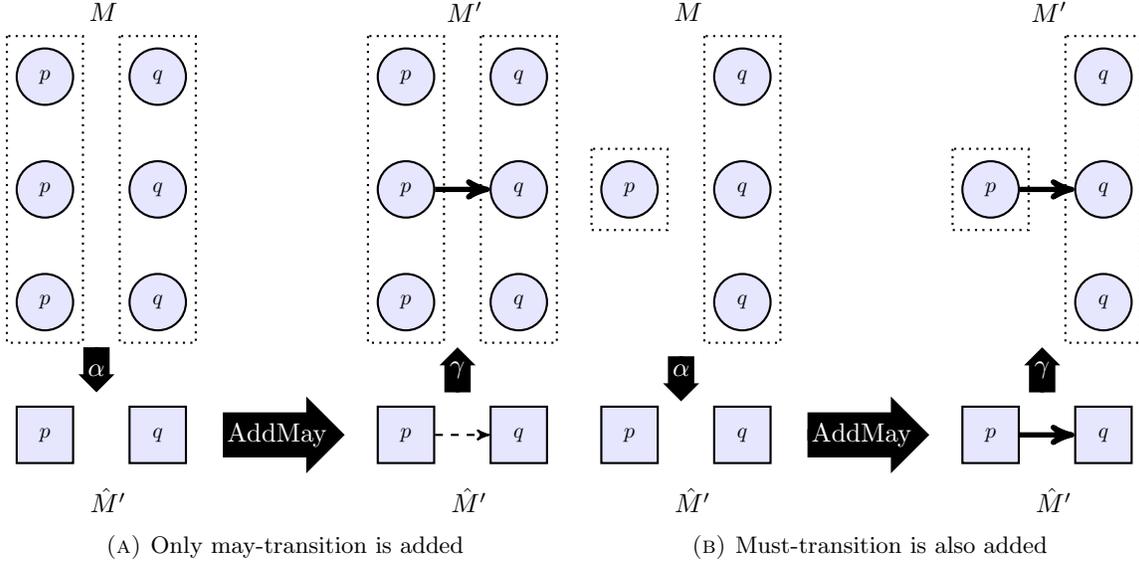
\begin{figure}[htb]
\centering
\subfloat[Only may-transition is 
added]{\begin{tikzpicture}[->,>=stealth',auto,node 
distance=2cm, scale=0.6, thick, main node/.style={scale=0.75, minimum size = 
1cm, align=center,circle,fill=blue!10,draw}, abs node/.style={scale=0.75, 
minimum size = 1cm, align=center,rectangle,fill=blue!10,draw}]

    \tikzstyle{bigarrows}=[line width=1.5mm,draw=black,-triangle 
90,postaction={draw, line width=6mm, shorten >=3mm, -}]

    \tikzstyle{smallarrows}=[line width=.5mm,draw=black,-triangle 
90,postaction={draw, line width=3.5mm, shorten >=2mm, -}]

\begin{scope}
  \node[main node] (1) {$p$};
  \node[main node] (2) [below of=1] {$p$};
  \node[main node] (3) [below of=2] {$p$};
  \node[main node] (4) [right of=1] {$q$};
  \node[main node] (5) [below of=4] {$q$};
  \node[main node] (6) [below of=5] {$q$};

     \draw[dotted] ([xshift=-2mm, yshift=9mm]1.west) rectangle ([xshift=2mm, 
yshift= -9mm]3.east);
     \draw[dotted] ([xshift=-2mm, yshift=9mm]4.west) rectangle ([xshift=2mm, 
yshift= -9mm]6.east);

   \node[font = \small] at ([xshift=-12mm, yshift=8mm]4.north) {$M$};

   \node[abs node] (7) at ([yshift=-23mm]3.south) {$p$};
   \node[abs node] (8) [right of=7] {$q$};
     
    \node[font = \small] at ([xshift=14mm, yshift=-8mm]7.south) {$\hat{M}'$};

\draw [smallarrows] ([xshift=5mm, yshift=-10mm]3.east) -- ([xshift=5mm, 
yshift=-20mm]3.east) node[color = white, font=\small] at ([xshift=5mm, 
yshift=-15mm]3.east) {$\alpha$};

\draw [bigarrows] ([xshift=8mm]8.east) -- ([xshift=35mm]8.east) node[color = 
white, font=\small] at ([xshift=20mm]8.east) {AddMay};
   
     \end{scope}

 \begin{scope}[xshift=8cm]
  \node[main node] (1) {$p$};
  \node[main node] (2) [below of=1] {$p$};
  \node[main node] (3) [below of=2] {$p$};
  \node[main node] (4) [right of=1] {$q$};
  \node[main node] (5) [below of=4] {$q$};
  \node[main node] (6) [below of=5] {$q$};

   \path
     (2) edge [line width = .7mm] (5);
     
     \draw[dotted] ([xshift=-2mm, yshift=9mm]1.west) rectangle ([xshift=2mm, 
yshift= -9mm]3.east);
     \draw[dotted] ([xshift=-2mm, yshift=9mm]4.west) rectangle ([xshift=2mm, 
yshift= -9mm]6.east);

   \node[font = \small] at ([xshift=-12mm, yshift=8mm]4.north) {$M'$};

   \node[abs node] (7) at ([yshift=-23mm]3.south) {$p$};
   \node[abs node] (8) [right of=7] {$q$};
     
    \draw[->, dashed] (7) -- (8);

    \node[font = \small] at ([xshift=14mm, yshift=-8mm]7.south) {$\hat{M}'$};

\draw [smallarrows] ([xshift=5mm, yshift=-20mm]3.east) -- ([xshift=5mm, 
yshift=-10mm]3.east) node[color = white, font=\small] at ([xshift=5mm, 
yshift=-15mm]3.east) {$\gamma$};
 \end{scope}

\end{tikzpicture}}
\hfill
\subfloat[Must-transition is also 
added]{\begin{tikzpicture}[->,>=stealth',auto,node 
distance=2cm, scale=0.6, thick, main node/.style={scale=0.75, minimum size = 
1cm, align=center,circle,fill=blue!10,draw}, abs node/.style={scale=0.75, 
minimum size = 1cm, align=center,rectangle,fill=blue!10,draw}]

    \tikzstyle{bigarrows}=[line width=1.5mm,draw=black,-triangle 
90,postaction={draw, line width=6mm, shorten >=3mm, -}]

    \tikzstyle{smallarrows}=[line width=.5mm,draw=black,-triangle 
90,postaction={draw, line width=3.5mm, shorten >=2mm, -}]

\begin{scope}
  \node[main node] (2) {$p$};
  \node[main node] (5) [right of=2] {$q$};
  \node[main node] (4) [above of=5]{$q$};
  \node[main node] (6) [below of=5] {$q$};
     
     \draw[dotted] ([xshift=-2mm, yshift=9mm]2.west) rectangle ([xshift=2mm, 
yshift= -9mm]2.east);
     \draw[dotted] ([xshift=-2mm, yshift=9mm]4.west) rectangle ([xshift=2mm, 
yshift= -9mm]6.east);

   \node[font = \small] at ([xshift=-12mm, yshift=8mm]4.north) {$M$};

   \node[abs node] (7) at ([yshift=-48mm]2.south) {$p$};
   \node[abs node] (8) [right of=7] {$q$};
     
    \node[font = \small] at ([xshift=14mm, yshift=-8mm]7.south) {$\hat{M}'$};

\draw [smallarrows] ([xshift=5mm, yshift=-37mm]2.east) -- ([xshift=5mm, 
yshift=-47mm]2.east) node[color = white, font=\small] at ([xshift=5mm, 
yshift=-40mm]2.east) {$\alpha$};

\draw [bigarrows] ([xshift=8mm]8.east) -- ([xshift=35mm]8.east) node[color = 
white, font=\small] at ([xshift=20mm]8.east) {AddMay};
   
     \end{scope}

 \begin{scope}[xshift=8cm]
  \node[main node] (2) {$p$};
  \node[main node] (5) [right of=2] {$q$};
  \node[main node] (4) [above of=5]{$q$};
  \node[main node] (6) [below of=5] {$q$};

   \path
     (2) edge [line width = .7mm] (5);
     
     \draw[dotted] ([xshift=-2mm, yshift=9mm]2.west) rectangle ([xshift=2mm, 
yshift= -9mm]2.east);
     \draw[dotted] ([xshift=-2mm, yshift=9mm]4.west) rectangle ([xshift=2mm, 
yshift= -9mm]6.east);

   \node[font = \small] at ([xshift=-12mm, yshift=8mm]4.north) {$M'$};

   \node[abs node] (7) at ([yshift=-48mm]2.south) {$p$};
   \node[abs node] (8) [right of=7] {$q$};
     
    \draw[->, line width = .7mm] (7) -- (8);

    \node[font = \small] at ([xshift=14mm, yshift=-8mm]7.south) {$\hat{M}'$};

\draw [smallarrows] ([xshift=5mm, yshift=5mm]3.east) -- ([xshift=5mm, 
yshift=15mm]3.east) node[color = white, font=\small] at ([xshift=5mm, 
yshift=10mm]3.east) {$\gamma$};
 \end{scope}

\end{tikzpicture}}
\caption{\emph{AddMay}: Adding a new must-transition}
\label{fig:AddMay}
\end{figure}

\begin{prop}
\label{prop:AddMay}
For any $\hat{M}^{\prime} = AddMay(\hat{M},\hat{r}_{n})$, it 
holds that $\hat{d}(\hat{M},\hat{M}^{\prime}) = 1$.\qed
\end{prop}

\begin{defi}
\label{def:add_may_ks}
Let $M = (S,S_{0},R,L)$ be a KS and let $\alpha(M) = (\hat{S},\hat{S_{0}}, 
R_{must}, R_{may}, \hat{L})$ be the abstract KMTS derived from $M$ 
as in Def.~\ref{def:abs_kmts}.  
Also, let $\hat{M}^{\prime} = AddMay(\alpha(M),\hat{r}_{n})$ 
for some 
$\hat{r}_{n} = (\hat{s}_{1},\hat{s}_{2}) \notin R_{may}$.  
The set $K_{min} \subseteq \gamma(\hat{M}^{\prime})$ with all KSs, 
whose structural distance $d$ from $M$ is minimized is 
given by:
\begin{equation}
K_{min} = \{M^{\prime} \mid M^{\prime} = (S, S_{0}, R \cup \{r_{n}\}, L)\} 
\end{equation}
where $r_{n} \in R_{n}$ and $R_{n} = \{r_{n}=(s_{1},s_{2}) \mid 
s_{1} \in \gamma(\hat{s}_{1}), s_{2} \in \gamma(\hat{s}_{2})$ and  
$r_{n} \notin R\}$.\qed
\end{defi}

\begin{prop}
\label{prop:add_may}
For all $M^{\prime} \in K_{min}$, it holds that $d(M,M^{\prime}) = 1$.
\end{prop}
\begin{proof}
$d(M,M^{\prime}) = |S\Delta S^{\prime}| + |R\Delta R^{\prime}| + 
\frac{|\pi(L\upharpoonright_{S\cap S^{\prime}})\Delta \pi(L^{\prime}\upharpoonright_{S\cap S^{\prime}})|}{2}$.
Because $|S\Delta S^{\prime}| = 0$ and $|\pi(L\upharpoonright_{S\cap S^{\prime}})\Delta \pi(L^{\prime}\upharpoonright_{S\cap S^{\prime}})| = 0$, 
$d(M,M^{\prime}) = |R\Delta R^{\prime}| = 
|R - R^{\prime}| + |R^{\prime} - R| = 0 + |\{r_{n}\}| = 1$.
So, we prove that $d(M,M^{\prime}) = 1$. 
\end{proof}

\subsubsection{Removing a must-transition}
\begin{defi}[RemoveMust]
\label{def:RemoveMust}
For a given KMTS $\hat{M} = (\hat{S},\hat{S_{0}}, R_{must}, 
R_{may}, \hat{L})$ and 
$\hat{r}_{m} = (\hat{s}_{1},\hat{s}_{2}) \in R_{must}$, 
$RemoveMust(\hat{M},\hat{r}_{m})$ is the KMTS 
$\hat{M^{\prime}} = (\hat{S}, \hat{S_{0}}, 
R_{must}^{\prime},$ $R_{may}^{\prime}, \hat{L})$
such that $R_{must}^{\prime} = R_{must} - \{\hat{r}_{m}\}$ and 
$R_{may}^{\prime} = R_{may} - \{\hat{r}_{m}\}$ 
if $\left|S_{1}\right| = 1$ or
$R_{may}^{\prime} = R_{may}$ if $\left|S_{1}\right| > 1$ 
for $S_{1} = \{s_{1} \mid s_{1} \in \gamma(\hat{s}_{1})\}$.\qed
\end{defi}

Removing a must-transition $\hat{r}_{m}$, in some 
special and maybe rare cases, could also result in the 
deletion of the may-transition $\hat{r}_{m}$ as well.  In fact, 
this occurs if transitions to the concrete states of 
$\hat{s}_{2}$ exist only from one concrete state 
of the corresponding ones of $\hat{s}_{1}$.  
These two cases for function \emph{RemoveMust} are 
presented graphically in Fig.~\ref{fig:RemoveMust}.     

\begin{figure}[htb]
\centering
\subfloat[May-transition is not 
removed]{\begin{tikzpicture}[->,>=stealth',auto,node 
distance=2cm, scale=0.55, thick, main node/.style={scale=0.75, minimum size = 
1cm, align=center,circle,fill=blue!10,draw}, abs node/.style={scale=0.75, 
minimum size = 1cm, align=center,rectangle,fill=blue!10,draw}]

    \tikzstyle{bigarrows}=[line width=1.5mm,draw=black,-triangle 
90,postaction={draw, line width=6mm, shorten >=3mm, -}]

    \tikzstyle{smallarrows}=[line width=.5mm,draw=black,-triangle 
90,postaction={draw, line width=3.5mm, shorten >=2mm, -}]

\begin{scope}
  \node[main node] (1) {$p$};
  \node[main node] (2) [below of=1] {$p$};
  \node[main node] (3) [below of=2] {$p$};
  \node[main node] (4) [right of=1] {$q$};
  \node[main node] (5) [below of=4] {$q$};
  \node[main node] (6) [below of=5] {$q$};

   \path
     (1) edge (4)
     (2) edge (5)
     (3) edge (6);
     
     \draw[dotted] ([xshift=-3mm, yshift=9mm]1.west) rectangle ([xshift=3mm, 
yshift= -9mm]3.east);
     \draw[dotted] ([xshift=-3mm, yshift=9mm]4.west) rectangle ([xshift=3mm, 
yshift= -9mm]6.east);

   \node[font = \small] at ([xshift=-12mm, yshift=8mm]4.north) {$M$};

   \node[abs node] (7) at ([yshift=-23mm]3.south) {$p$};
   \node[abs node] (8) [right of=7] {$q$};
     
    \draw[->] (7) -- (8);

    \node[font = \small] at ([xshift=14mm, yshift=-8mm]7.south) {$\hat{M}$};

\draw [smallarrows] ([xshift=5mm, yshift=-12mm]3.east) -- ([xshift=5mm, 
yshift=-22mm]3.east) node[color = white, font=\small] at ([xshift=5mm, 
yshift=-17mm]3.east) {$\alpha$};

\draw [bigarrows] ([xshift=1mm]8.east) -- ([xshift=41mm]8.east) node[color = 
white, font=\small] at ([xshift=20mm]8.east) {RemoveMust};
   
     \end{scope}

 \begin{scope}[xshift=8.3cm]
  \node[main node] (1) {$p$};
  \node[main node] (2) [below of=1] {$p$};
  \node[main node] (3) [below of=2] {$p$};
  \node[main node] (4) [right of=1] {$q$};
  \node[main node] (5) [below of=4] {$q$};
  \node[main node] (6) [below of=5] {$q$};

   \path
     (1) edge (4)
     (2) edge (5)
     (3) edge (6);
     
\node[black, auto=false, cross out, -, draw] at ([xshift=6.5mm]3.east){};

     \draw[dotted] ([xshift=-2mm, yshift=9mm]1.west) rectangle ([xshift=2mm, 
yshift= -9mm]3.east);
     \draw[dotted] ([xshift=-2mm, yshift=9mm]4.west) rectangle ([xshift=2mm, 
yshift= -9mm]6.east);

   \node[font = \small] at ([xshift=-12mm, yshift=8mm]4.north) {$M'$};

   \node[abs node] (7) at ([yshift=-23mm]3.south) {$p$};
   \node[abs node] (8) [right of=7] {$q$};
     
    \draw[->, dashed] (7) -- (8);

    \node[font = \small] at ([xshift=14mm, yshift=-8mm]7.south) {$\hat{M'}$};

\draw [smallarrows] ([xshift=5mm, yshift=-20mm]3.east) -- ([xshift=5mm, 
yshift=-10mm]3.east) node[color = white, font=\small] at ([xshift=5mm, 
yshift=-15mm]3.east) {$\gamma$};

\end{scope}

\end{tikzpicture}}    
\hfill
\subfloat[May-transition is also 
removed]{\begin{tikzpicture}[->,>=stealth',auto,node 
distance=2cm, scale=0.6, thick, main node/.style={scale=0.75, minimum size = 
1cm, align=center,circle,fill=blue!10,draw}, abs node/.style={scale=0.75, 
minimum size = 1cm, align=center,rectangle,fill=blue!10,draw}]

    \tikzstyle{bigarrows}=[line width=1.5mm,draw=black,-triangle 
90,postaction={draw, line width=6mm, shorten >=3mm, -}]

    \tikzstyle{smallarrows}=[line width=.5mm,draw=black,-triangle 
90,postaction={draw, line width=3.5mm, shorten >=2mm, -}]

\begin{scope}
  \node[main node] (2) {$p$};
  \node[main node] (5) [right of=2] {$q$};
  \node[main node] (4) [above of=5]{$q$};
  \node[main node] (6) [below of=5] {$q$};

  \path
    (2) edge (5);
    
     \draw[dotted] ([xshift=-2mm, yshift=9mm]2.west) rectangle ([xshift=2mm, 
yshift= -9mm]2.east);
     \draw[dotted] ([xshift=-2mm, yshift=9mm]4.west) rectangle ([xshift=2mm, 
yshift= -9mm]6.east);

   \node[font = \small] at ([xshift=-12mm, yshift=8mm]4.north) {$M$};

   \node[abs node] (7) at ([yshift=-48mm]2.south) {$p$};
   \node[abs node] (8) [right of=7] {$q$};

     \path
    (7) edge (8);

    \node[font = \small] at ([xshift=14mm, yshift=-8mm]7.south) {$\hat{M}'$};

\draw [smallarrows] ([xshift=5mm, yshift=-37mm]2.east) -- ([xshift=5mm, 
yshift=-47mm]2.east) node[color = white, font=\small] at ([xshift=5mm, 
yshift=-40mm]2.east) {$\alpha$};

\draw [bigarrows] ([xshift=4mm]8.east) -- ([xshift=41mm]8.east) node[color = 
white, font=\small] at ([xshift=21mm]8.east) {RemoveMust};
   
     \end{scope}

 \begin{scope}[xshift=8.3cm]
  \node[main node] (2) {$p$};
  \node[main node] (5) [right of=2] {$q$};
  \node[main node] (4) [above of=5]{$q$};
  \node[main node] (6) [below of=5] {$q$};

   \path
     (2) edge (5);
     
     \draw[dotted] ([xshift=-2mm, yshift=9mm]2.west) rectangle ([xshift=2mm, 
yshift= -9mm]2.east);
     \draw[dotted] ([xshift=-2mm, yshift=9mm]4.west) rectangle ([xshift=2mm, 
yshift= -9mm]6.east);

   \node[font = \small] at ([xshift=-12mm, yshift=8mm]4.north) {$M'$};

   \node[abs node] (7) at ([yshift=-48mm]2.south) {$p$};
   \node[abs node] (8) [right of=7] {$q$};
     
    \draw[->] (7) -- (8);

    \node[font = \small] at ([xshift=14mm, yshift=-8mm]7.south) {$\hat{M}'$};

\draw [smallarrows] ([xshift=5mm, yshift=5mm]3.east) -- ([xshift=5mm, 
yshift=15mm]3.east) node[color = white, font=\small] at ([xshift=5mm, 
yshift=10mm]3.east) {$\gamma$};

\node[black, auto=false, cross out, -, draw] at ([xshift=5mm]7.east){};
\node[black, auto=false, cross out, -, draw] at ([xshift=6mm]2.east){};

 \end{scope}

\end{tikzpicture}}
\caption{\emph{RemoveMust}: Removing an existing must-transition}
\label{fig:RemoveMust}
\end{figure}
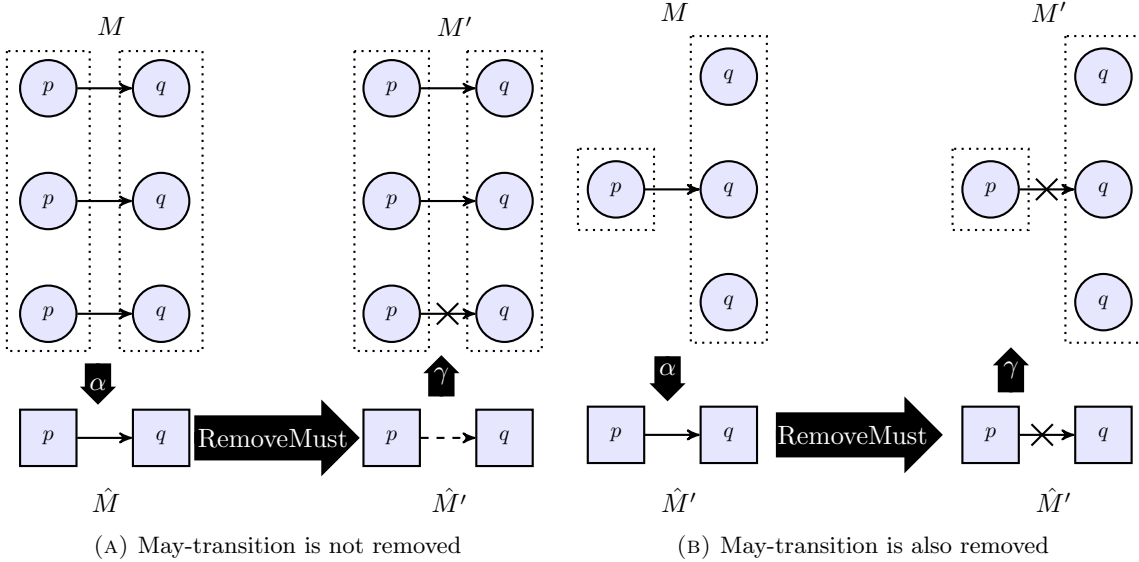

\begin{prop}
\label{prop:RemoveMust}
For any $\hat{M}^{\prime} = RemoveMust(\hat{M},\hat{r}_{m})$, it 
holds that $\hat{d}(\hat{M},\hat{M}^{\prime}) = 1$.\qed
\end{prop}

\begin{defi}
\label{def:remove_must_ks}
Let $M = (S,S_{0},R,L)$ be a KS and let $\alpha(M) = (\hat{S},\hat{S_{0}}, 
R_{must}, R_{may}, \hat{L})$ be the abstract KMTS derived from $M$ 
as in Def.~\ref{def:abs_kmts}.  Also, let 
$\hat{M}^{\prime} = RemoveMust(\alpha(M),\hat{r}_{m})$ for some 
$\hat{r}_{m} = (\hat{s}_{1},\hat{s}_{2}) \in R_{must}$.
The set $K_{min} \subseteq \gamma(\hat{M}^{\prime})$ with all KSs, 
whose structural distance $d$ from $M$ is minimized is given by:
\begin{equation}
K_{min} = \{M^{\prime} \mid M^{\prime} = (S, S_{0}, R - \{R_{m}\}, L)\} 
\end{equation}
where $R_{m}$ is given for one $s_{1} \in \gamma(\hat{s}_{1})$ as follows:
\begin{equation} \nonumber
R_{m} =  \bigcup_{s_{2} \in \gamma(\hat{s}_{2})} \{(s_{1}, s_{2}) \in R\} \eqno{\qEd}
\end{equation}
\end{defi}

\begin{prop}
\label{prop:remove_must}
For $M^{\prime}$, it holds that
$1 \leq d(M,M^{\prime}) \leq \left|S\right|$.
\end{prop}
\begin{proof}
$d(M,M^{\prime}) = |S\Delta S^{\prime}| + |R\Delta R^{\prime}| + 
\frac{|\pi(L\upharpoonright_{S\cap S^{\prime}})\Delta \pi(L^{\prime}\upharpoonright_{S\cap S^{\prime}})|}{2}$.
Because $|S\Delta S^{\prime}| = 0$ and $|\pi(L\upharpoonright_{S\cap S^{\prime}})\Delta \pi(L^{\prime}\upharpoonright_{S\cap S^{\prime}})| = 0$, 
$d(M,M^{\prime}) = |R\Delta R^{\prime}| = 
|R - R^{\prime}| + |R^{\prime} - R| = |R_{m}| +  0 = |R_{m}|$.
It holds that $|R_{m}| \geq 1$ and $|R_{m}| \leq |S|$.
So, we proved that 
$1 \leq d(M,M^{\prime}) \leq \left|S\right|$. 
\end{proof}

\subsubsection{Removing a may-transition}
\begin{defi}[RemoveMay]
\label{def:RemoveMay}
For a given KMTS $\hat{M} = (\hat{S},\hat{S_{0}}, R_{must}, 
R_{may}, \hat{L})$ and 
$\hat{r}_{m} = (\hat{s}_{1},\hat{s}_{2}) \in R_{may}$, 
$RemoveMay(\hat{M},\hat{r}_{m})$ is the KMTS 
$\hat{M^{\prime}} = (\hat{S}, \hat{S_{0}}, 
R_{must}^{\prime}, R_{may}^{\prime},$ $\hat{L})$
such that $R_{must}^{\prime} = R_{must} - \{\hat{r}_{m}\}$ and 
$R_{may}^{\prime} = R_{may} - \{\hat{r}_{m}\}$.   \qed
\end{defi}

Def.~\ref{def:RemoveMay} ensures that removing a 
may-transition $\hat{r}_{m}$ implies the removal of 
a must-transition, if $\hat{r}_{m}$ is 
also a must-transition.  Otherwise, there are not 
any changes in the set of must-transitions $R_{must}$.  
Fig.~\ref{fig:RemoveMay} shows how function \emph{RemoveMay} 
works in both cases.

\begin{figure}[htb]
\centering
\subfloat[May-transition is also a 
must-transition]{\begin{tikzpicture}[->,>=stealth',auto,node 
distance=2cm, scale=0.55, thick, main node/.style={scale=0.75, minimum size = 
1cm, align=center,circle,fill=blue!10,draw}, abs node/.style={scale=0.75, 
minimum size = 1cm, align=center,rectangle,fill=blue!10,draw}]

    \tikzstyle{bigarrows}=[line width=1.5mm,draw=black,-triangle 
90,postaction={draw, line width=6mm, shorten >=3mm, -}]

    \tikzstyle{smallarrows}=[line width=.5mm,draw=black,-triangle 
90,postaction={draw, line width=3.5mm, shorten >=2mm, -}]

\begin{scope}
  \node[main node] (1) {$p$};
  \node[main node] (2) [below of=1] {$p$};
  \node[main node] (3) [below of=2] {$p$};
  \node[main node] (4) [right of=1] {$q$};
  \node[main node] (5) [below of=4] {$q$};
  \node[main node] (6) [below of=5] {$q$};

   \path
     (1) edge (4)
     (2) edge (5)
     (3) edge (6);
     
     \draw[dotted] ([xshift=-3mm, yshift=9mm]1.west) rectangle ([xshift=3mm, 
yshift= -9mm]3.east);
     \draw[dotted] ([xshift=-3mm, yshift=9mm]4.west) rectangle ([xshift=3mm, 
yshift= -9mm]6.east);

   \node[font = \small] at ([xshift=-12mm, yshift=8mm]4.north) {$M$};

   \node[abs node] (7) at ([yshift=-23mm]3.south) {$p$};
   \node[abs node] (8) [right of=7] {$q$};
     
    \draw[->] (7) -- (8);

    \node[font = \small] at ([xshift=14mm, yshift=-8mm]7.south) {$\hat{M}$};

\draw [smallarrows] ([xshift=5mm, yshift=-12mm]3.east) -- ([xshift=5mm, 
yshift=-22mm]3.east) node[color = white, font=\small] at ([xshift=5mm, 
yshift=-17mm]3.east) {$\alpha$};

\draw [bigarrows] ([xshift=1mm]8.east) -- ([xshift=41mm]8.east) node[color = 
white, font=\small] at ([xshift=20mm]8.east) {RemoveMay};
   
     \end{scope}

 \begin{scope}[xshift=8.3cm]
  \node[main node] (1) {$p$};
  \node[main node] (2) [below of=1] {$p$};
  \node[main node] (3) [below of=2] {$p$};
  \node[main node] (4) [right of=1] {$q$};
  \node[main node] (5) [below of=4] {$q$};
  \node[main node] (6) [below of=5] {$q$};

   \path
     (1) edge (4)
     (2) edge (5)
     (3) edge (6);
     
     \draw[dotted] ([xshift=-2mm, yshift=9mm]1.west) rectangle ([xshift=2mm, 
yshift= -9mm]3.east);
     \draw[dotted] ([xshift=-2mm, yshift=9mm]4.west) rectangle ([xshift=2mm, 
yshift= -9mm]6.east);

   \node[font = \small] at ([xshift=-12mm, yshift=8mm]4.north) {$M'$};

   \node[abs node] (7) at ([yshift=-23mm]3.south) {$p$};
   \node[abs node] (8) [right of=7] {$q$};
     
    \draw[->] (7) -- (8);

    \node[font = \small] at ([xshift=14mm, yshift=-8mm]7.south) {$\hat{M'}$};

\draw [smallarrows] ([xshift=5mm, yshift=-20mm]3.east) -- ([xshift=5mm, 
yshift=-10mm]3.east) node[color = white, font=\small] at ([xshift=5mm, 
yshift=-15mm]3.east) {$\gamma$}; 

\node[black, auto=false, cross out, -, draw] at ([xshift=6.5mm]1.east){};
\node[black, auto=false, cross out, -, draw] at ([xshift=6.5mm]2.east){};
\node[black, auto=false, cross out, -, draw] at ([xshift=6.5mm]3.east){};
\node[black, auto=false, cross out, -, draw] at ([xshift=6.4mm]7.east){};

\end{scope}

\end{tikzpicture}}            
\hfill    
\subfloat[May-transition is not a 
must-transition]{\begin{tikzpicture}[->,>=stealth',auto,node 
distance=2cm, scale=0.55, thick, main node/.style={scale=0.75, minimum size = 
1cm, align=center,circle,fill=blue!10,draw}, abs node/.style={scale=0.75, 
minimum size = 1cm, align=center,rectangle,fill=blue!10,draw}]

    \tikzstyle{bigarrows}=[line width=1.5mm,draw=black,-triangle 
90,postaction={draw, line width=6mm, shorten >=3mm, -}]

    \tikzstyle{smallarrows}=[line width=.5mm,draw=black,-triangle 
90,postaction={draw, line width=3.5mm, shorten >=2mm, -}]

\begin{scope}
  \node[main node] (1) {$p$};
  \node[main node] (2) [below of=1] {$p$};
  \node[main node] (3) [below of=2] {$p$};
  \node[main node] (4) [right of=1] {$q$};
  \node[main node] (5) [below of=4] {$q$};
  \node[main node] (6) [below of=5] {$q$};

   \path
     (2) edge (5);
     
     \draw[dotted] ([xshift=-3mm, yshift=9mm]1.west) rectangle ([xshift=3mm, 
yshift= -9mm]3.east);
     \draw[dotted] ([xshift=-3mm, yshift=9mm]4.west) rectangle ([xshift=3mm, 
yshift= -9mm]6.east);

   \node[font = \small] at ([xshift=-12mm, yshift=8mm]4.north) {$M$};

   \node[abs node] (7) at ([yshift=-23mm]3.south) {$p$};
   \node[abs node] (8) [right of=7] {$q$};
     
    \draw[->, dashed] (7) -- (8);

    \node[font = \small] at ([xshift=14mm, yshift=-8mm]7.south) {$\hat{M}$};

\draw [smallarrows] ([xshift=5mm, yshift=-12mm]3.east) -- ([xshift=5mm, 
yshift=-22mm]3.east) node[color = white, font=\small] at ([xshift=5mm, 
yshift=-17mm]3.east) {$\alpha$};

\draw [bigarrows] ([xshift=1mm]8.east) -- ([xshift=41mm]8.east) node[color = 
white, font=\small] at ([xshift=20mm]8.east) {RemoveMay};
   
     \end{scope}

 \begin{scope}[xshift=8.3cm]
  \node[main node] (1) {$p$};
  \node[main node] (2) [below of=1] {$p$};
  \node[main node] (3) [below of=2] {$p$};
  \node[main node] (4) [right of=1] {$q$};
  \node[main node] (5) [below of=4] {$q$};
  \node[main node] (6) [below of=5] {$q$};

   \path
     (2) edge (5);
     
     \draw[dotted] ([xshift=-2mm, yshift=9mm]1.west) rectangle ([xshift=2mm, 
yshift= -9mm]3.east);
     \draw[dotted] ([xshift=-2mm, yshift=9mm]4.west) rectangle ([xshift=2mm, 
yshift= -9mm]6.east);

   \node[font = \small] at ([xshift=-12mm, yshift=8mm]4.north) {$M'$};

   \node[abs node] (7) at ([yshift=-23mm]3.south) {$p$};
   \node[abs node] (8) [right of=7] {$q$};
     
    \draw[->, dashed] (7) -- (8);

    \node[font = \small] at ([xshift=14mm, yshift=-8mm]7.south) {$\hat{M'}$};

\draw [smallarrows] ([xshift=5mm, yshift=-20mm]3.east) -- ([xshift=5mm, 
yshift=-10mm]3.east) node[color = white, font=\small] at ([xshift=5mm, 
yshift=-15mm]3.east) {$\gamma$}; 

\node[black, auto=false, cross out, -, draw] at ([xshift=6.5mm]2.east){};
\node[black, auto=false, cross out, -, draw] at ([xshift=6.4mm]7.east){};

\end{scope}

\end{tikzpicture}}
\caption{\emph{RemoveMay}: Removing an existing may-transition}
\label{fig:RemoveMay}
\end{figure}
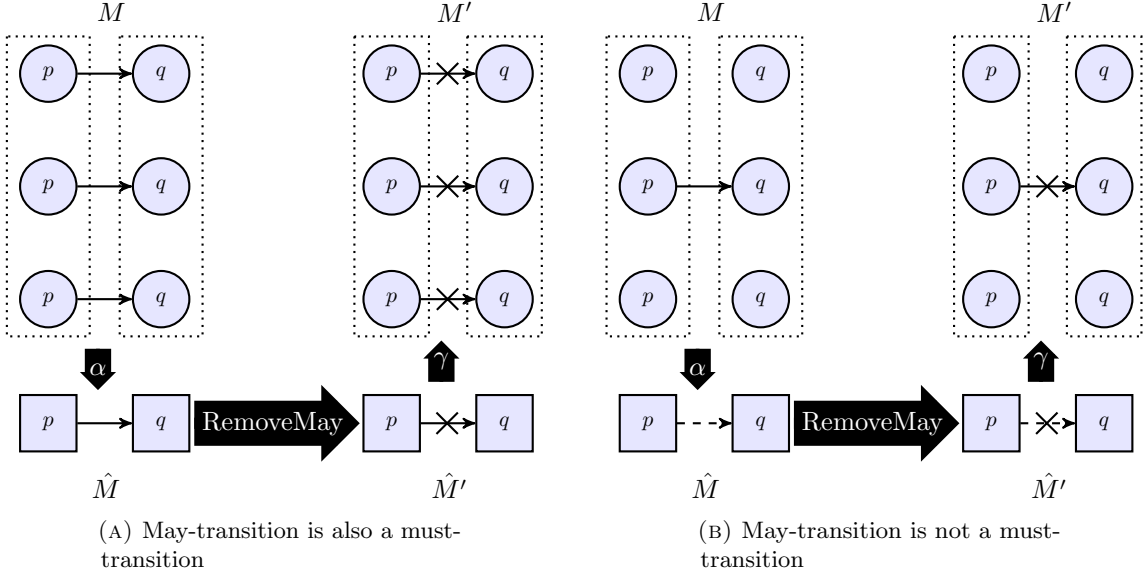

\begin{prop}
\label{prop:RemoveMay}
For any $\hat{M}^{\prime} = RemoveMay(\hat{M},\hat{r}_{m})$, it 
holds that $\hat{d}(\hat{M},\hat{M}^{\prime}) = 1$.\qed
\end{prop}

\begin{defi}
\label{def:remove_may_ks}
Let $M = (S,S_{0},R,L)$ be a KS and let $\alpha(M) = (\hat{S},\hat{S_{0}}, 
R_{must}, R_{may}, \hat{L})$ be the abstract KMTS derived from $M$ 
as in Def.~\ref{def:abs_kmts}.  Also, let 
$\hat{M}^{\prime} = RemoveMay(\alpha(M),\hat{r}_{m})$ 
for some $\hat{r}_{m} = (\hat{s}_{1},\hat{s}_{2}) \in R_{may}$ 
with $\hat{s}_{1},\hat{s}_{2} \in \hat{S}$.  
The KS $M^{\prime} \in \gamma(\hat{M}^{\prime})$, 
whose structural distance $d$ from $M$ is minimized is given by:
\begin{equation}
M^{\prime} = (S, S_{0}, R - R_{m}, L\} 
\end{equation}
where 
$R_{m} = \{r_{m}=(s_{1},s_{2}) \mid 
s_{1} \in \gamma(\hat{s}_{1}), 
s_{2} \in \gamma(\hat{s}_{2})$ and $r_{m} \in R\}$.\qed     
\end{defi}

\begin{prop}
\label{prop:remove_may}
For $M^{\prime}$, it holds that
$1 \leq d(M,M^{\prime}) \leq \left|S\right|^{2}$.
\end{prop}
\begin{proof}
$d(M,M^{\prime}) = |S\Delta S^{\prime}| + |R\Delta R^{\prime}| + 
\frac{|\pi(L\upharpoonright_{S\cap S^{\prime}})\Delta \pi(L^{\prime}\upharpoonright_{S\cap S^{\prime}})|}{2}$.
Because $|S\Delta S^{\prime}| = 0$ and $|\pi(L\upharpoonright_{S\cap S^{\prime}})\Delta \pi(L^{\prime}\upharpoonright_{S\cap S^{\prime}})| = 0$, 
$d(M,M^{\prime}) = |R\Delta R^{\prime}| = 
|R - R^{\prime}| + |R^{\prime} - R| = 0 + |R_{m}| = |R_{m}|$.
It holds that $|R_{m}| \geq 1$ and $|R_{m}| \leq |S|^{2}$.
So, we proved that 
$1 \leq d(M,M^{\prime}) \leq \left|S\right|^{2}$. 
\end{proof}

\subsubsection{Changing the labeling of a KMTS state}
\begin{defi}[ChangeLabel]
\label{def:ChangeLabel}
For a given KMTS $\hat{M} = (\hat{S},\hat{S_{0}}, R_{must}, 
R_{may}, \hat{L})$, a state $\hat{s} \in \hat{S}$ and an 
atomic CTL formula $\phi$ with $\phi \in 2^{LIT}$, 
$ChangeLabel(\hat{M},\hat{s},\phi)$ is the KMTS 
$\hat{M^{\prime}} = (\hat{S}, \hat{S_{0}}, 
R_{must}, R_{may}, \hat{L^{\prime}})$
such that  
$\hat{L^{\prime}} = ( \hat{L} - \{\hat{l}_{old}\} ) 
\cup \{\hat{l}_{new}\}$ for
$\hat{l}_{old} = (\hat{s},lit_{old})$ and 
$\hat{l}_{new} = (\hat{s},lit_{new})$ where 
$lit_{new} = \hat{L}(\hat{s}) \cup \{ lit \mid lit \in \phi \} 
- \{ \neg lit \mid lit \in \phi \}$.  \qed
\end{defi}

Basic repair operation \emph{ChangeLabel} gives the possibility 
of repairing a model by changing the labeling of a state, 
thus without inducing any changes in the structure of the model 
(number of states or transitions).  Fig.~\ref{fig:ChangeLabel} 
presents the application of \emph{ChangeLabel} in a graphical 
manner. 

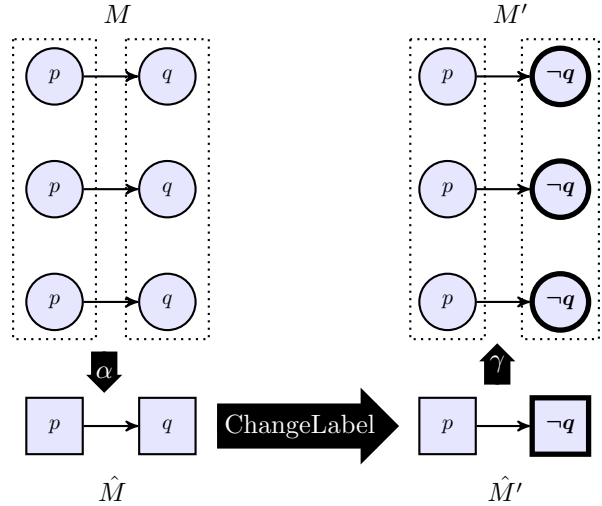
\begin{figure}[htb]
\centering          
\begin{tikzpicture}[->,>=stealth',auto,node 
distance=2cm, scale=0.55, thick, main node/.style={scale=0.75, minimum size = 
1cm, align=center,circle,fill=blue!10,draw}, abs node/.style={scale=0.75, 
minimum size = 1cm, align=center,rectangle,fill=blue!10,draw}, thick 
main node/.style={scale=0.75, minimum size = 1cm, 
align=center,circle,fill=blue!10, line width=2pt, draw}, thick abs 
node/.style={scale=0.75, 
minimum size = 1cm, align=center,rectangle,fill=blue!10, line width=2pt, draw}]

    \tikzstyle{bigarrows}=[line width=1.5mm,draw=black,-triangle 
90,postaction={draw, line width=6mm, shorten >=3mm, -}]

    \tikzstyle{smallarrows}=[line width=.5mm,draw=black,-triangle 
90,postaction={draw, line width=3.5mm, shorten >=2mm, -}]

\begin{scope}
  \node[main node] (1) {$p$};
  \node[main node] (2) [below of=1] {$p$};
  \node[main node] (3) [below of=2] {$p$};
  \node[main node] (4) [right of=1] {$q$};
  \node[main node] (5) [below of=4] {$q$};
  \node[main node] (6) [below of=5] {$q$};

   \path
     (1) edge (4)
     (2) edge (5)
     (3) edge (6);
     
     \draw[dotted] ([xshift=-3mm, yshift=9mm]1.west) rectangle ([xshift=3mm, 
yshift= -9mm]3.east);
     \draw[dotted] ([xshift=-3mm, yshift=9mm]4.west) rectangle ([xshift=3mm, 
yshift= -9mm]6.east);

   \node[font = \small] at ([xshift=-12mm, yshift=8mm]4.north) {$M$};

   \node[abs node] (7) at ([yshift=-23mm]3.south) {$p$};
   \node[abs node] (8) [right of=7] {$q$};
     
    \draw[->] (7) -- (8);

    \node[font = \small] at ([xshift=14mm, yshift=-8mm]7.south) {$\hat{M}$};

\draw [smallarrows] ([xshift=5mm, yshift=-12mm]3.east) -- ([xshift=5mm, 
yshift=-22mm]3.east) node[color = white, font=\small] at ([xshift=5mm, 
yshift=-17mm]3.east) {$\alpha$};

\draw [bigarrows] ([xshift=5mm]8.east) -- ([xshift=49mm]8.east) node[color = 
white, font=\small] at ([xshift=25mm]8.east) {ChangeLabel};
   
     \end{scope}

 \begin{scope}[xshift=9.5cm]
  \node[main node] (1) {$p$};
  \node[main node] (2) [below of=1] {$p$};
  \node[main node] (3) [below of=2] {$p$};
  \node[thick main node] (4) [right of=1] {$\boldsymbol{\neg q}$};
  \node[thick main node] (5) [below of=4] {$\boldsymbol{\neg q}$};
  \node[thick main node] (6) [below of=5] {$\boldsymbol{\neg q}$};

   \path
     (1) edge (4)
     (2) edge (5)
     (3) edge (6);
     
     \draw[dotted] ([xshift=-2mm, yshift=9mm]1.west) rectangle ([xshift=2mm, 
yshift= -9mm]3.east);
     \draw[dotted] ([xshift=-2mm, yshift=9mm]4.west) rectangle ([xshift=2mm, 
yshift= -9mm]6.east);

   \node[font = \small] at ([xshift=-12mm, yshift=8mm]4.north) {$M'$};

   \node[abs node] (7) at ([yshift=-23mm]3.south) {$p$};
   \node[thick abs node] (8) [right of=7] {$\boldsymbol{\neg q}$};
     
    \draw[->] (7) -- (8);

    \node[font = \small] at ([xshift=14mm, yshift=-8mm]7.south) {$\hat{M'}$};

\draw [smallarrows] ([xshift=5mm, yshift=-20mm]3.east) -- ([xshift=5mm, 
yshift=-10mm]3.east) node[color = white, font=\small] at ([xshift=5mm, 
yshift=-15mm]3.east) {$\gamma$};

\end{scope}

\end{tikzpicture}
\caption{\emph{ChangeLabel}: Changing the labeling of a KMTS state}
\label{fig:ChangeLabel}
\end{figure}

\begin{prop}
\label{prop:ChangeLabel}
For any $\hat{M}^{\prime} = ChangeLabel(\hat{M},\hat{s},\phi)$, it 
holds that $\hat{d}(\hat{M},\hat{M}^{\prime}) = 1$.\qed
\end{prop}

\begin{defi}
\label{def:change_label_ks}
Let $M = (S,S_{0},R,L)$ be a KS and let $\alpha(M) = (\hat{S},\hat{S_{0}}, 
R_{must}, R_{may}, \hat{L})$ be the abstract KMTS derived from $M$ 
as in Def.~\ref{def:abs_kmts}.  Also, let 
$\hat{M}^{\prime} = ChangeLabel(\alpha(M),\hat{s},\phi)$ 
for some 
$\hat{s} \in \hat{S}$ and $\phi \in 2^{LIT}$.  
The KS $M^{\prime} \in \gamma(\hat{M}^{\prime})$, 
whose structural distance $d$ from $M$ is minimized, is given by:
\begin{equation}
M^{\prime} = (S, S_{0}, R, L - L_{old} \cup L_{new}\} 
\end{equation}
where
\begin{equation} \nonumber 
L_{old} = \{ l_{old} = (s,lit_{old}) \mid s \in \gamma(\hat{s}), 
s \in S, \neg lit_{old} \not\in \phi \; \text{and} \; l_{old} \in L \}
\end{equation}
\begin{equation} \nonumber 
L_{new} = \{ l_{new} = (s,lit_{new}) \mid s \in \gamma(\hat{s}), 
s \in S, lit_{new} \in \phi \; \text{and} \; l_{new} \notin L \}
\end{equation}      
\qed
\end{defi}

\begin{prop}
\label{prop:change_label}
For $M^{\prime}$, it holds that
$1 \leq d(M,M^{\prime}) \leq |S|$.
\end{prop}
\begin{proof}
$d(M,M^{\prime}) = |S\Delta S^{\prime}| + |R\Delta R^{\prime}| + 
\frac{|\pi(L\upharpoonright_{S\cap S^{\prime}})\Delta \pi(L^{\prime}\upharpoonright_{S\cap S^{\prime}})|}{2}$.
Because $|R\Delta R^{\prime}| = 0$ and $|R\Delta R^{\prime}| = 0$, 
$d(M,M^{\prime}) = 
\frac{|\pi(L\upharpoonright_{S\cap S^{\prime}})\Delta \pi(L^{\prime}\upharpoonright_{S\cap S^{\prime}})|}{2}=
\frac{|L_{old}| + |L_{new}|}{2} = |L_{old}| = |L_{new}|$.
It holds that $L_{new} \geq 1$ and $L_{new} \leq |S|$.
So, we prove that $1 \leq d(M,M^{\prime}) \leq |S|$. 
\end{proof}

\subsubsection{Adding a new KMTS state}
\begin{defi}[AddState]
\label{def:AddState}
For a given KMTS $\hat{M} = (\hat{S},\hat{S_{0}}, R_{must}, 
R_{may}, \hat{L})$ and 
a state $\hat{s}_{n} \notin \hat{S}$, 
$AddState(\hat{M},\hat{s}_{n})$ is the KMTS 
$\hat{M^{\prime}} = (\hat{S^{\prime}}, \hat{S_{0}}, 
R_{must}, R_{may}, \hat{L^{\prime}})$
such that $\hat{S^{\prime}} = \hat{S} \cup \{\hat{s}_{n}\}$ and 
$\hat{L^{\prime}} = \hat{L} \cup \{\hat{l}_{n}\}$, where 
$\hat{l}_{n} = (\hat{s}_{n},\bot)$. \qed
\end{defi}

The most important issues for function $AddState$ is that the 
newly created abstract state $\hat{s}_{n}$ is isolated, thus 
there are no ingoing or outgoing transitions for this state, 
and additionally, the labeling of this new state is $\bot$.  
Another conclusion from Def.~\ref{def:AddState} is the fact 
that the inserted stated is not permitted to be initial.  
Application of function $AddState$ is presented graphically 
in Fig.~\ref{fig:AddState}.

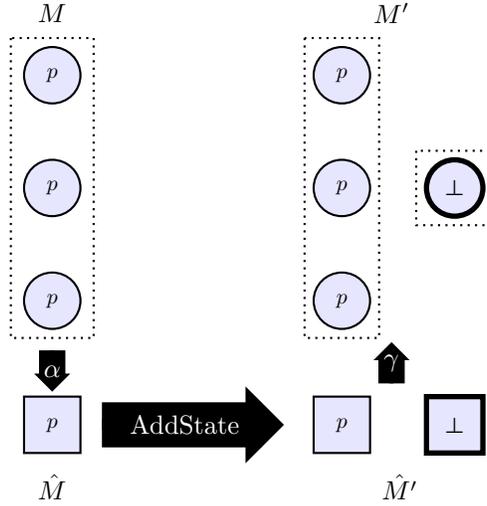
\begin{figure}[htb]
\centering          
\begin{tikzpicture}[->,>=stealth',auto,node 
distance=2cm, scale=0.55, thick, main node/.style={scale=0.75, minimum size = 
1cm, align=center,circle,fill=blue!10,draw}, abs node/.style={scale=0.75, 
minimum size = 1cm, align=center,rectangle,fill=blue!10,draw}, thick 
main node/.style={scale=0.75, minimum size = 1cm, 
align=center,circle,fill=blue!10, line width=2pt, draw}, thick abs 
node/.style={scale=0.75, 
minimum size = 1cm, align=center,rectangle,fill=blue!10, line width=2pt, draw}]

    \tikzstyle{bigarrows}=[line width=1.5mm,draw=black,-triangle 
90,postaction={draw, line width=6mm, shorten >=3mm, -}]

    \tikzstyle{smallarrows}=[line width=.5mm,draw=black,-triangle 
90,postaction={draw, line width=3.5mm, shorten >=2mm, -}]

\begin{scope}
  \node[main node] (1) {$p$};
  \node[main node] (2) [below of=1] {$p$};
  \node[main node] (3) [below of=2] {$p$};
     
     \draw[dotted] ([xshift=-3mm, yshift=9mm]1.west) rectangle ([xshift=3mm, 
yshift= -9mm]3.east);

   \node[font = \small] at ([yshift=8mm]1.north) {$M$};

   \node[abs node] (7) at ([yshift=-23mm]3.south) {$p$};

    \node[font = \small] at ([yshift=-8mm]7.south) {$\hat{M}$};

\draw [smallarrows] ([xshift=7mm, yshift=-12mm]3.west) -- ([xshift=7mm, 
yshift=-22mm]3.west) node[color = white, font=\small] at ([xshift=7mm, 
yshift=-17mm]3.west) {$\alpha$};

\draw [bigarrows] ([xshift=5mm]7.east) -- ([xshift=49mm]7.east) node[color = 
white, font=\small] at ([xshift=25mm]7.east) {AddState};
   
     \end{scope}

 \begin{scope}[xshift=7cm]
  \node[main node] (1) {$p$};
  \node[main node] (2) [below of=1] {$p$};
  \node[main node] (3) [below of=2] {$p$};
  \node[thick main node] (5) [right of=2] {$\boldsymbol{\bot}$};
     
     \draw[dotted] ([xshift=-2mm, yshift=9mm]1.west) rectangle ([xshift=2mm, 
yshift= -9mm]3.east);
     \draw[dotted] ([xshift=-2mm, yshift=9mm]5.west) rectangle ([xshift=2mm, 
yshift= -9mm]5.east);

   \node[font = \small] at ([xshift=12mm, yshift=8mm]1.north) {$M'$};

   \node[abs node] (7) at ([yshift=-23mm]3.south) {$p$};
   \node[thick abs node] (8) [right of=7] {$\boldsymbol{\bot}$};

    \node[font = \small] at ([xshift=14mm, yshift=-8mm]7.south) {$\hat{M'}$};

\draw [smallarrows] ([xshift=5mm, yshift=-20mm]3.east) -- ([xshift=5mm, 
yshift=-10mm]3.east) node[color = white, font=\small] at ([xshift=5mm, 
yshift=-15mm]3.east) {$\gamma$};

\end{scope}

\end{tikzpicture}
\caption{\emph{AddState}: Adding a new KMTS state}
\label{fig:AddState}
\end{figure}

\begin{prop}
\label{prop:AddState}
For any $\hat{M}^{\prime} = AddState(\hat{M},\hat{s}_{n})$, it 
holds that $\hat{d}(\hat{M},\hat{M}^{\prime}) = 1$.\qed
\end{prop}

\begin{defi}
\label{def:add_state_ks}
Let $M = (S,S_{0},R,L)$ be a KS and let $\alpha(M) = (\hat{S},\hat{S_{0}}, 
R_{must}, R_{may}, \hat{L})$ be the abstract KMTS derived from $M$ 
as in Def.~\ref{def:abs_kmts}.  Also, let 
$\hat{M}^{\prime} = AddState(\alpha(M),\hat{s}_{n})$ 
for some 
$\hat{s}_{n} \notin \hat{S}$.  
The KS $M^{\prime} \in \gamma(\hat{M}^{\prime})$, 
whose structural distance $d$ from $M$ is minimized is given by:
\begin{equation}
M^{\prime} = (S \cup \{s_{n}\}, S_{0}, R, L \cup \{l_{n}\})
\end{equation}
where 
$s_{n} \in \gamma(\hat{s}_{n})$ and 
$l_{n} = (s_{n},\bot)$.      \qed
\end{defi}

\begin{prop}
\label{prop:add_state}
For $M^{\prime}$, it holds that
$d(M,M^{\prime}) = 1$.
\end{prop}
\begin{proof}
$d(M,M^{\prime}) = |S\Delta S^{\prime}| + |R\Delta R^{\prime}| + 
\frac{|\pi(L\upharpoonright_{S\cap S^{\prime}})\Delta \pi(L^{\prime}\upharpoonright_{S\cap S^{\prime}})|}{2}$.
Because $|R\Delta R^{\prime}| = 0$ and $|\pi(L\upharpoonright_{S\cap S^{\prime}})\Delta \pi(L^{\prime}\upharpoonright_{S\cap S^{\prime}})| = 0$, 
$d(M,M^{\prime}) = |S\Delta S^{\prime}| = 
|S - S^{\prime}| + |S^{\prime} - S| = 0 + |\{s_{n}\}| = 1$.
So, we proved that $d(M,M^{\prime}) = 1$. 
\end{proof}

\subsubsection{Removing a disconnected KMTS state}
\begin{defi}[RemoveState]
\label{def:RemoveState}
For a given KMTS $\hat{M} = (\hat{S},\hat{S_{0}}, R_{must}, 
R_{may}, \hat{L})$ and 
a state $\hat{s}_{r} \in \hat{S}$ such that $\forall \hat{s} \in \hat{S} : 
(\hat{s},\hat{s}_{r}) \not\in R_{may} \, \wedge \, (\hat{s}_{r},\hat{s}) 
\not\in R_{may}$,   
$RemoveState(\hat{M},\hat{s}_{r})$ is the KMTS 
$\hat{M^{\prime}} = (\hat{S^{\prime}}, \hat{S_{0}^{\prime}}, 
R_{must}, R_{may}, \hat{L^{\prime}})$
such that $\hat{S^{\prime}} = \hat{S} - \{\hat{s}_{r}\}$, 
$\hat{S_{0}^{\prime}} = \hat{S_{0}} - \{\hat{s}_{r}\}$ and 
$\hat{L^{\prime}} = \hat{L} - \{\hat{l}_{r}\}$, where 
$\hat{l}_{r} = (\hat{s}_{r},lit) \in \hat{L}$. \qed
\end{defi}

From Def.~\ref{def:RemoveState}, it is clear that the state 
being removed should be isolated, thus there are not any may- 
or must-transitions from and to this state.  This means that 
before using \emph{RemoveState} to an abstract state, all its 
ingoing or outgoing must have been removed by using other 
basic repair operations.  \emph{RemoveState} are also used  
for the elimination of dead-end states, when such states 
arise during the repair process.  Fig.~\ref{fig:RemoveState} 
presents the application of \emph{RemoveState} in a graphical 
manner.    

\begin{figure}[htb]
\centering          
\begin{tikzpicture}[->,>=stealth',auto,node 
distance=2cm, scale=0.55, thick, main node/.style={scale=0.75, minimum size = 
1cm, align=center,circle,fill=blue!10,draw}, abs node/.style={scale=0.75, 
minimum size = 1cm, align=center,rectangle,fill=blue!10,draw}]

    \tikzstyle{bigarrows}=[line width=1.5mm,draw=black,-triangle 
90,postaction={draw, line width=6mm, shorten >=3mm, -}]

    \tikzstyle{smallarrows}=[line width=.5mm,draw=black,-triangle 
90,postaction={draw, line width=3.5mm, shorten >=2mm, -}]

\begin{scope}
  \node[main node] (1) {$p$};
  \node[main node] (2) [below of=1] {$p$};
  \node[main node] (3) [below of=2] {$p$};
  \node[main node] (4) [right of=1] {$q$};
  \node[main node] (5) [below of=4] {$q$};
  \node[main node] (6) [below of=5] {$q$};
     
     \draw[dotted] ([xshift=-3mm, yshift=9mm]1.west) rectangle ([xshift=3mm, 
yshift= -9mm]3.east);
     \draw[dotted] ([xshift=-3mm, yshift=9mm]4.west) rectangle ([xshift=3mm, 
yshift= -9mm]6.east);

   \node[font = \small] at ([xshift=-12mm, yshift=8mm]4.north) {$M$};

   \node[abs node] (7) at ([yshift=-23mm]3.south) {$p$};
   \node[abs node] (8) [right of=7] {$q$};
    
    \node[font = \small] at ([xshift=14mm, yshift=-8mm]7.south) {$\hat{M}$};

\draw [smallarrows] ([xshift=5mm, yshift=-12mm]3.east) -- ([xshift=5mm, 
yshift=-22mm]3.east) node[color = white, font=\small] at ([xshift=5mm, 
yshift=-17mm]3.east) {$\alpha$};

\draw [bigarrows] ([xshift=3mm]8.east) -- ([xshift=45mm]8.east) node[color = 
white, font=\small] at ([xshift=23mm]8.east) {RemoveState};
   
     \end{scope}

 \begin{scope}[xshift=9cm]
  \node[main node] (1) {$p$};
  \node[main node] (2) [below of=1] {$p$};
  \node[main node] (3) [below of=2] {$p$};
  \node[main node] (4) [right of=1] {$q$};
  \node[main node] (5) [below of=4] {$q$};
  \node[main node] (6) [below of=5] {$q$};
     
     \draw[dotted] ([xshift=-2mm, yshift=9mm]1.west) rectangle ([xshift=2mm, 
yshift= -9mm]3.east);
     \draw[dotted] ([xshift=-2mm, yshift=9mm]4.west) rectangle ([xshift=2mm, 
yshift= -9mm]6.east);

   \node[font = \small] at ([xshift=-12mm, yshift=8mm]4.north) {$M'$};

   \node[abs node] (7) at ([yshift=-23mm]3.south) {$p$};
   \node[abs node] (8) [right of=7] {$q$};
    
    \node[font = \small] at ([xshift=14mm, yshift=-8mm]7.south) {$\hat{M'}$};

\draw [smallarrows] ([xshift=5mm, yshift=-20mm]3.east) -- ([xshift=5mm, 
yshift=-10mm]3.east) node[color = white, font=\small] at ([xshift=5mm, 
yshift=-15mm]3.east) {$\gamma$}; 

 \node[black, auto=false, cross out, minimum size=6mm, -, draw] at (4) {};
 \node[black, auto=false, cross out, minimum size=6mm, -, draw] at (5) {};
 \node[black, auto=false, cross out, minimum size=6mm, -, draw] at (6) {};
 \node[black, auto=false, cross out, minimum size=6mm, -, draw] at (8) {};

\end{scope}

\end{tikzpicture}
\caption{\emph{RemoveState}: Removing a disconnected KMTS state}
\label{fig:RemoveState}
\end{figure}
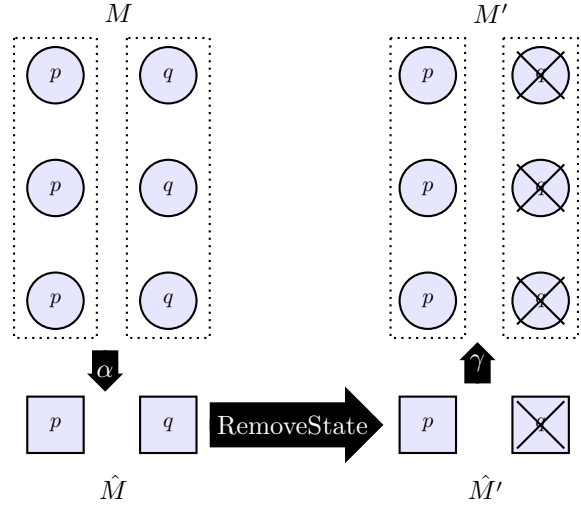

\begin{prop}
\label{prop:RemoveState}
For any $\hat{M}^{\prime} = RemoveState(\hat{M},\hat{s}_{r})$, it 
holds that $\hat{d}(\hat{M},\hat{M}^{\prime}) = 1$.\qed
\end{prop}

\begin{defi}
\label{def:remove_state_ks}
Let $M = (S,S_{0},R,L)$ be a KS and let $\alpha(M) = (\hat{S},\hat{S_{0}}, 
R_{must}, R_{may}, \hat{L})$ be the abstract KMTS derived from $M$ 
as in Def.~\ref{def:abs_kmts}.  Also, let 
$\hat{M}^{\prime} = RemoveState(\alpha(M),\hat{s}_{r})$ 
for some 
$\hat{s}_{r} \in \hat{S}$ with 
$\hat{l}_{r} = (\hat{s}_{r},lit) \in \hat{L}$.  
The KS $M^{\prime} \in \gamma(\hat{M}^{\prime})$, 
whose structural distance $d$ from $M$ is minimized, is given by:
\begin{equation}
M^{\prime} = (S^{\prime}, S_{0}^{\prime}, 
R^{\prime}, L^{\prime}) \mbox{ s.t. } 
S^{\prime} = S - S_{r},  
S_{0}^{\prime} = S_{0} - S_{r}, R^{\prime} = R, 
L^{\prime} = L - L_{r} 
\end{equation}
where 
$S_{r} = \{ s_{r} \mid s_{r} \in S \mbox{ and } s_{r} 
\in \gamma(\hat{s}_{r}) \}$ and 
$L_{r} = \{ l_{r} = (s_{r},lit) \mid l_{r} \in L \}$. \qed
\end{defi}

\begin{prop}
\label{prop:remove_state}
For $M^{\prime}$, it holds that
$1 \leq d(M,M^{\prime}) \leq |S|$.
\end{prop}
\begin{proof}
$d(M,M^{\prime}) = |S\Delta S^{\prime}| + |R\Delta R^{\prime}| + 
\frac{|\pi(L\upharpoonright_{S\cap S^{\prime}})\Delta \pi(L^{\prime}\upharpoonright_{S\cap S^{\prime}})|}{2}$.
Because $|R\Delta R^{\prime}| = 0$ and $|\pi(L\upharpoonright_{S\cap S^{\prime}})\Delta \pi(L^{\prime}\upharpoonright_{S\cap S^{\prime}})| = 0$, 
$d(M,M^{\prime}) = |S\Delta S^{\prime}| = 
|S - S^{\prime}| + |S^{\prime} - S| = |S_{r}| + 0 = |S_{r}|$.
It holds that $|S_{r}| \geq 1$ and $|S_{r}| \leq |S|$.
So, we proved that $1 \leq d(M,M^{\prime}) \leq |S|$. 
\end{proof}

\subsubsection{Minimality Of Changes Ordering For Basic Repair Operations}
\label{subsec:minimal_basic_ops}

The distance metric $d$ of Def.~\ref{def:metric_space} reflects the need to quantify structural changes in the concrete model that are attributed to model repair steps applied to the abstract KMTS.  Every such repair step implies multiple structural changes in the concrete KSs, due to the use of abstraction.  In this context, 
our distance metric is an essential means for the effective application of the abstraction in the repair process.

Based on the upper bound given by Prop.~\ref{prop:add_must} 
and all the respective results for the other basic repair 
operations, we introduce the partial ordering shown in 
Fig.~\ref{fig:order_basic_ops}.  This ordering is used 
in our \emph{AbstractRepair} algorithm to 
\emph{heuristically} select at each step the basic repair 
operation that \textit{generates the KSs with the least changes}.
When it is possible to apply more than one basic repair operation with the same
upper bound, our algorithm successively uses them until a repair solution is found,   
in an order based on the computational complexity of their application. \enlargethispage{2\baselineskip}
 
If instead of our approach, all possible repaired KSs were 
checked to identify the basic repair operation with the 
minimum changes, this would defeat the purpose of
using abstraction.  The reason is that such a check 
inevitably would depend on the size of concrete KSs.  

\begin{figure}[t]
\centering
\includegraphics[scale=1]{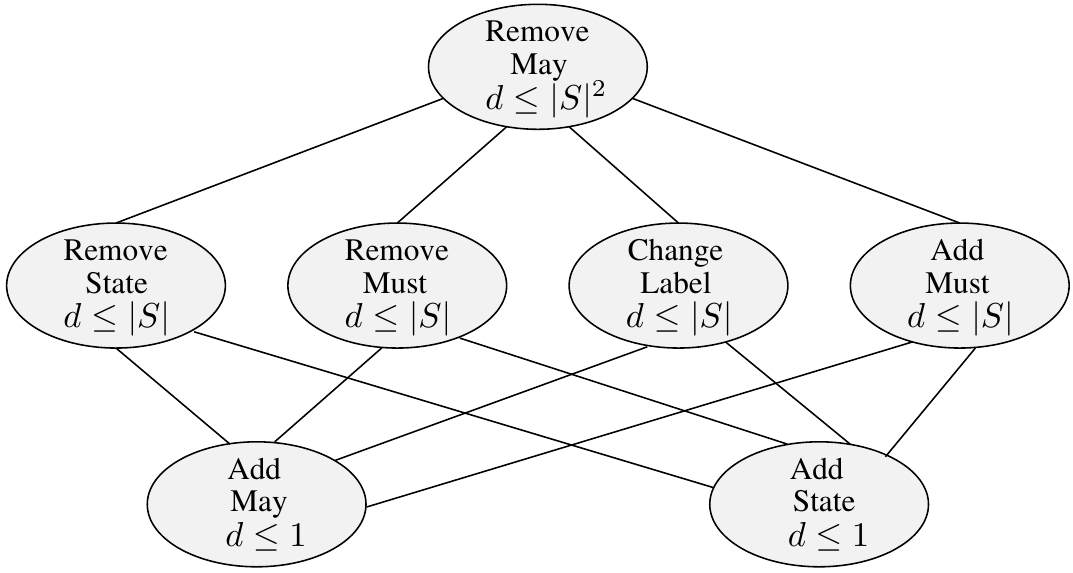}
\caption{Minimality of changes ordering of the set of basic repair operations}
\label{fig:order_basic_ops}
\end{figure}

\section{The Abstract Model Repair Algorithm}
\label{sec:alg}
The \emph{AbstractRepair} algorithm used in Step 
3 of our repair process is a recursive, syntax-directed algorithm, 
where the syntax for the property $\phi$ in question is that of CTL.  The same approach is followed by the SAT 
model checking algorithm in~\cite{HR04} and a number 
of model repair solutions applied to concrete KSs~\cite{ZD08,CR09}. In our case, we aim to the repair of an abstract KMTS by successively calling 
primitive repair functions that handle atomic formulas, logical connectives and CTL operators. At each step, the repair with the least changes for the concrete model among all the possible repairs is applied first.      

\begin{algorithm}[t]
\caption{AbstractRepair}
\label{alg:main}        	
\begin{algorithmic}[1]
\renewcommand{\algorithmicrequire}{\textbf{Input:}}
\renewcommand{\algorithmicensure}{\textbf{Output:}}              		
\REQUIRE $\hat{M} = (\hat{S}, \hat{S}_{0}, R_{must}, R_{may}, 
\hat{L})$, $\hat{s} \in \hat{S}$, a CTL property $\phi$ in PNF for 
which $(\hat{M},\hat{s}) \not\models \phi$, and a set of constraints 
$C = \{ (\hat{s}_{c_{1}},\phi_{c_{1}}), 
(\hat{s}_{c_{2}},\phi_{c_{2}}), ..., (\hat{s}_{c_{n}},\phi_{c_{n}}) \}$ where 
$\hat{s}_{c_{i}} \in \hat{S}$ and $\phi_{c_{i}}$ is a CTL formula. 
\ENSURE $\hat{M^{\prime}} = (\hat{S^{\prime}}, \hat{S_{0}^{\prime}}, 
R_{must}^{\prime}, R_{may}^{\prime}, \hat{L^{\prime}})$ 
and 
$(\hat{M^{\prime}},\hat{s}) \models \phi$ or FAILURE.  
\IF {$\phi$ is $false$}
\RETURN FAILURE
\ELSIF {$\phi \in LIT$}
\RETURN $AbstractRepair_{ATOMIC}(\hat{M},\hat{s},\phi,C)$
\ELSIF {$\phi$ is $\phi_{1} \wedge \phi_{2}$}
\RETURN $AbstractRepair_{AND}(\hat{M},\hat{s},\phi,C)$
\ELSIF {$\phi$ is $\phi_{1} \vee \phi_{2}$}
\RETURN $AbstractRepair_{OR}(\hat{M},\hat{s},\phi,C)$
\ELSIF {$\phi$ is $OPER\phi_{1}$}
\RETURN $AbstractRepair_{OPER}(\hat{M},\hat{s},\phi,C)$
\STATE where $OPER \in \{AX,EX,AU,EU,AF,EF,AG,EG\}$
\ENDIF
\end{algorithmic}
\end{algorithm}

The main routine of \emph{AbstractRepair} is presented in 
Algorithm~\ref{alg:main}. If the property $\phi$ is 
not in Positive Normal Form, i.e. negations are applied only to 
atomic propositions, then we transform it into such a form before 
applying Algorithm~\ref{alg:main}.
    
An initially empty set of constraints $C = \{ (\hat{s}_{c_{1}},\phi_{c_{1}}), 
(\hat{s}_{c_{2}},\phi_{c_{2}}), ..., (\hat{s}_{c_{n}},\phi_{c_{n}}) \}$ 
is passed as an argument in the successive recursive calls of 
\emph{AbstractRepair}. We note that these constraints can also specify 
\emph{existing} properties that should be preserved during repair. 
If $C$ is not empty, then for the returned KMTS $\hat{M}^{\prime}$, 
it holds that $(\hat{M^{\prime}},\hat{s}_{c_{i}}) \models \phi_{c_{i}}$ 
for all $(\hat{s}_{c_{i}},\phi_{c_{i}}) \in C$.  For brevity, we denote 
this with $\hat{M}^{\prime} \models C$.  
We use $C$ in order to handle conjunctive formulas 
of the form $\phi = \phi_{1} \wedge \phi_{2}$ 
for some state $\hat{s}$.  In this case, 
\emph{AbstractRepair} is called for the KMTS $\hat{M}$ 
and property $\phi_{1}$ with $C = \{ (\hat{s},\phi_{2}) \}$.  
The same is repeated for property $\phi_{2}$ with 
$C = \{ (\hat{s},\phi_{1}) \}$ and the two results 
are combined appropriately.    

For any CTL formula $\phi$ and KMTS state $\hat{s}$, 
\emph{AbstractRepair} either outputs a KMTS 
$\hat{M}^{\prime}$ for which 
$(\hat{M^{\prime}},\hat{s}) \models \phi$ or else 
returns FAILURE, if such a model cannot be found.  
This is the case when the algorithm handles 
conjunctive formulas and a KMTS that 
simultaneously satisfies all conjuncts cannot 
be found.   

\begin{algorithm}[htb]       
\floatname{algorithm}{Algorithm}
\caption{$AbstractRepair_{ATOMIC}$}  
\label{alg:ATOMIC}        	
\begin{algorithmic}[1]
\renewcommand{\algorithmicrequire}{\textbf{Input:}}
\renewcommand{\algorithmicensure}{\textbf{Output:}}              		
\REQUIRE $\hat{M} = (\hat{S}, \hat{S}_{0}, R_{must}, R_{may}, 
\hat{L})$, $\hat{s} \in \hat{S}$, 
a CTL property $\phi$ where $\phi$ is an atomic formula for which 
$(\hat{M},\hat{s}) \not\models \phi$, and a set of constraints 
$C = \{ (\hat{s}_{c_{1}},\phi_{c_{1}}), 
(\hat{s}_{c_{2}},\phi_{c_{2}}), ..., (\hat{s}_{c_{n}},\phi_{c_{n}}) \}$ where 
$\hat{s}_{c_{i}} \in \hat{S}$ and $\phi_{c_{i}}$ is a CTL formula.
\ENSURE $\hat{M^{\prime}} = (\hat{S^{\prime}}, \hat{S_{0}^{\prime}}, 
R_{must}^{\prime}, R_{may}^{\prime}, \hat{L^{\prime}})$ and 
$(\hat{M^{\prime}},\hat{s}) \models \phi$ or FAILURE. 
\STATE $\hat{M^{\prime}} := ChangeLabel(\hat{M},\hat{s},\phi)$
\IF {$\hat{M^{\prime}} \models C$} 
	\RETURN $\hat{M^{\prime}}$
\ELSE
	\RETURN FAILURE
\ENDIF
\end{algorithmic}
\end{algorithm}

\begin{algorithm}[htb]      
\floatname{algorithm}{Algorithm}
\caption{$AbstractRepair_{OR}$}   
\label{alg:OR}       	
\begin{algorithmic}[1]
\renewcommand{\algorithmicrequire}{\textbf{Input:}}
\renewcommand{\algorithmicensure}{\textbf{Output:}} 
\REQUIRE $\hat{M} = (\hat{S}, \hat{S}_{0}, R_{must}, R_{may}, 
\hat{L})$, $\hat{s} \in \hat{S}$, 
a CTL property $\phi = \phi_{1} \vee \phi_{2}$ for which  
$(\hat{M},\hat{s}) \not\models \phi$, and a set of constraints 
$C = ( (\hat{s}_{c_{1}},\phi_{c_{1}}), 
(\hat{s}_{c_{2}},\phi_{c_{2}}), ..., (\hat{s}_{c_{n}},\phi_{c_{n}}) )$ where 
$\hat{s}_{c_{i}} \in \hat{S}$ and $\phi_{c_{i}}$ is a CTL formula. 
\ENSURE $\hat{M^{\prime}} = (\hat{S^{\prime}}, \hat{S_{0}^{\prime}}, 
R_{must}^{\prime}, R_{may}^{\prime}, \hat{L^{\prime}})$, 
$\hat{s} \in \hat{S^{\prime}}$ and  
$(\hat{M^{\prime}},\hat{s}) \models \phi$ or FAILURE.  
\STATE $RET_{1} := AbstractRepair(\hat{M},\hat{s},\phi_{1},C)$
\STATE $RET_{2} := AbstractRepair(\hat{M},\hat{s},\phi_{2},C)$
\IF { $RET_{1} \neq FAILURE$ \&\& $RET_{2} \neq FAILURE $  }
	\STATE $\hat{M}_{1} := RET_{1}$
	\STATE $\hat{M}_{2} := RET_{2}$
	\STATE $\hat{M^{\prime}} := MinimallyChanged(\hat{M},\hat{M_{1}},\hat{M_{2}})$
\ELSIF { $RET_{1} \neq FAILURE$ }
	\STATE $\hat{M^{\prime}} := RET_{1}$
\ELSIF { $RET_{2} \neq FAILURE$ }
	\STATE $\hat{M^{\prime}} := RET_{2}$
\ELSE
	\RETURN FAILURE
\ENDIF
\RETURN $\hat{M}^{\prime}$
\end{algorithmic}
\end{algorithm}

\begin{algorithm}[htb]      
\floatname{algorithm}{Algorithm}
\caption{$AbstractRepair_{AND}$}   
\label{alg:AND}       	
\begin{algorithmic}[1]
\renewcommand{\algorithmicrequire}{\textbf{Input:}}
\renewcommand{\algorithmicensure}{\textbf{Output:}} 
\REQUIRE $\hat{M} = (\hat{S}, \hat{S}_{0}, R_{must}, R_{may}, 
\hat{L})$, $\hat{s} \in \hat{S}$, 
a CTL property $\phi = \phi_{1} \wedge \phi_{2}$ for which  
$(\hat{M},\hat{s}) \not\models \phi$, and a set of constraints 
$C = ( (\hat{s}_{c_{1}},\phi_{c_{1}}), 
(\hat{s}_{c_{2}},\phi_{c_{2}}), ..., (\hat{s}_{c_{n}},\phi_{c_{n}}) )$ where 
$\hat{s}_{c_{i}} \in \hat{S}$ and $\phi_{c_{i}}$ is a CTL formula. 
\ENSURE $\hat{M^{\prime}} = (\hat{S^{\prime}}, \hat{S_{0}^{\prime}}, 
R_{must}^{\prime}, R_{may}^{\prime}, \hat{L^{\prime}})$, 
$\hat{s} \in \hat{S^{\prime}}$ and  
$(\hat{M^{\prime}},\hat{s}) \models \phi$ or FAILURE. 
\STATE $RET_{1} := AbstractRepair(\hat{M},\hat{s},\phi_{1},C)$
\STATE $RET_{2} := AbstractRepair(\hat{M},\hat{s},\phi_{2},C)$
\STATE $C_{1} := C \cup \{ (\hat{s},\phi_{1}) \}$, $C_{2} := C \cup \{(\hat{s},\phi_{2})\}$ 
\STATE $RET_{1}^{\prime} := FAIURE$, $RET_{2}^{\prime} := FAIURE$
\IF { $RET_{1} \neq FAILURE$ }
	\STATE $\hat{M}_{1} := RET_{1}$
	\STATE $RET_{1}^{\prime} := AbstractRepair(\hat{M}_{1},\hat{s},\phi_{2},C_{1})$
	\IF { $RET_{1}^{\prime} \neq FAILURE$ }
		\STATE $\hat{M}_{1}^{\prime} := RET_{1}^{\prime}$
	\ENDIF
\ENDIF
\IF { $RET_{2} \neq FAILURE$ }
	\STATE $\hat{M}_{2} := RET_{2}$
	\STATE $RET_{2}^{\prime} := AbstractRepair(\hat{M}_{2},\hat{s},\phi_{1},C_{2})$
	\IF { $RET_{2}^{\prime} \neq FAILURE$ }
		\STATE $\hat{M}_{2}^{\prime} := RET_{2}^{\prime}$
	\ENDIF
\ENDIF
\IF { $RET_{1}^{\prime} \neq FAILURE$ \&\& $RET_{2}^{\prime} \neq FAILURE $  }
	\STATE $\hat{M^{\prime}} := MinimallyChanged(\hat{M},\hat{M}_{1}^{\prime},\hat{M}_{2}^{\prime})$
\ELSIF { $RET_{1}^{\prime} \neq FAILURE$ }
	\STATE $\hat{M^{\prime}} := RET_{1}^{\prime}$
\ELSIF { $RET_{2}^{\prime} \neq FAILURE$ }
	\STATE $\hat{M^{\prime}} := RET_{2}^{\prime}$
\ELSE
	\RETURN FAILURE
\ENDIF
\RETURN $\hat{M}^{\prime}$
\end{algorithmic}
\end{algorithm}

\begin{algorithm}[htb]       
\floatname{algorithm}{Algorithm}
\caption{$AbstractRepair_{AG}$} 
\label{alg:AG}         	
\begin{algorithmic}[1]
\renewcommand{\algorithmicrequire}{\textbf{Input:}}
\renewcommand{\algorithmicensure}{\textbf{Output:}}    
\REQUIRE $\hat{M} = (\hat{S}, \hat{S}_{0}, R_{must}, R_{may}, 
\hat{L})$, $\hat{s} \in \hat{S}$, 
a CTL property $\phi = AG\phi_{1}$ for which 
$(\hat{M},\hat{s}) \not\models \phi$, and a set of constraints 
$C = \{ (\hat{s}_{c_{1}},\phi_{c_{1}}), 
(\hat{s}_{c_{2}},\phi_{c_{2}}), ..., (\hat{s}_{c_{n}},\phi_{c_{n}}) \}$ where 
$\hat{s}_{c_{i}} \in \hat{S}$ and $\phi_{c_{i}}$ is a CTL formula.
\ENSURE $\hat{M^{\prime}} = (\hat{S^{\prime}}, \hat{S_{0}^{\prime}}, 
R_{must}^{\prime}, R_{may}^{\prime}, \hat{L^{\prime}})$ and 
$(\hat{M^{\prime}},\hat{s}) \models \phi$ or FAILURE.  
\IF {$(\hat{M},\hat{s}) \not\models \phi_{1}$}
	\STATE $RET := AbstractRepair(\hat{M},\hat{s},\phi_{1},C)$
	\IF { $RET == FAILURE$ }
		\RETURN FAILURE
	\ELSE
		\STATE $\hat{M^{\prime}} := RET$
	\ENDIF
\ELSE
	\STATE $\hat{M^{\prime}} := \hat{M}$
\ENDIF
\FORALL{ reachable states $\hat{s}_{k}$ through may-transitions from $\hat{s}$ 
	such that $(\hat{M^{\prime}},\hat{s}_{k}) \not\models \phi_{1}$ }
	\STATE $RET := AbstractRepair(\hat{M^{\prime}},\hat{s}_{k},\phi_{1},C)$
	\IF { $RET == FAILURE$ }
		\RETURN FAILURE
	\ELSE
		\STATE $\hat{M^{\prime}} := RET$
	\ENDIF
\ENDFOR
\IF { $\hat{M^{\prime}} \models C$ }
	\RETURN $\hat{M^{\prime}}$
\ENDIF
\RETURN FAILURE
\end{algorithmic}
\end{algorithm}
\subsection{Primitive Functions}
\label{subsec:alg_prim_func}
Algorithm~\ref{alg:ATOMIC} describes $AbstractRepair_{ATOMIC}$, 
which for a simple atomic formula, updates the labeling of 
the input state with the given atomic proposition.   
Disjunctive formulas are handled by repairing the disjunct leading to the 
minimum change (Algorithm~\ref{alg:OR}), while conjunctive formulas are handled 
by the algorithm with the use of constraints (Algorithm~\ref{alg:AND}).

Algorithm~\ref{alg:AG} describes the primitive function 
$AbstractRepair_{AG}$ which is called when 
$\phi = AG\phi_{1}$.  If $AbstractRepair_{AG}$ is called 
for a state $\hat{s}$, it recursively calls \emph{AbstractRepair} 
for $\hat{s}$ and for all reachable states through may-transitions 
from $\hat{s}$ which do not satisfy $\phi_{1}$.  The resulting KMTS 
$\hat{M}^{\prime}$ is returned, if it does not violate any 
constraint in $C$.    

\begin{algorithm}[htb]  
\floatname{algorithm}{Algorithm}
\caption{$AbstractRepair_{EX}$}     
\label{alg:EX}     	
\begin{algorithmic}[1]
\renewcommand{\algorithmicrequire}{\textbf{Input:}}
\renewcommand{\algorithmicensure}{\textbf{Output:}}       
\REQUIRE $\hat{M} = (\hat{S}, \hat{S}_{0}, R_{must}, R_{may}, 
\hat{L})$, $\hat{s} \in \hat{S}$, 
a CTL property $\phi = EX\phi_{1}$ for which  
$(\hat{M},\hat{s}) \not\models \phi$, and a set of constraints 
$C = \{ (\hat{s}_{c_{1}},\phi_{c_{1}}), 
(\hat{s}_{c_{2}},\phi_{c_{2}}), ..., (\hat{s}_{c_{n}},\phi_{c_{n}}) \}$ where 
$\hat{s}_{c_{i}} \in \hat{M}$ and $\phi_{c_{i}}$ is a CTL formula. 
\ENSURE $\hat{M^{\prime}} = (\hat{S^{\prime}}, \hat{S_{0}^{\prime}}, 
R_{must}^{\prime}, R_{may}^{\prime}, \hat{L^{\prime}})$ and 
$(\hat{M^{\prime}},\hat{s}) \models \phi$ or FAILURE. 
\IF {there exists $\hat{s}_{1} \in \hat{S}$ such that $(\hat{M},\hat{s}_{1}) \models \phi_{1}$}
	\FORALL {$\hat{s}_{i} \in \hat{S}$ such that $(\hat{M},\hat{s}_{i}) \models \phi_{1}$}
		\STATE $\hat{r}_{i} := (\hat{s},\hat{s}_{i})$, 
				$\hat{M^{\prime}} := AddMust(\hat{M},\hat{r}_{i})$
		\IF {$\hat{M^{\prime}} \models C$}
			\RETURN $\hat{M^{\prime}}$
		\ENDIF
	\ENDFOR
\ELSE
	\FORALL{direct must-reachable states $\hat{s}_{i}$ from $\hat{s}$ such that $(\hat{M},\hat{s}_{i}) \not\models \phi_{1}$}
		\STATE $RET := AbstractRepair(\hat{M},\hat{s}_{i},\phi_{1},C)$
		\IF { $RET \neq FAILURE$ }
			\STATE $\hat{M^{\prime}} := RET$
			\RETURN $\hat{M^{\prime}}$
		\ENDIF
	\ENDFOR	
	\STATE $\hat{M^{\prime}} := AddState(\hat{M},\hat{s}_{n})$, 
			$\hat{r}_{n} := (\hat{s},\hat{s}_{n})$, 
			$\hat{M^{\prime}} := AddMust(\hat{M^{\prime}},\hat{r}_{n})$
	\STATE $\hat{r}_{n} := (\hat{s}_{n},\hat{s}_{n})$
	\STATE $\hat{M^{\prime}} := AddMay(\hat{M^{\prime}},\hat{r}_{n})$
	\STATE $RET := AbstractRepair(\hat{M^{\prime}},\hat{s}_{n},\phi_{1},C)$
	\IF { $RET \neq FAILURE$ } 
		\STATE $\hat{M^{\prime}} := RET$
		\RETURN $\hat{M^{\prime}}$
	\ENDIF
\ENDIF
\RETURN FAILURE
\end{algorithmic}
\end{algorithm}

$AbstractRepair_{EX}$ presented in Algorithm~\ref{alg:EX} 
is the primitive function for handling properties of the form 
$EX\phi_{1}$ for some state $\hat{s}$.  At first, 
$AbstractRepair_{EX}$ attempts to repair the KMTS by adding a 
must-transition from $\hat{s}$ to a state that 
satisfies property $\phi_{1}$.  If a repaired KMTS is not found,
then \emph{AbstractRepair} is recursively called for an 
immediate successor of $\hat{s}$ through a must-transition,  
such that $\phi_{1}$ is not satisfied.   
If a constraint in $C$ is violated, then (i) a new state 
is added, (ii) \emph{AbstractRepair} is called for the new state and 
(iii) a must-transition from $\hat{s}$ to the new state 
is added.  The resulting KMTS is returned by the algorithm if all 
constraints of $C$ are satisfied.   

\begin{algorithm}[htb]  
\floatname{algorithm}{Algorithm}
\caption{$AbstractRepair_{AX}$}     
\label{alg:AX}     	
\begin{algorithmic}[1]
\renewcommand{\algorithmicrequire}{\textbf{Input:}}
\renewcommand{\algorithmicensure}{\textbf{Output:}}       
\REQUIRE $\hat{M} = (\hat{S}, \hat{S}_{0}, R_{must}, R_{may}, 
\hat{L})$, $\hat{s} \in \hat{S}$, 
a CTL property $\phi = AX\phi_{1}$ for which  
$(\hat{M},\hat{s}) \not\models \phi$, and a set of constraints 
$C = \{ (\hat{s}_{c_{1}},\phi_{c_{1}}), 
(\hat{s}_{c_{2}},\phi_{c_{2}}), ..., (\hat{s}_{c_{n}},\phi_{c_{n}}) \}$ where 
$\hat{s}_{c_{i}} \in \hat{M}$ and $\phi_{c_{i}}$ is a CTL formula. 
\ENSURE $\hat{M^{\prime}} = (\hat{S^{\prime}}, \hat{S_{0}^{\prime}}, 
R_{must}^{\prime}, R_{may}^{\prime}, \hat{L^{\prime}})$ and 
$(\hat{M^{\prime}},\hat{s}) \models \phi$ or FAILURE.
\STATE $\hat{M^{\prime}} := \hat{M}$
\STATE $RET := FAILURE$
\FORALL {direct may-reachable states $\hat{s}_{i}$   
from $\hat{s}$ with $(\hat{s},\hat{s}_{i}) \in R_{may}$}
	\IF {$(\hat{M^{\prime}},\hat{s}_{i}) \not\models \phi_{1}$}
		\STATE $RET := AbstractRepair(\hat{M^{\prime}},\hat{s}_{i},\phi_{1},C)$
		\IF {$RET == FAILURE$}
			\STATE BREAK
		\ENDIF
		\STATE $\hat{M^{\prime}} := RET$
	\ENDIF
\ENDFOR
\IF {$RET \neq FAILURE$}
	\RETURN $\hat{M^{\prime}}$
\ENDIF
\STATE $\hat{M^{\prime}} := \hat{M}$
\FORALL {direct may-reachable states $\hat{s}_{i}$   
from $\hat{s}$ with $\hat{r}_{i} := (\hat{s},\hat{s}_{i}) \in R_{may}$}
	\IF {$(\hat{M^{\prime}},\hat{s}_{i}) \not\models \phi_{1}$}
		\STATE $\hat{M^{\prime}} := RemoveMay(\hat{M^{\prime}},\hat{r}_{i})$
	\ENDIF
\ENDFOR
\IF {there exists direct may-reachable state $\hat{s}_{1}$   
from $\hat{s}$ such that $(\hat{s},\hat{s}_{1}) \in R_{may}$}
	\IF {$\hat{M^{\prime}} \models C$}
		\RETURN $\hat{M^{\prime}}$
	\ENDIF
\ELSE
	\FORALL {$\hat{s}_{j} \in \hat{S}$ such that $(\hat{M^{\prime}},\hat{s}_{j}) \models \phi_{1}$}
		\STATE $\hat{r}_{j} := (\hat{s},\hat{s}_{j})$, 
				$\hat{M^{\prime}} := AddMay(\hat{M^{\prime}},\hat{r}_{j})$
		\IF {$\hat{M^{\prime}} \models C$}
			\RETURN $\hat{M^{\prime}}$
		\ENDIF
	\ENDFOR
	\STATE $\hat{M^{\prime}} := AddState(\hat{M},\hat{s}_{n})$
	\IF {$\hat{s}_{n}$ is a dead-end state}
		\STATE $\hat{r}_{n} := (\hat{s}_{n},\hat{s}_{n})$, 
				$\hat{M^{\prime}} := AddMay(\hat{M^{\prime}},\hat{r}_{n})$
	\ENDIF
	\STATE $RET := AbstractRepair(\hat{M^{\prime}},\hat{s}_{n},\phi_{1},C)$
	\IF { $RET \neq FAILURE$ }
		\STATE $\hat{M^{\prime}} := RET$, 
				$\hat{r}_{n} := (\hat{s},\hat{s}_{n})$, 
				$\hat{M^{\prime}} := AddMay(\hat{M^{\prime}},\hat{r}_{n})$
		\IF {$\hat{M^{\prime}} \models C$}
			\RETURN $\hat{M^{\prime}}$
		\ENDIF
	\ENDIF
\ENDIF
\RETURN FAILURE
\end{algorithmic}
\end{algorithm}

Algorithm~\ref{alg:AX} presents primitive function 
$AbstractRepair_{AX}$ which is used when $\phi = AX\phi_{1}$.  
Firstly, $AbstractRepair_{AX}$ tries to repair the KMTS by 
applying $AbstractRepair$ for all direct may-successors 
$\hat{s}_{i}$ of $\hat{s}$ which do not satisfy property 
$\phi_{1}$, and in the case that all the constraints are 
satisfied the new KMTS is returned by the function.  
If such states do not exist or a constraint is violated, 
all may-transitions $(\hat{s},\hat{s}_{i})$ for which 
$(\hat{M},\hat{s}_{i}) \not\models \phi_{1}$, are removed.    
If there are states $\hat{s}_{i}$ such that 
$r_{m} := (\hat{s},\hat{s}_{i}) \in R_{may}$ and all constraints 
are satisfied then a repaired KMTS has been produced and it 
is returned by the function.  
Otherwise, a repaired KMTS results by the application of 
$AddMay$ from $\hat{s}$ to all states $\hat{s}_{j}$ which 
satisfy $\phi_{1}$.  
If any constraint is violated, then the KMTS is repaired by 
adding a new state, applying $AbstractRepair$ to this state 
for property $\phi_{1}$ and adding a may-transition from 
$\hat{s}$ to this state.  If all constraints are satisfied, the repaired KMTS is returned.  

\begin{algorithm}[htb]       
\floatname{algorithm}{Algorithm}
\caption{$AbstractRepair_{EG}$} 
\label{alg:EG}         	
\begin{algorithmic}[1]
\renewcommand{\algorithmicrequire}{\textbf{Input:}}
\renewcommand{\algorithmicensure}{\textbf{Output:}}    
\REQUIRE $\hat{M} = (\hat{S}, \hat{S}_{0}, R_{must}, R_{may}, 
\hat{L})$, $\hat{s} \in \hat{S}$, 
a CTL property $\phi = EG\phi_{1}$ for which 
$(\hat{M},\hat{s}) \not\models \phi$, and a set of constraints 
$C = \{ (\hat{s}_{c_{1}},\phi_{c_{1}}), 
(\hat{s}_{c_{2}},\phi_{c_{2}}), ..., (\hat{s}_{c_{n}},\phi_{c_{n}}) \}$ where 
$\hat{s}_{c_{i}} \in \hat{S}$ and $\phi_{c_{i}}$ is a CTL formula.
\ENSURE $\hat{M^{\prime}} = (\hat{S^{\prime}}, \hat{S_{0}^{\prime}}, 
R_{must}^{\prime}, R_{may}^{\prime}, \hat{L^{\prime}})$ and 
$(\hat{M^{\prime}},\hat{s}) \models \phi$ or FAILURE.  
\STATE $\hat{M}_{1} := \hat{M}$
\IF {$(\hat{M},\hat{s}) \not\models \phi_{1}$}
	\STATE $RET := AbstractRepair(\hat{M},\hat{s},\phi_{1},C)$
	\IF { $RET == FAILURE$ }
		\RETURN FAILURE
	\ENDIF
	\STATE $\hat{M}_{1} := RET$
\ENDIF
\WHILE {there exists maximal path $\pi_{must} := [\hat{s}_{1},\hat{s}_{2},...]$ 
such that $\forall \hat{s}_{i} \in \pi_{must}$ it holds that 
$(\hat{M}_{1},\hat{s}_{i}) \models \phi_{1}$}
	\STATE $\hat{r}_{1} := (\hat{s},\hat{s}_{1})$, 
			$\hat{M^{\prime}} := AddMust(\hat{M}_{1},\hat{r}_{1})$
	\IF { $\hat{M^{\prime}} \models C$ }
		\RETURN $\hat{M^{\prime}}$
	\ENDIF
\ENDWHILE
\WHILE {there exists maximal path $\pi_{must} := [\hat{s},\hat{s}_{1},\hat{s}_{2},...]$ 
such that $\forall \hat{s}_{i} \neq \hat{s} \in \pi_{must}$ it holds that 
$(\hat{M}_{1},\hat{s}_{i}) \not\models \phi_{1}$}
	\STATE $\hat{M^{\prime}} := \hat{M}_{1}$
	\FORALL {$\hat{s}_{i} \in \pi_{must}$}
		\IF {$(\hat{M}_{1},\hat{s}_{i}) \not\models \phi_{1}$}
			\STATE $RET := AbstractRepair(\hat{M^{\prime}},\hat{s}_{i},\phi_{1},C)$
			\IF { $RET \neq FAILURE$ }
				\STATE $\hat{M^{\prime}} := RET$
			\ELSE 
				\STATE continue to next path
			\ENDIF
		\ENDIF
	\ENDFOR
	\RETURN $\hat{M^{\prime}}$
\ENDWHILE
\STATE $\hat{M^{\prime}} := AddState(\hat{M}_{1},\hat{s}_{n})$
\STATE $RET := AbstractRepair(\hat{M^{\prime}},\hat{s}_{n},\phi_{1},C)$
\IF { $RET \neq FAILURE$ }
	\STATE $\hat{M^{\prime}} := RET$
	\STATE $\hat{r}_{n} := (\hat{s},\hat{s}_{n})$, 
			$\hat{M^{\prime}} := AddMust(\hat{M^{\prime}},\hat{r}_{n})$
	\IF {$\hat{s}_{n}$ is a dead-end state}
		\STATE $\hat{r}_{n} := (\hat{s}_{n},\hat{s}_{n})$, 
				$\hat{M^{\prime}} := AddMust(\hat{M^{\prime}},\hat{r}_{n})$
	\ENDIF
	\IF { $\hat{M^{\prime}} \models C$ }
		\RETURN $\hat{M^{\prime}}$
	\ENDIF
\ENDIF
\RETURN FAILURE
\end{algorithmic}
\end{algorithm}

$AbstractRepair_{EG}$ which is presented in 
Algorithm~\ref{alg:EG} is the primitive function which 
is called when input CTL property is in the form of 
$EG\phi_{1}$.  
Initially, if $\phi_{1}$ is not satisfied at $\hat{s}$ 
$AbstractRepair$ is called for $\hat{s}$ and $\phi_{1}$, 
and a KMTS $\hat{M}_{1}$ is produced.  
At first, a must-transition is added from $\hat{s}$ to a 
state $\hat{s}_{1}$ of a maximal must-path (i.e. a must-path in which each transition appears at most once)    
$\pi_{must} := [\hat{s}_{1},\hat{s}_{2},...]$ such that $\forall \hat{s}_{i} \in \pi_{must}$, 
$(\hat{M}_{1},\hat{s}_{i}) \models \phi_{1}$.  If all 
constraints are satisfied, then the repaired KMTS is returned.  
Otherwise, a KMTS is produced by recursively calling 
$AbstractRepair$ to all states $\hat{s}_{i} \neq \hat{s}$ 
of any maximal must-path $\pi_{must} := [\hat{s}_{1},\hat{s}_{2},...]$ 
with $\forall \hat{s}_{i} \in \pi_{must}$,  
$(\hat{M}_{1},\hat{s}_{i}) \not\models \phi_{1}$.  
If there are violated constraints in $C$, then a repaired KMTS 
is produced by adding a new state, calling $AbstractRepair$ 
for this state and property $\phi_{1}$ and calling $AddMust$  
to insert a must-transition from $\hat{s}$ to 
the new state.  The resulting KMTS is returned by the algorithm, 
if all constraints in $C$ are satisfied. 

\begin{algorithm}[htb]       
\floatname{algorithm}{Algorithm}
\caption{$AbstractRepair_{AF}$} 
\label{alg:AF}         	
\begin{algorithmic}[1]
\renewcommand{\algorithmicrequire}{\textbf{Input:}}
\renewcommand{\algorithmicensure}{\textbf{Output:}}    
\REQUIRE $\hat{M} = (\hat{S}, \hat{S}_{0}, R_{must}, R_{may}, 
\hat{L})$, $\hat{s} \in \hat{S}$, 
a CTL property $\phi = AF\phi_{1}$ for which 
$(\hat{M},\hat{s}) \not\models \phi$, and a set of constraints 
$C = \{ (\hat{s}_{c_{1}},\phi_{c_{1}}), 
(\hat{s}_{c_{2}},\phi_{c_{2}}), ..., (\hat{s}_{c_{n}},\phi_{c_{n}}) \}$ where 
$\hat{s}_{c_{i}} \in \hat{S}$ and $\phi_{c_{i}}$ is a CTL formula.
\ENSURE $\hat{M^{\prime}} = (\hat{S^{\prime}}, \hat{S_{0}^{\prime}}, 
R_{must}^{\prime}, R_{may}^{\prime}, \hat{L^{\prime}})$ and 
$(\hat{M^{\prime}},\hat{s}) \models \phi$ or FAILURE.  
\STATE $\hat{M^{\prime}} := \hat{M}$
\WHILE {there exists maximal path $\pi_{may} := [\hat{s},\hat{s}_{1},...]$ 
such that $\forall \hat{s}_{i} \in \pi_{may}$ it holds 
that $(\hat{M^{\prime}},\hat{s}_{i}) \not\models \phi_{1}$}
	\FORALL {$\hat{s}_{i} \in \pi_{may}$}
		\STATE $RET := AbstractRepair(\hat{M^{\prime}},\hat{s}_{i},\phi_{1},C)$
		\IF { $RET \neq FAILURE$ }
			\STATE $\hat{M^{\prime}} := RET$
			\STATE continue to next path
		\ENDIF
	\ENDFOR
	\RETURN FAILURE
\ENDWHILE
\RETURN $\hat{M}^{\prime}$
\end{algorithmic}
\end{algorithm}

$AbstractRepair_{AF}$ shown in Algorithm~\ref{alg:AF} is called when the 
CTL formula $\phi$ is in the form of $AF\phi_{1}$.  
While there is maximal may-path  
$\pi_{may} := [\hat{s},\hat{s}_{1},...]$ such 
that $\forall \hat{s}_{i} \in \pi_{may}$, 
$(\hat{M^{\prime}},\hat{s}_{i}) \not\models \phi_{1}$, 
$AbstractRepair_{AF}$ tries to obtain a repaired KMTS by recursively 
calling $AbstractRepair$ to some state $\hat{s}_{i} \in \pi_{may}$.  
If all constraints are satisfied to the new KMTS, then it is returned 
as the repaired model.    
 
\begin{algorithm}[htb]       
\floatname{algorithm}{Algorithm}
\caption{$AbstractRepair_{EF}$} 
\label{alg:EF}         	
\begin{algorithmic}[1]
\renewcommand{\algorithmicrequire}{\textbf{Input:}}
\renewcommand{\algorithmicensure}{\textbf{Output:}}    
\REQUIRE $\hat{M} = (\hat{S}, \hat{S}_{0}, R_{must}, R_{may}, 
\hat{L})$, $\hat{s} \in \hat{S}$, 
a CTL property $\phi = EF\phi_{1}$ for which 
$(\hat{M},\hat{s}) \not\models \phi$, and a set of constraints 
$C = \{ (\hat{s}_{c_{1}},\phi_{c_{1}}), 
(\hat{s}_{c_{2}},\phi_{c_{2}}), ..., (\hat{s}_{c_{n}},\phi_{c_{n}}) \}$ where 
$\hat{s}_{c_{i}} \in \hat{S}$ and $\phi_{c_{i}}$ is a CTL formula.
\ENSURE $\hat{M^{\prime}} = (\hat{S^{\prime}}, \hat{S_{0}^{\prime}}, 
R_{must}^{\prime}, R_{may}^{\prime}, \hat{L^{\prime}})$ and 
$(\hat{M^{\prime}},\hat{s}) \models \phi$ or FAILURE.  
\FORALL {must-reachable states $\hat{s}_{i}$ from 
$\hat{s}$ with $(\hat{M},\hat{s}_{i}) \not\models \phi_{1}$ or 
$\hat{s}_{i} := \hat{s}$}
	\FORALL {$\hat{s}_{k} \in \hat{S}$ such that 
	$(\hat{M},\hat{s}_{k}) \models \phi_{1}$ }
		\STATE $\hat{r}_{k} := (\hat{s}_{i},\hat{s}_{k})$, 
				$\hat{M^{\prime}} := AddMust(\hat{M},\hat{r}_{k})$
		\IF {$\hat{M^{\prime}} \models C$}
			\RETURN $\hat{M^{\prime}}$
		\ENDIF
	\ENDFOR
\ENDFOR
\FORALL {must-reachable states $\hat{s}_{i}$ from 
$\hat{s}$ with $(\hat{M},\hat{s}_{i}) \not\models \phi_{1}$ }
	\STATE $RET := AbstractRepair(\hat{M},\hat{s}_{i},\phi_{1},C)$
	\IF { $RET \neq FAILURE$ }
		\STATE $\hat{M^{\prime}} := RET$
		\RETURN $\hat{M^{\prime}}$
	\ENDIF
\ENDFOR
\STATE $\hat{M}_{1} := AddState(\hat{M^{\prime}},\hat{s}_{n})$, 
		$RET := AbstractRepair(\hat{M}_{1},\hat{s}_{n},\phi_{1},C)$
\IF { $RET \neq FAILURE$ }			
	\STATE $\hat{M}_{1} := RET$
	\FORALL {must-reachable states $\hat{s}_{i}$ from 
	$\hat{s}$ with $(\hat{M},\hat{s}_{i}) \not\models \phi_{1}$ or 
	$\hat{s}_{i} := \hat{s}$}
		\STATE $\hat{r}_{i} := (\hat{s}_{i},\hat{s}_{n})$, 
			$\hat{M^{\prime}} := AddMust(\hat{M}_{1},\hat{r}_{i})$
		\IF {$\hat{s}_{n}$ is a dead-end state}
			\STATE $\hat{r}_{n} := (\hat{s}_{n},\hat{s}_{n})$, 
				$\hat{M^{\prime}} := AddMust(\hat{M^{\prime}},\hat{r}_{n})$
		\ENDIF
		\IF { $\hat{M^{\prime}} \models C$ }
			\RETURN $\hat{M^{\prime}}$
		\ENDIF
	\ENDFOR
\ENDIF
\RETURN FAILURE
\end{algorithmic}
\end{algorithm}

$AbstractRepair_{EF}$ shown in 
Algorithm~\ref{alg:EF} is called when the CTL property 
$\phi$ is in the form $EF\phi_{1}$.  
Initially, a KMTS is acquired by adding a 
must-transition from a must-reachable state $\hat{s}_{i}$ 
from $\hat{s}$ to a state 
$\hat{s}_{k} \in \hat{S}$ such that 
$(\hat{M},\hat{s}_{k}) \models \phi_{1}$.  If all 
constraints are satisfied then this KMTS is returned.  
Otherwise, a KMTS is produced by applying $AbstractRepair$ 
to a must-reachable state $\hat{s}_{i}$ from $\hat{s}$ for $\phi_{1}$.  
If none of the constraints is violated then this KMTS is 
returned.  
At any other case, a new KMTS is produced by adding a new 
state $\hat{s}_{n}$, recursively calling $AbstractRepair$ 
for this state and $\phi_{1}$ and adding a must-transition 
from $\hat{s}$ or from a must-reachable $\hat{s}_{i}$ from  
$\hat{s}$ to $\hat{s}_{n}$.  If all 
constraints are satisfied, then this KMTS is returned as a 
repaired model by the algorithm.   

\begin{algorithm}[htb]       
\floatname{algorithm}{Algorithm}
\caption{$AbstractRepair_{AU}$} 
\label{alg:AU}         	
\begin{algorithmic}[1]
\renewcommand{\algorithmicrequire}{\textbf{Input:}}
\renewcommand{\algorithmicensure}{\textbf{Output:}}    
\REQUIRE $\hat{M} = (\hat{S}, \hat{S}_{0}, R_{must}, R_{may}, 
\hat{L})$, $\hat{s} \in \hat{S}$, 
a CTL property $\phi = A(\phi_{1}U\phi_{2})$ for which 
$(\hat{M},\hat{s}) \not\models \phi$, and a set of constraints 
$C = \{ (\hat{s}_{c_{1}},\phi_{c_{1}}), 
(\hat{s}_{c_{2}},\phi_{c_{2}}), ..., (\hat{s}_{c_{n}},\phi_{c_{n}}) \}$ where 
$\hat{s}_{c_{i}} \in \hat{S}$ and $\phi_{c_{i}}$ is a CTL formula.
\ENSURE $\hat{M^{\prime}} = (\hat{S^{\prime}}, \hat{S_{0}^{\prime}}, 
R_{must}^{\prime}, R_{may}^{\prime}, \hat{L^{\prime}})$ and 
$(\hat{M^{\prime}},\hat{s}) \models \phi$ or FAILURE.  
\STATE $\hat{M}_{1} := \hat{M}$
\IF {$(\hat{M},\hat{s}) \not\models \phi_{1}$}
	\STATE $RET := AbstractRepair(\hat{M},\hat{s},\phi_{1},C)$
	\IF { $RET == FAILURE$ }
		\RETURN FAILURE
	\ELSE
		\STATE $\hat{M}_{1} := RET$
	\ENDIF
\ENDIF
\WHILE {there exists path $\pi_{may} := [\hat{s}_{1},...,\hat{s}_{m}]$  
such that $\forall \hat{s}_{i} \in \pi_{may}$ it holds 
that $(\hat{M}_{1},\hat{s}_{i}) \models \phi_{1}$ and there does not 
exist $\hat{r}_{m} := (\hat{s}_{m},\hat{s}_{n}) \in R_{may}$ such that 
$(\hat{M}_{1},\hat{s}_{n}) \models \phi_{2}$}
	\FORALL {$\hat{s}_{j} \in \pi_{may}$ for which $(\hat{M}_{1},\hat{s}_{j}) \not\models \phi_{2}$ with $\hat{s}_{j} \neq \hat{s}_{1}$ }
		\STATE $RET := AbstractRepair(\hat{M}_{1},\hat{s}_{j},\phi_{2},C)$
		\IF { $RET \neq FAILURE$ }
			\STATE $\hat{M^{\prime}} := RET$
			\STATE continue to next path
		\ENDIF
	\ENDFOR
	\RETURN FAILURE
\ENDWHILE
\RETURN $\hat{M^{\prime}}$
\end{algorithmic}
\end{algorithm}

$AbstractRepair_{AU}$ is presented in Algorithm~\ref{alg:AU} 
and is called when $\phi = A(\phi_{1}U\phi_{2})$. 
If $\phi_{1}$ is not satisfied at $\hat{s}$, 
then a KMTS $\hat{M}_{1}$ is produced by applying $AbstractRepair$ 
to $\hat{s}$ for $\phi_{1}$.  Otherwise, $\hat{M}_{1}$ is same to 
$\hat{M}$.  
A new KMTS is produced as follows: for all may-paths 
$\pi_{may} := [\hat{s}_{1},...,\hat{s}_{m}]$ such that 
$\forall \hat{s}_{i} \in \pi_{may}$, $(\hat{M}_{1},\hat{s}_{i}) \models \phi_{1}$ 
and for which there does not $\hat{r}_{m} := (\hat{s}_{m},\hat{s}_{n}) \in R_{may}$ with  
$(\hat{M}_{1},\hat{s}_{n}) \models \phi_{2}$, 
$AbstractRepair$ is called for property $\phi_{2}$  for some state 
$\hat{s}_{j} \in \pi_{may}$ with $(\hat{M}_{1},\hat{s}_{j}) \not\models \phi_{2}$.     
If the resulting KMTS satisfies all constraints, then it is 
returned as a repair solution.  

\begin{algorithm}[htb]       
\floatname{algorithm}{Algorithm}
\caption{$AbstractRepair_{EU}$} 
\label{alg:EU}         	
\begin{algorithmic}[1]
\renewcommand{\algorithmicrequire}{\textbf{Input:}}
\renewcommand{\algorithmicensure}{\textbf{Output:}}    
\REQUIRE $\hat{M} = (\hat{S}, \hat{S}_{0}, R_{must}, R_{may}, 
\hat{L})$, $\hat{s} \in \hat{S}$, 
a CTL property $\phi = E(\phi_{1}U\phi_{2})$ for which 
$(\hat{M},\hat{s}) \not\models \phi$, and a set of constraints 
$C = \{ (\hat{s}_{c_{1}},\phi_{c_{1}}), 
(\hat{s}_{c_{2}},\phi_{c_{2}}), ..., (\hat{s}_{c_{n}},\phi_{c_{n}}) \}$ where 
$\hat{s}_{c_{i}} \in \hat{S}$ and $\phi_{c_{i}}$ is a CTL formula.
\ENSURE $\hat{M^{\prime}} = (\hat{S^{\prime}}, \hat{S_{0}^{\prime}}, 
R_{must}^{\prime}, R_{may}^{\prime}, \hat{L^{\prime}})$ and 
$(\hat{M^{\prime}},\hat{s}) \models \phi$ or FAILURE.  
\STATE $\hat{M}_{1} := \hat{M}$
\IF {$(\hat{M},\hat{s}) \not\models \phi_{1}$}
	\STATE $RET := AbstractRepair(\hat{M},\hat{s},\phi_{1},C)$
	\IF { $RET == FAILURE$ }
		\RETURN FAILURE
	\ELSE
		\STATE $\hat{M}_{1} := RET$
	\ENDIF
\ENDIF
\WHILE { there exists path $\pi_{must} := [\hat{s}_{1},...,\hat{s}_{m}]$ 
such that $\forall \hat{s}_{i} \in \pi_{must}$, 
$(\hat{M}_{1},\hat{s}_{i}) \models \phi_{1}$}
	\FORALL { $\hat{s}_{j} \in \hat{S}$ with     
		$(\hat{M}_{1},\hat{s}_{j}) \models \phi_{2}$}
		\STATE $\hat{r}_{j} := (\hat{s}_{m},\hat{s}_{j})$, 
				$\hat{M}^{\prime} := AddMust(\hat{M}_{1},\hat{r}_{j})$
		\IF { $\hat{M^{\prime}} \models C$ }
			\RETURN $\hat{M^{\prime}}$
		\ENDIF
	\ENDFOR
\ENDWHILE	
\STATE $\hat{M^{\prime}} := AddState(\hat{M}_{1},\hat{s}_{k})$
\STATE $RET := AbstractRepair(\hat{M^{\prime}},\hat{s}_{k},\phi_{2},C)$
\IF { $RET \neq FAILURE$ }
	\STATE $\hat{M^{\prime}} := RET$
	\STATE $\hat{r}_{n} := (\hat{s},\hat{s}_{k})$, 
			$\hat{M^{\prime}} := AddMust(\hat{M^{\prime}},\hat{r}_{n})$
	\IF {$\hat{s}_{k}$ is a dead-end state}
		\STATE $\hat{r}_{k} := (\hat{s}_{k},\hat{s}_{k})$, 
				$\hat{M^{\prime}} := AddMust(\hat{M^{\prime}},\hat{r}_{k})$
	\ENDIF
	\IF { $\hat{M^{\prime}} \models C$ }
		\RETURN $\hat{M^{\prime}}$
	\ENDIF
\ENDIF
\RETURN FAILURE
\end{algorithmic}
\end{algorithm}

$AbstractRepair_{EU}$ is called if for input CTL formula $\phi$ 
it holds that $\phi = E(\phi_{1}U\phi_{2})$.  
$AbstractRepair_{EU}$ is presented in Algorithm~\ref{alg:EU}.  
Firstly, if $\phi_{1}$ is not satisfied at $\hat{s}$, then 
$AbstractRepair$ is called for $\hat{s}$ and $\phi_{1}$ 
and a KMTS $\hat{M}_{1}$ is produced for which 
$(\hat{M}_{1},\hat{s}) \models \phi_{1}$.  Otherwise, 
$\hat{M}_{1}$ is same to $\hat{M}$.  
A new KMTS is produced as follows: 
for a must-path $\pi_{must} := [\hat{s}_{1},...,\hat{s}_{m}]$ 
such that $\forall \hat{s}_{i} \in \pi_{must}$, 
$(\hat{M}_{1},\hat{s}_{i}) \models \phi_{1}$
and for a $\hat{s}_{j} \in \hat{S}$ with 
$(\hat{M}_{1},\hat{s}_{j}) \models \phi_{2}$,
a must-transition is added from 
$\hat{s}_{m}$ to $\hat{s}_{j}$.  If all 
constraints are satisfied then the new KMTS is returned.  
Alternatively, a KMTS is produced by adding a new state 
$\hat{s}_{n}$, recursively calling $AbstractRepair$ for 
$\phi_{2}$ and $\hat{s}_{n}$ and adding a must-transition 
from $\hat{s}$ to $\hat{s}_{n}$.  In the case that no 
constraint is violated then this is a repaired KMTS and it 
is returned from the function.

\subsection{Properties of the Algorithm}
\label{subsec:alg_props}
\emph{AbstractRepair} is \emph{well-defined}~\cite{BGS07}, in the 
sense that the algorithm always proceeds and eventually returns a result $\hat{M}^{\prime}$ or FAILURE such that $(\hat{M}^\prime,\hat{s}) \models \phi$, for any input $\hat{M}$, $\phi$ and $C$, with $(\hat{M},\hat{s}) \not\models \phi$.  Moreover, the algorithm steps are well-ordered, as opposed to existing concrete model repair solutions~\cite{CR11,ZD08} that entail nondeterministic behavior.  

\subsubsection{Soundness}
\label{subsubsec:alg_soundness}

\begin{lem}
\label{theor:sound_help}
Let a KMTS $\hat{M}$, a CTL formula $\phi$ with $(\hat{M},\hat{s}) \not\models \phi$ for some $\hat{s}$ of $\hat{M}$, and a set $C = \{ (\hat{s}_{c_{1}},\phi_{c_{1}}), (\hat{s}_{c_{2}},\phi_{c_{2}}), ..., (\hat{s}_{c_{n}},\phi_{c_{n}}) \}$ with $(\hat{M},\hat{s}_{c_{i}}) \models \phi_{c_{i}}$ for all 
$(\hat{s}_{c_{n}},\phi_{c_{n}}) \in C$.  If $AbstractRepair(\hat{M},\hat{s},\phi,C)$ returns a KMTS $\hat{M}^{\prime}$, then $(\hat{M}^{\prime},\hat{s}) \models \phi$ and $(\hat{M}^{\prime},\hat{s}_{c_{i}}) \models \phi_{c_{i}}$ for all $(\hat{s}_{c_{i}},\phi_{c_{i}}) \in C$.
\end{lem}
\begin{proof}
We use structural induction on $\phi$.  For brevity, we write $\hat{M} \models C$ 
to denote that $(\hat{M},\hat{s}_{c_{i}}) \models \phi_{c_{i}}$, 
for all $(\hat{s}_{c_{i}},\phi_{c_{i}}) \in C$.  

\paragraph{Base Case: }
\begin{itemize}
\item if $\phi = \top$, the lemma is trivially true, because $(\hat{M},\hat{s}) \models \phi$    
\item if $\phi = \bot$, then $AbstractRepair(\hat{M},\hat{s},\phi,C)$ 
returns FAILURE at line 2 of Algorithm~\ref{alg:main} and the lemma is also trivially true. 
\item if $\phi = p \in AP$, $AbstractRepair_{ATOMIC}(\hat{M},\hat{s},p,C)$ is called at line 4 of Algorithm~\ref{alg:main} and an $\hat{M^{\prime}} = ChangeLabel(\hat{M},\hat{s},p)$ is computed 
at line 1 of Algorithm~\ref{alg:ATOMIC}.  Since $p \in \hat{L}^{\prime}(\hat{s})$ in $\hat{M^{\prime}}$,  
from 3-valued semantics of CTL over KMTSs we have $(\hat{M^{\prime}},\hat{s}) \models \phi$.  Algorithm~\ref{alg:ATOMIC} returns $\hat{M^{\prime}}$ at line 3, if and only if $\hat{M}^{\prime} \models C$ and the lemma is true.     
\end{itemize}

\paragraph{Induction Hypothesis:} For CTL formulae $\phi_{1}, \phi_{2}$, the lemma is true.  Thus, for $\phi_{1}$ (resp. $\phi_{2}$), if $AbstractRepair(\hat{M},\hat{s},\phi_{1},C)$ returns a KMTS $\hat{M}^{\prime}$, then $(\hat{M^{\prime}},\hat{s}) \models \phi_{1}$ and $\hat{M}^{\prime} \models C$.    

\paragraph{Inductive Step:}
\begin{itemize}
\item if $\phi = \phi_{1} \vee \phi_{2}$, then $AbstractRepair(\hat{M},\hat{s},\phi,C)$ calls $AbstractRepair_{OR}(\hat{M},\hat{s},\phi_{1} \vee \phi_{2},C)$ at line 8 of Algorithm~\ref{alg:main}.     From the induction hypothesis, if a KMTS $\hat{M}_{1}$ is returned by $AbstractRepair(\hat{M},\hat{s},\phi_{1},C)$ at line 1 of Algorithm~\ref{alg:OR} and a KMTS $\hat{M}_{2}$ is returned by $AbstractRepair(\hat{M},\hat{s},\phi_{2},C)$ respectively, then $(\hat{M}_{1},\hat{s}) \models \phi_{1}$, $\hat{M}_{1} \models C$ and $(\hat{M}_{2},\hat{s}) \models \phi_{1}$, $\hat{M}_{2} \models C$.  $AbstractRepair_{OR}(\hat{M},\hat{s},\phi_{1} \vee \phi_{2},C)$ returns at line 8 of Algorithm~\ref{alg:main} the KMTS $\hat{M^{\prime}}$, which can be either $\hat{M}_{1}$ or $\hat{M}_{2}$.  Therefore, $(\hat{M^{\prime}},\hat{s}) \models \phi_{1}$ or $(\hat{M^{\prime}},\hat{s}) \models \phi_{2}$ and $\hat{M^{\prime}} \models C$ in both cases.  From 3-valued semantics of CTL, $(\hat{M^{\prime}},\hat{s}) \models \phi_{1} \vee \phi_{2}$ and the lemma is true. 

\item if $\phi = \phi_{1} \wedge \phi_{2}$, then 
$AbstractRepair(\hat{M},\hat{s},\phi,C)$ calls 
$AbstractRepair_{AND}(\hat{M},\hat{s},\phi_{1} \wedge \phi_{2},C)$ at line 6 of 
Algorithm~\ref{alg:main}.  From the induction hypothesis, if at line 1 of Algorithm~\ref{alg:AND} $AbstractRepair(\hat{M},\hat{s},\phi_{1},C)$ returns a KMTS $\hat{M}_{1}$, then 
$(\hat{M}_{1},\hat{s}) \models \phi_{1}$ and $\hat{M}_{1} \models C$.  Consequently, $\hat{M}_{1} \models C_{1}$, where $C_{1} = C \cup {(\hat{s},\phi_{1})}$.  At line 7, if $AbstractRepair(\hat{M}_{1},\hat{s},\phi_{2},C_{1})$ returns a KMTS $\hat{M}_{1}^{\prime}$, then from the induction hypothesis $(\hat{M}_{1}^{\prime},\hat{s}) \models \phi_{2}$ and $\hat{M}_{1}^{\prime} \models C_{1}$.

In the same manner, if the calls at lines 2 and 12 of Algorithm~\ref{alg:AND} return the KMTSs $\hat{M}_{2}$ and $\hat{M}_{2}^{\prime}$, then from the induction hypothesis $(\hat{M}_{2},\hat{s}) \models \phi_{2}$, $\hat{M}_{2} \models C$ and $(\hat{M}_{2}^{\prime},\hat{s}) \models \phi_{1}$, 
$\hat{M}_{2}^{\prime} \models C_{2}$ with $C_{2} = C \cup {(\hat{s},\phi_{2})}$.

The KMTS $\hat{M^{\prime}}$ at line 6 of Algorithm~\ref{alg:main} can be either $\hat{M}_{1}^{\prime}$ or $\hat{M}_{2}^{\prime}$ and therefore, $(\hat{M^{\prime}},\hat{s}) \models \phi_{1}$, 
$(\hat{M^{\prime}},\hat{s}) \models \phi_{2}$ and $\hat{M^{\prime}} \models C$.  From 3-valued semantics of CTL it holds that $(\hat{M^{\prime}},\hat{s}) \models \phi_{1} \wedge \phi_{2}$ and the lemma is true.  

\item if $\phi = EX\phi_{1}$, $AbstractRepair(\hat{M},\hat{s},\phi,C)$ calls 
$AbstractRepair_{EX}(\hat{M},\hat{s},EX\phi_{1},C)$ at line 10 of 
Algorithm~\ref{alg:main}.    

If a KMTS $\hat{M}^{\prime}$ is returned at line 5 of Algorithm~\ref{alg:EX}, there is a
state $\hat{s}_{1}$ with $(\hat{M},\hat{s}_{1}) \models \phi_{1}$ such that 
$\hat{M}^{\prime} = AddMust(\hat{M},(\hat{s},\hat{s}_{1}))$ and $\hat{M^{\prime}} \models C$.
From 3-valued semantics of CTL, we conclude that $(\hat{M^{\prime}},\hat{s}) \models EX\phi_{1}$.

If a $\hat{M}^{\prime}$ is returned at line 11, there is $(\hat{s},\hat{s}_{1}) \in R_{must}$ such that $(\hat{M^{\prime}},\hat{s}_{1}) \models \phi_{1}$
and $\hat{M^{\prime}} \models C$ from the induction hypothesis, since $\hat{M^{\prime}}=AbstractRepair(\hat{M},\hat{s}_{1},\phi_{1},C)$.  From 3-valued semantics of CTL, we conclude that $(\hat{M^{\prime}},\hat{s}) \models EX\phi_{1}$.   

If a $\hat{M}^{\prime}$ is returned at line 18, a must transition $(\hat{s},\hat{s}_{n})$ to a new state has been added and $\hat{M}^{\prime}=AbstractRepair(AddMust(\hat{M},(\hat{s},\hat{s}_{n})),\hat{s}_{n},\phi_{1},C)$.  Then, from the induction hypothesis $(\hat{M^{\prime}},\hat{s}_{n}) \models \phi_{1}$, $\hat{M^{\prime}} \models C$ and from 3-valued semantics of CTL, we also conclude that $(\hat{M^{\prime}},\hat{s}) \models EX\phi_{1}$.

\item if $\phi = AG\phi_{1}$, $AbstractRepair(\hat{M},\hat{s},\phi,C)$ calls 
$AbstractRepair_{AG}(\hat{M},\hat{s},AG\phi_{1},C)$ at line 10 of 
Algorithm~\ref{alg:main}.  If $(\hat{M},\hat{s}) \not\models \phi_{1}$ and $AbstractRepair(\hat{M},\hat{s},\phi_{1},C)$ returns 
a KMTS $\hat{M}_{0}$ at line 2 of Algorithm~\ref{alg:AG}, then from the induction 
hypothesis $(\hat{M}_{0},\hat{s}) \models \phi_{1}$ and 
$\hat{M}_{0} \models C$.  Otherwise, $\hat{M}_{0} = \hat{M}$ and $(\hat{M}_{0},\hat{s}) \models \phi_{1}$ also hold true.    

If Algorithm~\ref{alg:AG} returns a $\hat{M}^{\prime}$ at line 16, then $\hat{M}^{\prime} \models C$ and $\hat{M}^{\prime}$ is the result of successive $AbstractRepair(\hat{M_{i}},\hat{s}_{k},\phi_{1},C)$ calls with $\hat{M_{i}}=AbstractRepair(\hat{M}_{i-1},\hat{s}_{k},\phi_{1},C)$ and $i=1, . . .$, for
all may-reachable states $\hat{s}_{k}$ from $\hat{s}$ 
such that $(\hat{M}_{0},\hat{s}_{k}) \not\models \phi_{1}$.     
From the induction hypothesis, $(\hat{M}^{\prime},\hat{s}_{k}) \models \phi_{1}$ and 
$\hat{M^{\prime}} \models C$ for all such $\hat{s}_{k}$ and from 3-valued semantics of CTL we conclude that $(\hat{M^{\prime}},\hat{s}) \models AG\phi_{1}$.   

\end{itemize}

\noindent We prove the lemma for all other cases in a similar manner.  
\end{proof}

\begin{thm}[Soundness]
\label{theor:sound}
Let a KMTS $\hat{M}$, a CTL formula $\phi$ with 
$(\hat{M},\hat{s}) \not\models \phi$, for some $\hat{s}$ of 
$\hat{M}$.  If $AbstractRepair(\hat{M},\hat{s},\phi,\emptyset)$ returns a 
KMTS $\hat{M}^{\prime}$, then $(\hat{M}^{\prime},\hat{s}) \models \phi$.  
\end{thm}
\begin{proof}
We use structural induction on $\phi$ and Lemma~\ref{theor:sound_help}
in the inductive step for $\phi_{1} \wedge \phi_{2}$.  

\paragraph{Base Case:}
\begin{itemize}
\item if $\phi = \top$, Theorem~\ref{theor:sound} is trivially true, because 
$(\hat{M},\hat{s}) \models \phi$.     
\item if $\phi = \bot$, then $AbstractRepair(\hat{M},\hat{s},\bot,\emptyset)$ 
returns FAILURE at line 2 of Algorithm~\ref{alg:main} and the theorem is also 
trivially true.    
\item if $\phi = p \in AP$, $AbstractRepair_{ATOMIC}(\hat{M},\hat{s},p,\emptyset)$ 
is called at line 4 of Algorithm~\ref{alg:main} and an $\hat{M^{\prime}} = 
ChangeLabel(\hat{M},\hat{s},p)$ is computed at line 1.  Because of the fact that 
$p \in \hat{L}^{\prime}(\hat{s})$ in $\hat{M^{\prime}}$, from 3-valued semantics 
of CTL over KMTSs we have $(\hat{M^{\prime}},\hat{s}) \models \phi$.    
Algorithm~\ref{alg:ATOMIC} returns $\hat{M^{\prime}}$ at line 3 because $C$ is empty, 
and the theorem is true.    
\end{itemize}

\paragraph{Induction Hypothesis:}
For CTL formulae $\phi_{1}$, $\phi_{2}$, the theorem is true.  Thus, for $\phi_{1}$ 
(resp. $\phi_{2}$), if $AbstractRepair(\hat{M},\hat{s},\phi,\emptyset)$ returns a 
KMTS $\hat{M}^{\prime}$, then $(\hat{M^{\prime}},\hat{s}) \models \phi_{1}$.      

\paragraph{Inductive Step:}
\begin{itemize}
\item if $\phi = \phi_{1} \vee \phi_{2}$, then $AbstractRepair(\hat{M},\hat{s},\phi,\emptyset)$ 
calls $AbstractRepair_{OR}(\hat{M},\hat{s},\phi_{1} \vee \phi_{2},\emptyset)$ at 
line 8 of Algorithm~\ref{alg:main}.    

From the induction hypothesis, if $AbstractRepair(\hat{M},\hat{s},\phi_{1},\emptyset)$
returns a KMTS $\hat{M}_{1}$ at line 1 of Algorithm~\ref{alg:OR} 
and $AbstractRepair(\hat{M},\hat{s},\phi_{2},\emptyset)$ returns a KMTS $\hat{M}_{2}$ 
respectively, then $(\hat{M}_{1},\hat{s}) \models \phi_{1}$ and 
$(\hat{M}_{2},\hat{s}) \models \phi_{1}$.  
$AbstractRepair_{OR}(\hat{M},\hat{s},\phi_{1} \vee \phi_{2},\emptyset)$ 
returns at line 8 of Algorithm~\ref{alg:main} the KMTS $\hat{M^{\prime}}$, 
which can be either $\hat{M}_{1}$ or $\hat{M}_{2}$.  Therefore, 
$(\hat{M^{\prime}},\hat{s}) \models \phi_{1}$ or 
$(\hat{M^{\prime}},\hat{s}) \models \phi_{2}$.  
From 3-valued semantics of CTL, 
$(\hat{M^{\prime}},\hat{s}) \models \phi_{1} \vee \phi_{2}$ and the theorem is true. 

\item if $\phi = \phi_{1} \wedge \phi_{2}$, then 
$AbstractRepair(\hat{M},\hat{s},\phi,\emptyset)$ calls 
$AbstractRepair_{AND}(\hat{M},\hat{s},\phi_{1} \wedge \phi_{2},\emptyset)$ 
at line 6 of Algorithm~\ref{alg:main}.  From the induction 
hypothesis, if at line 1 of Algorithm~\ref{alg:AND} 
$AbstractRepair(\hat{M},\hat{s},\phi_{1},\emptyset)$ returns a 
KMTS $\hat{M}_{1}$, then $(\hat{M}_{1},\hat{s}) \models \phi_{1}$.  
Consequently, $\hat{M}_{1} \models C_{1}$, where 
$C_{1} = \emptyset \cup {(\hat{s},\phi_{1})}$.  
At line 7, if $AbstractRepair(\hat{M}_{1},\hat{s},\phi_{2},C_{1})$ 
returns a KMTS $\hat{M}_{1}^{\prime}$, then from Lemma~\ref{theor:sound_help} 
$(\hat{M}_{1}^{\prime},\hat{s}) \models \phi_{2}$ and $\hat{M}_{1}^{\prime} \models C_{1}$.

Likewise, if the calls at lines 2 and 12 of Algorithm~\ref{alg:AND} 
return the KMTSs $\hat{M}_{2}$ and $\hat{M}_{2}^{\prime}$, then from the induction hypothesis 
$(\hat{M}_{2},\hat{s}) \models \phi_{2}$ and from Lemma~\ref{theor:sound_help} 
$(\hat{M}_{2}^{\prime},\hat{s}) \models \phi_{1}$, $\hat{M}_{2}^{\prime} \models C_{2}$ 
with $C_{2} = \emptyset \cup {(\hat{s},\phi_{2})}$.

The KMTS $\hat{M^{\prime}}$ at line 7 of Algorithm~\ref{alg:main} 
can be either $\hat{M}_{1}^{\prime}$ or $\hat{M}_{2}^{\prime}$ and 
therefore, $(\hat{M^{\prime}},\hat{s}) \models \phi_{1}$ and  
$(\hat{M^{\prime}},\hat{s}) \models \phi_{2}$.  From 3-valued semantics 
of CTL it holds that $(\hat{M^{\prime}},\hat{s}) \models \phi_{1} \wedge \phi_{2}$ 
and the lemma is true. 

\item if $\phi = EX\phi_{1}$, $AbstractRepair(\hat{M},\hat{s},\phi,\emptyset)$ calls 
$AbstractRepair_{EX}(\hat{M},\hat{s},EX\phi_{1},\emptyset)$ at line 10 of 
Algorithm~\ref{alg:main}.  

If a KMTS $\hat{M}^{\prime}$ is returned at line 5 of Algorithm~\ref{alg:EX}, 
there is a state $\hat{s}_{1}$ with $(\hat{M},\hat{s}_{1}) \models \phi_{1}$ 
such that $\hat{M}^{\prime} = AddMust(\hat{M},(\hat{s},\hat{s}_{1}))$.  
From 3-valued semantics of CTL, we conclude that 
$(\hat{M^{\prime}},\hat{s}) \models EX\phi_{1}$.

If a $\hat{M}^{\prime}$ is returned at line 11, there is 
$(\hat{s},\hat{s}_{1}) \in R_{must}$ such that 
$(\hat{M^{\prime}},\hat{s}_{1}) \models \phi_{1}$ from the induction hypothesis, 
since $\hat{M^{\prime}} = AbstractRepair(\hat{M},\hat{s}_{1},\phi_{1},\emptyset)$.  
From 3-valued semantics of CTL, we conclude that 
$(\hat{M^{\prime}},\hat{s}) \models EX\phi_{1}$.   

If a $\hat{M}^{\prime}$ is returned at line 18, a must transition 
$(\hat{s},\hat{s}_{n})$ to a new state has been added and $\hat{M}^{\prime} = AbstractRepair(AddMust(\hat{M},(\hat{s},\hat{s}_{n})),\hat{s}_{n},\phi_{1},\emptyset)$.  
Then, from the induction hypothesis $(\hat{M^{\prime}},\hat{s}_{n}) \models \phi_{1}$
 and from 3-valued semantics of CTL, we also conclude that 
$(\hat{M^{\prime}},\hat{s}) \models EX\phi_{1}$.

\item if $\phi = AG\phi_{1}$, $AbstractRepair(\hat{M},\hat{s},\phi,\emptyset)$ 
calls $AbstractRepair_{AG}(\hat{M},\hat{s},AG\phi_{1},\emptyset)$ at line 10 of 
Algorithm~\ref{alg:main}.  If $(\hat{M},\hat{s}) \not\models \phi_{1}$ and 
$AbstractRepair(\hat{M},\hat{s},\phi_{1},\emptyset)$ returns a KMTS 
$\hat{M}_{0}$ at line 2 of Algorithm~\ref{alg:AG}, then from the induction 
hypothesis $(\hat{M}_{0},\hat{s}) \models \phi_{1}$.  Otherwise, 
$\hat{M}_{0} = \hat{M}$ and $(\hat{M}_{0},\hat{s}) \models \phi_{1}$, 
$\hat{M}_{0} \models C$ also hold true.    

If Algorithm~\ref{alg:AG} returns a $\hat{M}^{\prime}$ at line 16, 
this KMTS is the result of successive calls of $AbstractRepair(\hat{M_{i}},\hat{s}_{k},\phi_{1},\emptyset)$ with $\hat{M_{i}}=AbstractRepair(\hat{M}_{i-1},\hat{s}_{k},\phi_{1},\emptyset)$ and $i=1, . . .$, for
all may-reachable states $\hat{s}_{k}$ from $\hat{s}$ 
such that $(\hat{M}_{0},\hat{s}_{k}) \not\models \phi_{1}$.     
From the induction hypothesis, $(\hat{M}^{\prime},\hat{s}_{k}) \models \phi_{1}$ 
for all such $\hat{s}_{k}$ and from 3-valued semantics of CTL we conclude that 
$(\hat{M^{\prime}},\hat{s}) \models AG\phi_{1}$.   
\end{itemize}
\noindent We prove the theorem for all other cases in the same way. 
\end{proof}

\noindent Theorem~\ref{theor:sound} shows that \emph{AbstractRepair} 
is \emph{sound} in the sense that if it returns a KMTS 
$\hat{M}^{\prime}$, then $\hat{M}^{\prime}$ satisfies 
property $\phi$.  In this case, from the definitions of the basic 
repair operations, it follows that one or more KSs can be 
obtained for which $\phi$ holds true.

\subsubsection{Semi-completeness}
\label{subsubsec:alg_completeness}

\begin{defi}[\emph{mr}-CTL]
Given a set $AP$ of atomic propositions, we define the syntax of 
a CTL fragment inductively via a Backus Naur Form:  
\begin{align*}
	\phi ::== &\bot \, | \, \top \, | \, p \, | \, (\neg \phi) \, | \, (\phi \vee \phi) \, | \, AXp \, | \, EXp \, | \, AFp \\ 
	& | \, EFp \, | \, AGp \, | \, EGp \, | \, A[p \, U \, p] \, | \, E[p \, U \, p]
\end{align*}
where $p$ ranges over $AP$.  
\end{defi}

\emph{mr}-CTL includes most of the CTL formulae apart from those with nested 
path quantifiers or conjunction.

\begin{thm}[Completeness]
\label{theor:complete}
Given a KMTS $\hat{M}$, an \textit{mr}-CTL formula $\phi$ with 
$(\hat{M},\hat{s}) \not\models \phi$, for some $\hat{s}$ 
of $\hat{M}$, if there exists a KMTS 
$\hat{M}^{\prime\prime}$ over the same set $AP$ of atomic propositions with 
$(\hat{M}^{\prime\prime},\hat{s}) \models \phi$, 
$AbstractRepair(\hat{M},\hat{s},\phi,\emptyset)$ returns a 
KMTS $\hat{M}^{\prime}$ such that 
$(\hat{M}^{\prime},\hat{s}) \models \phi$.    
\end{thm}
\begin{proof}
We prove the theorem using structural induction on $\phi$.

\paragraph{Base Case:}
\begin{itemize}
\item if $\phi = \top$, Theorem~\ref{theor:complete} is trivially true, 
because for any KMTS $\hat{M}$ it holds that 
$(\hat{M},\hat{s}) \models \phi$.    
\item if $\phi = \bot$, then the theorem is trivially true, because 
there does not exist a KMTS $\hat{M}^{\prime\prime}$ such that 
$(\hat{M}^{\prime\prime},\hat{s}) \models \phi$.   
\item if $\phi = p \in AP$, there is a KMTS $\hat{M}^{\prime\prime}$
with $p \in \hat{L}^{\prime\prime}(\hat{s})$ and therefore $(\hat{M}^{\prime\prime},\hat{s}) \models \phi$.  
Algorithm~\ref{alg:main} calls $AbstractRepair_{ATOMIC}(\hat{M},\hat{s},p,\emptyset)$ 
at line 4 and an $\hat{M^{\prime}} = ChangeLabel(\hat{M},\hat{s},p)$ 
is computed at line 1 of Algorithm~\ref{alg:ATOMIC}.  Since $C$ is empty, $\hat{M^{\prime}}$ is returned at line 3 and $(\hat{M}^{\prime},\hat{s}) \models \phi$ from 3-valued semantics of CTL.  Therefore, the theorem is true.  
\end{itemize}

\paragraph{Induction Hypothesis:}
For \emph{mr}-CTL formulae $\phi_{1}$, $\phi_{2}$, the theorem is true.  
Thus, for $\phi_{1}$ (resp. $\phi_{2}$), 
if there is a KMTS $\hat{M}^{\prime\prime}$ over the same set $AP$ of atomic propositions with 
$(\hat{M}^{\prime\prime},\hat{s}) \models \phi_{1}$, 
$AbstractRepair(\hat{M},\hat{s},\phi_{1},\emptyset)$ returns a KMTS 
$\hat{M}^{\prime}$ such that $(\hat{M}^{\prime},\hat{s}) \models \phi_{1}$.  

\paragraph{Inductive Step:}
\begin{itemize}
\item if $\phi = \phi_{1} \vee \phi_{2}$, from the 3-valued semantics of CTL a 
KMTS that satisfies $\phi$ exists if and only if there is a KMTS satisfying any of the $\phi_{1}$, $\phi_{2}$.
From the induction hypothesis,
if there is a KMTS $\hat{M}_{1}^{\prime\prime}$ with $(\hat{M}_{1}^{\prime\prime},\hat{s}) \models \phi_{1}$, $AbstractRepair(\hat{M},\hat{s},\phi_{1},\emptyset)$
at line 1 of Algorithm~\ref{alg:OR} returns a KMTS $\hat{M}_{1}^{\prime}$ such that $(\hat{M}_{1}^{\prime},\hat{s}) \models \phi_{1}$. Respectively, $AbstractRepair(\hat{M},\hat{s},\phi_{2},\emptyset)$ at line 2 of Algorithm~\ref{alg:OR} can return a KMTS $\hat{M}_{2}^{\prime}$ with $(\hat{M}_{2}^{\prime},\hat{s}) \models \phi_{2}$. In any case, if either $\hat{M}_{1}^{\prime}$ or $\hat{M}_{2}^{\prime}$ exists, for the KMTS $\hat{M}^{\prime}$ that is returned at line 13 of Algorithm~\ref{alg:OR} we have $(\hat{M}^{\prime},\hat{s}) \models \phi_{1}$ or $(\hat{M}^{\prime},\hat{s}) \models \phi_{2}$ and therefore $(\hat{M}^{\prime},\hat{s}) \models \phi$.   

\item if $\phi = EX\phi_{1}$, from the 3-valued semantics of CTL a KMTS that satisfies $\phi$ at $\hat{s}$ exists if and only if there is KMTS satisfying
$\phi_{1}$ at some direct must-successor of $\hat{s}$.

If in the KMTS $\hat{M}$ there is a state $\hat{s}_{1}$ with $(\hat{M},\hat{s}_{1}) \models \phi_{1}$, then the new KMTS $\hat{M}^{\prime} = AddMust(\hat{M},(\hat{s},\hat{s}_{1}))$ is computed at line 3 of Algorithm~\ref{alg:EX}.  Since $C$ is empty $\hat{M}^{\prime}$ is returned at line 5 and $(\hat{M^{\prime}},\hat{s}) \models EX\phi_{1}$.    

Otherwise, if there is a direct must-successor $\hat{s}_{i}$ of $\hat{s}$, $AbstractRepair(\hat{M},\hat{s}_{i},\phi_{1},\emptyset)$ is called at line 8.  From the induction hypothesis, if there is a KMTS $\hat{M}^{\prime\prime}$ with $(\hat{M}^{\prime\prime},\hat{s}_{i}) \models \phi_{1}$, then a KMTS $\hat{M}^{\prime}$ is computed such that 
$(\hat{M}^{\prime},\hat{s}_{i}) \models \phi_{1}$ and therefore the theorem is true.     

If there are no must-successors of $\hat{s}$, a new state $\hat{s}_{n}$ is added 
and subsequently connected with a must-transition from $\hat{s}$.  $AbstractRepair$ is then
called for $\phi_{1}$ and $\hat{s}_{n}$ as previously and the theorem holds also true.  
 
\item if $\phi = AG\phi_{1}$, from the 3-valued semantics of CTL a KMTS that satisfies $\phi$ at $\hat{s}$ exists, if and only if there is KMTS satisfying $\phi_{1}$ at $\hat{s}$ and at each may-reachable state from $\hat{s}$.  

$AbstractRepair(\hat{M},\hat{s},\phi_{1},\emptyset)$ is called at line 2 of Algorithm~\ref{alg:AG} and from the induction hypothesis if there is KMTS $\hat{M}_{0}^{\prime}$ with $(\hat{M}_{0}^{\prime},\hat{s}) \models \phi_{1}$, then a KMTS $\hat{M}_{0}$ is computed such that 
$(\hat{M}_{0},\hat{s}) \models \phi_{1}$.  $AbstractRepair$ is subsequently called for $\phi_{1}$ and for all may-reachable $\hat{s}_{k}$ from $\hat{s}$ with $(\hat{M}_{0},\hat{s}_{k}) \not\models \phi_{1}$ one-by-one.  
From the induction hypothesis, if there is KMTS $\hat{M}_{i}^{\prime}$ that satisfies $\phi_{1}$ at each such $\hat{s}_{k}$, then all $\hat{M}_{i}=AbstractRepair(\hat{M}_{i-1},\hat{s}_{k},\phi_{1},\emptyset), \, i=1, . . .,$ satisfy $\phi_{1}$ at $\hat{s}_{k}$ and the theorem holds true.
\end{itemize} 
\noindent We prove the theorem for all other cases in the same way.    
\end{proof}

\noindent Theorem~\ref{theor:complete} shows that \emph{AbstractRepair} 
is \emph{semi-complete} with respect to full CTL: if there is a KMTS that satisfies a \emph{mr}-CTL formula $\phi$, then the algorithm finds one such KMTS.  

\subsection{Complexity Issues}
\label{subsec:alg_complex}

AMR's complexity analysis is restricted to \emph{mr}-CTL, for which the algorithm has been proved 
complete.  For these formulas, we show that AMR is upper bounded by a polynomial 
expression in the state space size and the number of may-transitions of the abstract KMTS, 
and depends also on the length of the \emph{mr}-CTL formula.

For CTL formulas with nested 
path quantifiers and/or conjunction, AMR is looking for a repaired 
model satisfying all conjunctives (constraints), which increases the worst-case execution time exponentially
to the state space size of the abstract KMTS.  In general, as shown 
in~\cite{BK12}, the complexity
of all model repair algorithms gets worse when raising the level of their completeness, but AMR
has the advantage of working exclusively over an abstract model with a reduced state space compared to its concrete counterpart.    
 
Our complexity analysis for \emph{mr}-CTL is based on the following results.  
For an abstract KMTS $\hat{M} = (\hat{S}, \hat{S_{0}},$ $R_{must}, R_{may}, 
\hat{L})$ and a \emph{mr}-CTL property $\phi$, (i) 3-valued CTL model checking 
is performed in $O(|\phi| \cdot (|\hat{S}|+|R_{may}|))$~\cite{GHJ01}, (ii) Depth 
First Search (DFS) of states reachable from $\hat{s} \in \hat{S}$ is performed 
in $O(|\hat{S}|+|R_{may}|)$ in the worst case or in $O(|\hat{S}|+|R_{must}|)$ 
when only must-transitions are accessed, (iii) finding a maximal path from 
$\hat{s} \in \hat{S}$ using Breadth First Search (BFS) is performed in 
$O(|\hat{S}|+|R_{may}|)$ for may-paths and in $O(|\hat{S}|+|R_{must}|)$ for 
must-paths.     

We analyze the computational cost for each of the AMR's primitive functions:  
\begin{itemize}
\item if $\phi = p \in AP$, $AbstractRepair_{ATOMIC}$ 
is called and the operation $ChangeLabel$ is applied, which is in $O(1)$.    

\item if $\phi = EX\phi_{1}$, then $AbstractRepair_{EX}$ is called and
the applied operations with the highest cost are: (1) finding a state satisfying
$\phi_{1}$, which depends on the cost of 3-valued CTL model checking and
is in $O(|\hat{S}| \cdot |\phi_{1}| \cdot (|\hat{S}|+|R_{may}|))$, (2) finding a 
must-reachable state, which is in $O(|\hat{S}| + |R_{must}|)$.  These operations 
are called at most once and the overall complexity for this primitive functions is 
therefore in $O(|\hat{S}| \cdot |\phi_{1}| \cdot (|\hat{S}|+|R_{may}|))$.
  
\item if $\phi = AX\phi_{1}$, then $AbstractRepair_{AX}$ is called and the most 
costly operations are: (1) finding a may-reachable state, which is in 
$O(|\hat{S}| + |R_{may}|)$, and (2) checking if a state satisfies $\phi_{1}$, which is in 
$O(|\phi_{1}| \cdot (|\hat{S}|+|R_{may}|))$.  These operations are called at most 
$|\hat{S}|$ times and the overall bound class is  
$O(|\hat{S}| \cdot |\phi_{1}| \cdot (|\hat{S}|+|R_{may}|))$.  

\item if $\phi = EF\phi_{1}$, $AbstractRepair_{EF}$ is called and the 
operations with the highest cost are: (1) finding a must-reachable 
state, which is in $O(|\hat{S}| + |R_{must}|)$, (2) checking if a state satisfies $\phi_{1}$ with 
its bound class being $O(|\phi_{1}| \cdot (|\hat{S}|+|R_{may}|))$ and (3) finding a state 
that satisfies $\phi_{1}$, which is in $O(|\hat{S}| \cdot |\phi_{1}| \cdot (|\hat{S}|+|R_{may}|))$.  
These three operations are called at most $|\hat{S}|$ times and consequently, the 
overall bound class is $O(|\hat{S}|^{2} \cdot |\phi_{1}| \cdot (|\hat{S}|+|R_{may}|))$.  

\item if $\phi = AF\phi_{1}$, $AbstractRepair_{AF}$ is called and the most costly operation 
is: finding a maximal may-path violating 
$\phi_{1}$ in all states, which is in $O(|\hat{S}| \cdot |\phi_{1}| \cdot (|\hat{S}|+|R_{may}|)$.  
This operation is called at most $|\hat{S}|$ times and therefore, the overall bound class is
$O(|\hat{S}|^2 \cdot |\phi_{1}| \cdot (|\hat{S}|+|R_{may}|))$.  
\end{itemize}
In the same way, it is easy to show that: (i) if $\phi = EG\phi_{1}$, then $AbstractRepair_{EG}$ is 
in $O(|\hat{S}| \cdot |\phi_{1}| \cdot (|\hat{S}|+|R_{must}|)$, 
(ii) if $\phi = AG\phi_{1}$, then $AbstractRepair_{AG}$ is in 
$O(|\hat{S}| \cdot |\phi_{1}| \cdot (|\hat{S}|+|R_{may}|))$, (iii) if $\phi = E(\phi_{1}U\phi_{2})$, 
then the bound class of $AbstractRepair_{EU}$ is $O(|\hat{S}| \cdot |\phi_{1}| \cdot (|\hat{S}|+|R_{must}|)$,
(iv) if $\phi = A(\phi_{1}U\phi_{2})$ then $AbstractRepair_{AU}$ is in 
$O(|\hat{S}|^2 \cdot |\phi_{1}| \cdot (|\hat{S}|+|R_{may}|))$.  

For a \emph{mr}-CTL property $\phi$, the main body of the algorithm is called at most $|\phi|$ times 
and the overall bound class of the AMR algorithm is $O(|\hat{S}|^2 \cdot |\phi|^{2} \cdot (|\hat{S}|+|R_{may}|))$.

\subsection{Application}
\label{sec:app}
We present the application of \emph{AbstractRepair} 
on the ADO system from Section~\ref{sec:mc}.    
After the first two steps of our repair process,  
\emph{AbstractRepair} 
is called for the KMTS $\alpha_{\mathit{Refined}}(M)$ that 
is shown in Fig.~\ref{fig:ado_refined}, 
the state $\hat{s}_{01}$ and the CTL property $\phi = AGEXq$.  

\emph{AbstractRepair} calls $AbstractRepair_{AG}$ with 
arguments $\alpha_{\mathit{Refined}}(M)$, $\hat{s}_{01}$ and $AGEXq$.  
The $AbstractRepair_{AG}$ algorithm at line 10 triggers a 
recursive call of \emph{AbstractRepair} with the same 
arguments.  Eventually,  $AbstractRepair_{EX}$ is called 
with arguments $\alpha_{\mathit{Refined}}(M)$, $\hat{s}_{01}$ and 
$EXq$, that in turn calls \emph{AddMust} 
at line 3, thus adding a must-transition from 
$\hat{s}_{01}$ to $\hat{s}_{1}$.  
\emph{AbstractRepair} terminates by returning a KMTS 
$\hat{M^{\prime}}$ that satisfies $\phi = AGEXq$.  
The repaired KS $M^{\prime}$ is the single element  
in the set of KSs derived by the concretization 
of $\hat{M^{\prime}}$ (cf. Def.~\ref{def:add_must_ks}).  The execution steps of 
\emph{AbstractRepair} and the obtained 
repaired KMTS and KS are shown in 
Fig.~\ref{fig:ado_repair_process} and 
Fig.~\ref{fig:ado_repaired} respectively. 


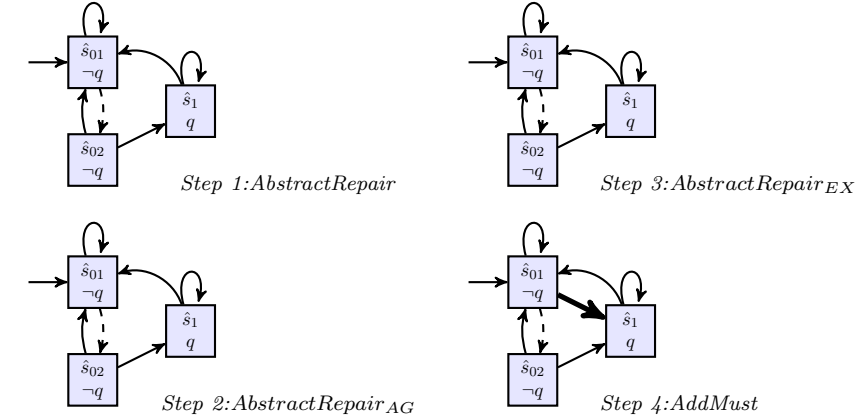
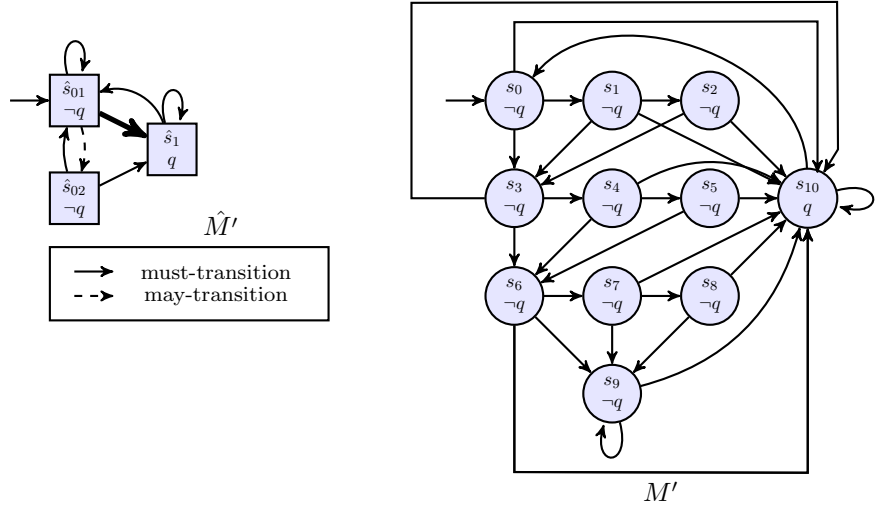
\begin{figure}[htb]
\centering
\subfloat[Application of \emph{AbstractRepair}.]
{\label{fig:ado_repair_process}\begin{tikzpicture}[->,>=stealth',auto,node 
distance=2cm, scale=0.65, thick, main node/.style={scale=0.65, minimum size = 
1cm, align=center,circle,fill=blue!10,draw}, abs node/.style={scale=0.65, 
minimum size = 1cm, align=center,rectangle,fill=blue!10,draw}]

\begin{scope}
  \node[abs node] (12) {$\hat{s}_{01}$ \\ $\neg q$};
  \node[abs node] (13) [below of=12] {$\hat{s}_{02}$ \\ $\neg q$};
  \node[abs node, yshift=1cm] (14) [right of=13] {$\hat{s}_1$ \\ $q$};

     \path
     (12) edge [loop above] (12)
     (14) edge [loop above] (14)
     (12) edge [bend left=15, dashed] (13)
     (14) edge [out=135, in=45,bend right=45] (12)
     (13) edge (14)
     (13) edge [bend left=15] (12);

     \draw[->] ([xshift=-.8cm]12.west) --  (12.west);

     \node[font = \small] at (4, -2.5) {\scriptsize{\textit{Step 2:}$AbstractRepair_{AG}$}};
  \end{scope}

\begin{scope}[xshift=9cm]
  \node[abs node] (12) {$\hat{s}_{01}$ \\ $\neg q$};
  \node[abs node] (13) [below of=12] {$\hat{s}_{02}$ \\ $\neg q$};
  \node[abs node, yshift=1cm] (14) [right of=13] {$\hat{s}_1$ \\ $q$};

     \path
     (12) edge [loop above] (12)
     (14) edge [loop above] (14)
     (12) edge [bend left=15, dashed] (13)
     (12) edge [line width = .7mm] (14)
     (14) edge [out=135, in=45,bend right=45] (12)
     (13) edge (14)
     (13) edge [bend left=15] (12);

     \draw[->] ([xshift=-.8cm]12.west) --  (12.west);

     \node[font = \small] at (3, -2.5) {\scriptsize{\textit{Step 4:AddMust}}};

  \end{scope}
  
\begin{scope}[yshift=4.5cm]
  \node[abs node] (12) {$\hat{s}_{01}$ \\ $\neg q$};
  \node[abs node] (13) [below of=12] {$\hat{s}_{02}$ \\ $\neg q$};
  \node[abs node, yshift=1cm] (14) [right of=13] {$\hat{s}_1$ \\ $q$};

     \path
     (12) edge [loop above] (12)
     (14) edge [loop above] (14)
     (12) edge [bend left=15, dashed] (13)
     (14) edge [out=135, in=45,bend right=45] (12)
     (13) edge (14)
     (13) edge [bend left=15] (12);

     \draw[->] ([xshift=-.8cm]12.west) --  (12.west);

     \node[font = \small] at (4, -2.5) {\scriptsize{\textit{Step 1:AbstractRepair}}};

  \end{scope}
  
      \begin{scope}[xshift=9cm,yshift=4.5cm]
  \node[abs node] (12) {$\hat{s}_{01}$ \\ $\neg q$};
  \node[abs node] (13) [below of=12] {$\hat{s}_{02}$ \\ $\neg q$};
  \node[abs node, yshift=1cm] (14) [right of=13] {$\hat{s}_1$ \\ $q$};

     \path
     (12) edge [loop above] (12)
     (14) edge [loop above] (14)
     (12) edge [bend left=15, dashed] (13)
     (14) edge [out=135, in=45,bend right=45] (12)
     (13) edge (14)
     (13) edge [bend left=15] (12);

     \draw[->] ([xshift=-.8cm]12.west) --  (12.west);

     \node[font = \small] at (4, -2.5) {\scriptsize{\textit{Step 3:}$AbstractRepair_{EX}$}};

  \end{scope}

\end{tikzpicture}}          
\hfill
\subfloat[The repaired KMTS and KS.]
{\label{fig:ado_repaired}\begin{tikzpicture}[->,>=stealth',auto,node 
distance=2cm, scale=0.65, thick, main node/.style={scale=0.65, minimum size = 
1cm, align=center,circle,fill=blue!10,draw}, abs node/.style={scale=0.65, 
minimum size = 1cm, align=center,rectangle,fill=blue!10,draw}]

\begin{scope}
  \node[abs node] (12) {$\hat{s}_{01}$ \\ $\neg q$};
  \node[abs node] (13) [below of=12] {$\hat{s}_{02}$ \\ $\neg q$};
  \node[abs node, yshift=1cm] (14) [right of=13] {$\hat{s}_1$ \\ $q$};

     \path
     (12) edge [loop above] (12)
     (14) edge [loop above] (14)
     (12) edge [bend left=15, dashed] (13)
     (12) edge [line width = .7mm] (14)
     (14) edge [out=135, in=45,bend right=45] (12)
     (13) edge (14)
     (13) edge [bend left=15] (12);

     \draw[->] ([xshift=-.8cm]12.west) --  (12.west);

     \node[font = \small] at (3, -2.5) {$\hat{M^{\prime}}$};

     \draw (-.5, -3) rectangle (5.2, -4.5);
     \node[font = \scriptsize] at (2.9, -3.5) {must-transition};
     \node[font = \scriptsize] at (2.9, -4) {may-transition};
     \draw[->] (0, -3.5) -- (.8, -3.5);
     \draw[->, dashed] (0, -4) -- (.8, -4);
  \end{scope}

\begin{scope}[xshift=9cm]
  \node[main node] (1) {$s_0$ \\ $\neg q$};
  \node[main node] (2) [below of=1] {$s_3$ \\ $\neg q$};
  \node[main node] (3) [below of=2] {$s_6$ \\ $\neg q$};
  \node[main node] (4) [right of=1] {$s_1$ \\ $\neg q$};
  \node[main node] (5) [below of=4] {$s_4$ \\ $\neg q$};
  \node[main node] (6) [below of=5] {$s_7$ \\ $\neg q$};
  \node[main node] (7) [right of=4] {$s_2$ \\ $\neg q$};
  \node[main node] (8) [below of=7] {$s_5$ \\ $\neg q$};
  \node[main node] (9) [below of=8] {$s_8$ \\ $\neg q$};
  \node[main node] (10) [right of=8] {$s_{10}$ \\ $q$};
  \node[main node] (11) [below of=6] {$s_9$ \\ $\neg q$};

   \path
     (1) edge (4)
     (4) edge (7)
     (2) edge (5)
     (5) edge (8)
     (3) edge (6)
     (6) edge (9)
     (7) edge (10)
     (8) edge (10)
     (9) edge (10)
     (3) edge (11)
     (6) edge (11)
     (9) edge (11)

     (1) edge (2)
     (4) edge (2)
     (7) edge (2)

     (2) edge (3)
     (5) edge (3)
     (8) edge (3)

     (10) edge [bend right=70] node {} (1)

     (4) edge (10)
     (6) edge (10)
     (5) edge [bend left=30] node {} (10)
     (11) edge [bend right=30] node {} (10)

     (10) edge [loop right] (10)
     (11) edge [loop below] (11);
     \draw[->] ([xshift=-.8cm]1.west) --  (1.west);
     
     \draw[->] (1) -- ([yshift=1cm]1.north) -- 
([xshift=2mm, yshift=3cm]10.north) -- ([xshift=2mm, yshift=-1mm]10.north);

     \draw[->] (2) -- ([xshift=-1.5cm]2.west) -- 
([xshift=-1.5cm, yshift=4cm]2.west) -- 
([xshift=6mm, yshift=3.4cm]10.north) -- ([xshift=0mm, yshift=10mm]10.east) -- 
(10);

     \draw[->] (3) -- ([yshift=-3cm]3.south) -- 
([yshift=-5cm]10.south) -- (10.south);

     \draw[->] (3) -- ([yshift=-3cm]3.south) -- 
([yshift=-5cm]10.south) -- (10.south);
     \node[font = \small] at (3, -8) {$M'$};
\end{scope}

\end{tikzpicture}}
\caption{Repair of ADO system using abstraction.}
\label{fig:ado_repair}
\end{figure}

\noindent Although the ADO is not a system with a large state space, 
it is shown that the repair process is accelerated by the proposed 
use of abstraction.    
If on the other hand model repair was applied directly to the 
concrete model, new transitions would have have been inserted from
all the states labeled with $\neg open$ to the one labeled with \emph{open}.  
In the ADO, we have seven such states, but in a system with a large state 
space this number can be significantly higher.  The repair of such a model 
without the use of abstraction would be impractical.  

\section{Experimental Results: The Andrew File System 1 (AFS1) Protocol}
\label{sec:exp}

In this section, we provide experimental results for the relative performance of a
prototype implementation of our AMR algorithm in comparison with a prototype implementation
of a concrete model repair solution~\cite{ZD08}. The results serve as a proof of concept 
for the use of abstraction in model repair and demonstrate the practical utility of our approach.  

As a model we use a KS for the Andrew File System Protocol 1 (AFS1)~\cite{WV95},  
which has been repaired for a specific property in~\cite{ZD08}. AFS1 is a client-server cache coherence protocol for a distributed file system.  
Four values are used for the client's belief about a file (nofile, valid, invalid, 
suspect) and three values for the server's belief (valid, invalid, none).  

A property which is not satisfied in the AFS1 protocol 
in the form of CTL is:
\[
AG((Server.belief = valid) \rightarrow (Client.belief = valid))
\]
 
\begin{figure}[H]
\centering
\subfloat[The KS after the final refinement step.]
{\label{fig:afs1_refined2_ks}\begin{tikzpicture}[->,>=stealth',auto,node 
distance=2.5cm, scale=0.45, thick, main 
node white/.style={font=\footnotesize, scale=0.75, 
align=center,ellipse,fill=white,draw=black, minimum width=1.8cm}, main 
node red/.style={font=\footnotesize, scale=0.75, 
align=center,ellipse,fill=red!20,draw=black, minimum width=1.8cm},
main node blue/.style={font=\footnotesize, scale=0.75, 
align=center,ellipse,fill=blue!10,draw=black, minimum width=1.8cm},
main node orange/.style={font=\footnotesize, scale=0.75, 
align=center,ellipse,fill=orange!20,draw=black, minimum width=1.8cm},
main node gray/.style={font=\footnotesize, scale=0.75, 
align=center,ellipse,fill=gray!20,draw=black, minimum width=1.8cm}
]

\node[main node white] (11) at (0, 0) {$s_{11}$ \\ $\neg p \wedge \neg q$};
\node[main node white] (12) [right of=11] {$s_{12}$ \\ $\neg p \wedge \neg q$};
\node[main node blue] (17) [below of=11] {$s_{17}$ \\ $\neg p \wedge \neg q$};
\node[main node blue] (18) [right of=17] {$s_{18}$ \\ $\neg p \wedge \neg q$};
\node[main node red] (19) [below of=17] {$s_{19}$ \\ $p \wedge \neg q$};
\node[main node red] (20) [right of=19] {$s_{20}$ \\ $p \wedge \neg q$};

\node[main node orange] (3) [right of=12] {$s_{3}$ \\ $\neg p \wedge \neg q$};
\node[main node orange] (4) [right of=3] {$s_{4}$ \\ $\neg p \wedge \neg q$};
\node[main node white] (1) [below of=3] {$s_{1}$ \\ $\neg p \wedge \neg q$};
\node[main node white] (2) [right of=1] {$s_{2}$ \\ $\neg p \wedge \neg q$};
\node[main node orange] (22) [below of=1] {$s_{22}$ \\ $\neg p \wedge \neg q$};
\node[main node blue] (21) [right of=22] {$s_{21}$ \\ $\neg p \wedge \neg q$};

\node[main node red] (7) [right of=4] {$s_{7}$ \\ $p \wedge \neg q$};
\node[main node red] (8) [right of=7] {$s_{8}$ \\ $p \wedge \neg q$};
\node[main node gray] (9) [below of=7] {$s_{9}$ \\ $p \wedge q$};
\node[main node gray] (10) [right of=9] {$s_{10}$ \\ $p \wedge q$};

\node[main node orange, xshift=-1cm] (6) [above of=4] {$s_{6}$ \\ $\neg p 
\wedge \neg q$};
\node[main node white] (13) [above of=6] {$s_{13}$ \\ $\neg p \wedge \neg q$};

\node[main node blue, xshift=-1cm] (5) [above of=8] {$s_{5}$ \\ $\neg p 
\wedge \neg q$};
\node[main node white] (14) [above of=5] {$s_{14}$ \\ $\neg p \wedge \neg q$};

\node[main node red] (23) [below of=22] {$s_{23}$ \\ $p \wedge \neg q$};
\node[main node red] (24) [right of=23] {$s_{24}$ \\ $p \wedge \neg q$};

\node[main node gray, yshift=-2.5cm] (25) [below of=19] {$s_{25}$ \\ $p \wedge 
q$};
\node[main node gray] (26) [right of=25] {$s_{26}$ \\ $p \wedge q$};

\node[main node gray, yshift=-6cm] (15) [below of=9] {$s_{15}$ \\ $p \wedge 
q$};
\node[main node gray] (16) [right of=15] {$s_{16}$ \\ $p \wedge q$};

      \path
      (15) edge [loop below] (15)
      (16) edge [loop below] (16)

      (13) edge (5)
      (13) edge (6)
      (14) edge (6)
      (14) edge (5)
      (6) edge (3)
      (6) edge (4)    
      (5) edge (7)
      (5) edge (8)    

      (11) edge (17)
      (11) edge (18)    
      (12) edge (17)
      (12) edge (18)    

      (17) edge (19)
      (17) edge (20)    
      (18) edge (19)
      (18) edge (20)    

      (3) edge (1)
      (3) edge (2)    
      (4) edge (1)
      (4) edge (2)    

      (1) edge (22)
      (1) edge (21)    
      (2) edge (22)
      (2) edge (21)    

      (7) edge (9)
      (7) edge (10)    
      (8) edge (9)
      (8) edge (10)    

      (9) edge (15)
      (9) edge (16)    
      (10) edge (15)
      (10) edge (16)    

      (21) edge (23)
      (21) edge (24)    
      (22) edge (25)
      (22) edge (26)   
      (20) edge (25)
      (20) edge (26)    
      (19) edge (25)
      (19) edge (26)    
      (23) edge (25)
      (23) edge (26)    
      (24) edge (25)
      (24) edge (26)    
      
      (26) edge (15)

      (15) edge [bend right] (16)
      (16) edge [bend right] (15);
      
      \draw[->] (25.south) -- ([yshift=-5mm]25.south) -- (15);
      \draw[->] (26.east) -- ([xshift=-15mm, yshift=5mm]16.north) -- (16);
      \draw[->] ([xshift=-2mm]25.south) -- ([xshift=-2mm, 
yshift=-30mm]25.south) -- ([xshift=-20mm, yshift=-15mm]16.south) -- (16);
      
      \draw[->] ([yshift=8mm]11.north) -- (11);
      \draw[->] ([yshift=8mm]12.north) -- (12);
      \draw[->] ([xshift=-8mm]13.west) -- (13);
      \draw[->] ([xshift=-8mm]14.west) -- (14);
      
%

\end{tikzpicture}} 
               
\subfloat[The refined KMTS.]
{\label{fig:afs1_refined2_kmts}\begin{tikzpicture}[->,>=stealth',auto,node 
distance=2cm, scale=0.5, thick, abs node/.style={font=\footnotesize, 
rectangle,fill=blue!10, text centered, text width=1.2cm, draw=black, minimum 
height = 1cm}]

  \node[abs node, fill=red!20] (12) at (0, 0) {$p \wedge \neg q$};
  \node[abs node, fill=white] (13) at (0, 6) {$\neg p \wedge \neg q$};
  \node[abs node] (14) at (3, 3) {$\neg p \wedge \neg q$};
  \node[abs node, fill=gray!20] (15) at (6, 0) {$p \wedge q$};
  \node[abs node, fill=orange!20] (16) at (6, 6) {$\neg p \wedge \neg q$};

  \draw[->] ([xshift=-8mm]13.west) -- (13);

     \path
     (12) edge [loop below] (12)
     (15) edge [loop below] (15)
     (16) edge [loop above, dashed] (16)
    
     (12) edge (15)
     (14) edge (12)
     (13) edge (14)
     (16) edge [bend right, dashed] (13)    
     (13) edge [bend right, dashed] (16)   
     (16) edge [dashed] (15);    

     \draw[->] ([xshift=-.8cm]15.west) --  (15.west);

\end{tikzpicture}}
\caption{The KS and the KMTS of the AFS1 protocol after the 2nd refinement 
step.}
\label{fig:afs1_ks_kmts}
\end{figure}
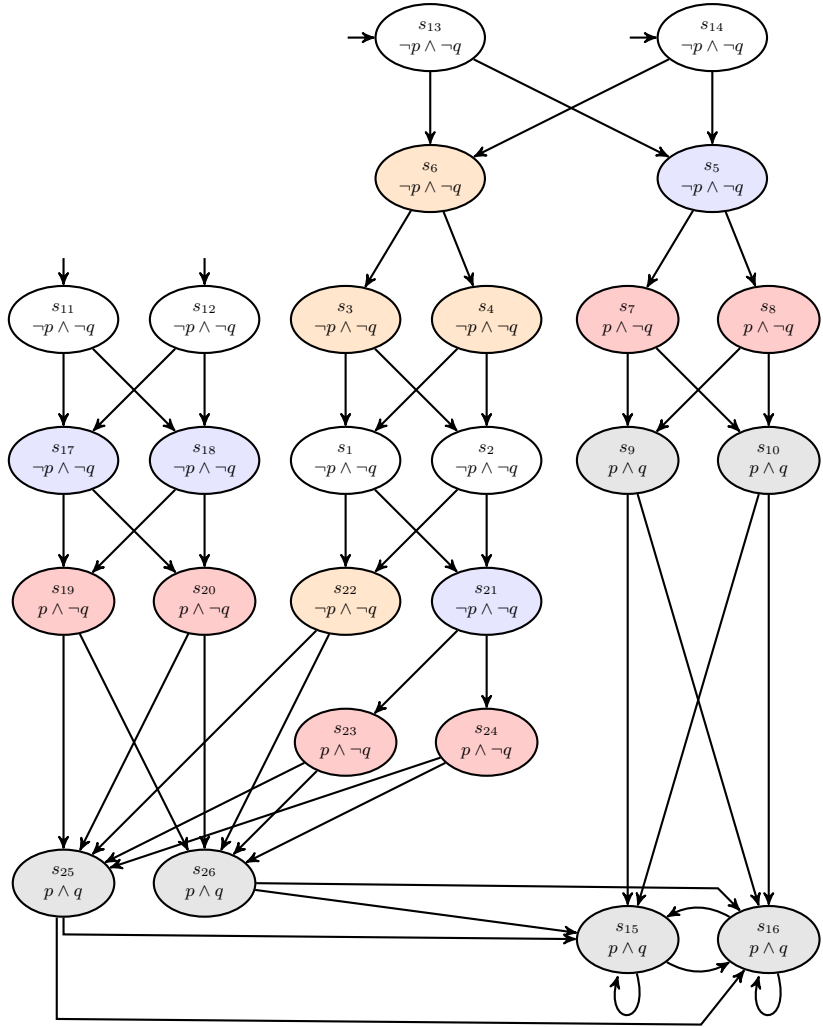
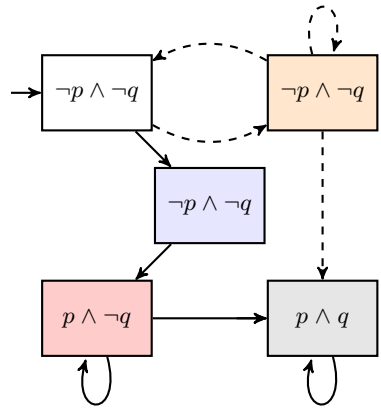

\begin{figure}[H]
\centering
\subfloat[The repaired KMTS.]
{\label{fig:afs1_repaired_kmts}\begin{tikzpicture}[->,>=stealth',auto,node 
distance=2cm, scale=0.5, thick, abs node/.style={font=\footnotesize, 
rectangle,fill=blue!10, text centered, text width=1.2cm, draw=black, minimum 
height = 1cm}]

  \node[abs node, fill=red!20] (12) at (0, 0) {$\neg p \wedge \neg q$};
  \node[abs node, fill=white] (13) at (0, 6) {$\neg p \wedge \neg q$};
  \node[abs node] (14) at (3, 3) {$\neg p \wedge \neg q$};
  \node[abs node, fill=gray!20] (15) at (6, 0) {$p \wedge q$};
  \node[abs node, fill=orange!20] (16) at (6, 6) {$\neg p \wedge \neg q$};

  \draw[->] ([xshift=-8mm]13.west) -- (13);

     \path
     (12) edge [loop below] (12)
     (15) edge [loop below] (15)
     (16) edge [loop above, dashed] (16)
    
     (12) edge (15)
     (14) edge (12)
     (13) edge (14)
     (16) edge [bend right, dashed] (13)    
     (13) edge [bend right, dashed] (16)   
     (16) edge [dashed] (15);    

     \draw[->] ([xshift=-.8cm]15.west) --  (15.west);

\end{tikzpicture}}                

\subfloat[The repaired KS.]
{\label{fig:afs1_repaired_ks}\begin{tikzpicture}[->,>=stealth',auto,node 
distance=2.5cm, scale=0.45, thick, main 
node white/.style={font=\footnotesize, scale=0.75, 
align=center,ellipse,fill=white,draw=black, minimum width=1.8cm}, main 
node red/.style={font=\footnotesize, scale=0.75, 
align=center,ellipse,fill=red!20,draw=black, minimum width=1.8cm},
main node blue/.style={font=\footnotesize, scale=0.75, 
align=center,ellipse,fill=blue!10,draw=black, minimum width=1.8cm},
main node orange/.style={font=\footnotesize, scale=0.75, 
align=center,ellipse,fill=orange!20,draw=black, minimum width=1.8cm},
main node gray/.style={font=\footnotesize, scale=0.75, 
align=center,ellipse,fill=gray!20,draw=black, minimum width=1.8cm}
]

\node[main node white] (11) at (0, 0) {$s_{11}$ \\ $\neg p \wedge \neg q$};
\node[main node white] (12) [right of=11] {$s_{12}$ \\ $\neg p \wedge \neg q$};
\node[main node blue] (17) [below of=11] {$s_{17}$ \\ $\neg p \wedge \neg q$};
\node[main node blue] (18) [right of=17] {$s_{18}$ \\ $\neg p \wedge \neg q$};
\node[main node red] (19) [below of=17] {$s_{19}$ \\ $\neg p \wedge \neg q$};
\node[main node red] (20) [right of=19] {$s_{20}$ \\ $\neg p \wedge \neg q$};

\node[main node orange] (3) [right of=12] {$s_{3}$ \\ $\neg p \wedge \neg q$};
\node[main node orange] (4) [right of=3] {$s_{4}$ \\ $\neg p \wedge \neg q$};
\node[main node white] (1) [below of=3] {$s_{1}$ \\ $\neg p \wedge \neg q$};
\node[main node white] (2) [right of=1] {$s_{2}$ \\ $\neg p \wedge \neg q$};
\node[main node orange] (22) [below of=1] {$s_{22}$ \\ $\neg p \wedge \neg q$};
\node[main node blue] (21) [right of=22] {$s_{21}$ \\ $\neg p \wedge \neg q$};

\node[main node red] (7) [right of=4] {$s_{7}$ \\ $\neg p \wedge \neg q$};
\node[main node red] (8) [right of=7] {$s_{8}$ \\ $\neg p \wedge \neg q$};
\node[main node gray] (9) [below of=7] {$s_{9}$ \\ $p \wedge q$};
\node[main node gray] (10) [right of=9] {$s_{10}$ \\ $p \wedge q$};

\node[main node orange, xshift=-1cm] (6) [above of=4] {$s_{6}$ \\ $\neg p 
\wedge \neg q$};
\node[main node white] (13) [above of=6] {$s_{13}$ \\ $\neg p \wedge \neg q$};

\node[main node blue, xshift=-1cm] (5) [above of=8] {$s_{5}$ \\ $\neg p 
\wedge \neg q$};
\node[main node white] (14) [above of=5] {$s_{14}$ \\ $\neg p \wedge \neg q$};

\node[main node red] (23) [below of=22] {$s_{23}$ \\ $\neg p \wedge \neg q$};
\node[main node red] (24) [right of=23] {$s_{24}$ \\ $\neg p \wedge \neg q$};

\node[main node gray, yshift=-2.5cm] (25) [below of=19] {$s_{25}$ \\ $p \wedge 
q$};
\node[main node gray] (26) [right of=25] {$s_{26}$ \\ $p \wedge q$};

\node[main node gray, yshift=-6cm] (15) [below of=9] {$s_{15}$ \\ $p \wedge 
q$};
\node[main node gray] (16) [right of=15] {$s_{16}$ \\ $p \wedge q$};

      \path
      (15) edge [loop below] (15)
      (16) edge [loop below] (16)

      (13) edge (5)
      (13) edge (6)
      (14) edge (6)
      (14) edge (5)
      (6) edge (3)
      (6) edge (4)    
      (5) edge (7)
      (5) edge (8)    

      (11) edge (17)
      (11) edge (18)    
      (12) edge (17)
      (12) edge (18)    

      (17) edge (19)
      (17) edge (20)    
      (18) edge (19)
      (18) edge (20)    

      (3) edge (1)
      (3) edge (2)    
      (4) edge (1)
      (4) edge (2)    

      (1) edge (22)
      (1) edge (21)    
      (2) edge (22)
      (2) edge (21)    

      (7) edge (9)
      (7) edge (10)    
      (8) edge (9)
      (8) edge (10)    

      (9) edge (15)
      (9) edge (16)    
      (10) edge (15)
      (10) edge (16)    

      (21) edge (23)
      (21) edge (24)    
      (22) edge (25)
      (22) edge (26)   
      (20) edge (25)
      (20) edge (26)    
      (19) edge (25)
      (19) edge (26)    
      (23) edge (25)
      (23) edge (26)    
      (24) edge (25)
      (24) edge (26)    
      
      (26) edge (15)

      (15) edge [bend right] (16)
      (16) edge [bend right] (15);
      
      \draw[->] (25.south) -- ([yshift=-5mm]25.south) -- (15);
      \draw[->] (26.east) -- ([xshift=-15mm, yshift=5mm]16.north) -- (16);
      \draw[->] ([xshift=-2mm]25.south) -- ([xshift=-2mm, 
yshift=-30mm]25.south) -- ([xshift=-20mm, yshift=-15mm]16.south) -- (16);
      
      \draw[->] ([yshift=8mm]11.north) -- (11);
      \draw[->] ([yshift=8mm]12.north) -- (12);
      \draw[->] ([xshift=-8mm]13.west) -- (13);
      \draw[->] ([xshift=-8mm]14.west) -- (14);
      
%

\end{tikzpicture}}
\caption{The repaired KMTS and KS of the AFS1 protocol.}
\label{fig:afs1_repaired_ks_kmts}
\end{figure}
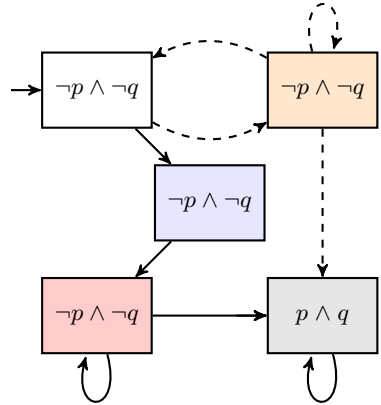
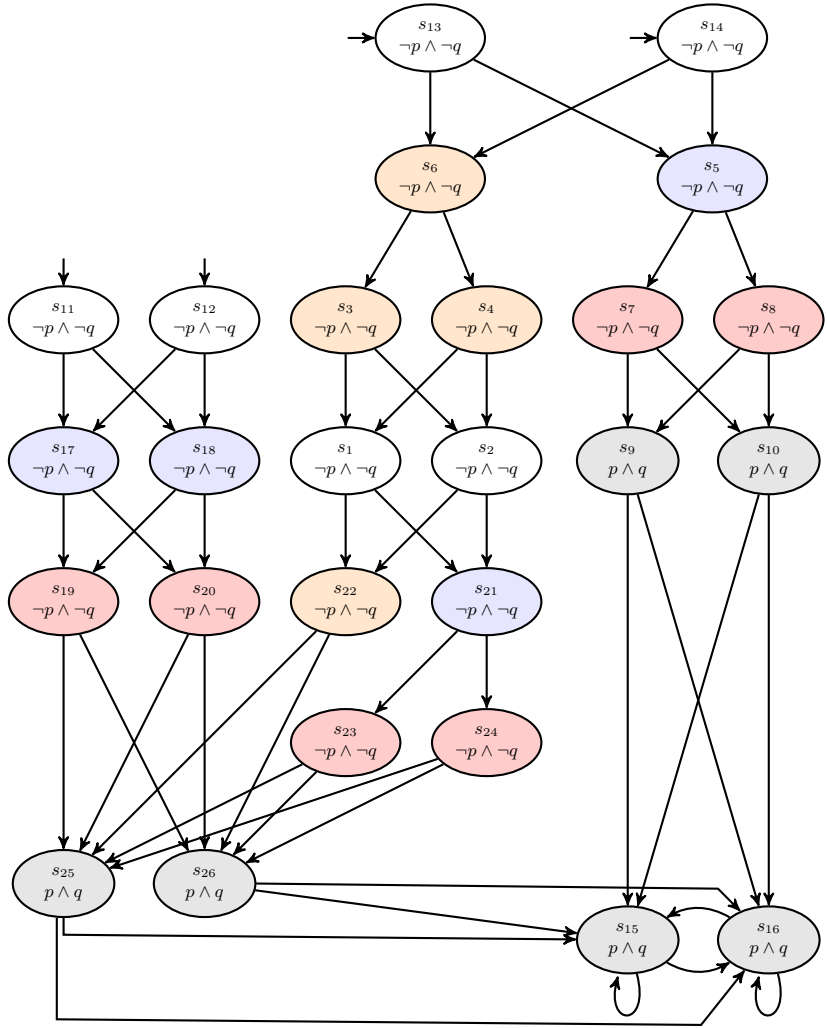

We define the atomic proposition $p$ as $Server.belief = valid$ and $q$ as 
$Client.belief = valid$, and the property is thus written as $AG(p \rightarrow 
q)$.  
The KS for the AFS1 protocol is 
depicted in Fig.~\ref{fig:afs1_refined2_ks}. State colors show how they are 
abstracted
in the KMTS of Fig.~\ref{fig:afs1_refined2_kmts}, which is derived after the 
2nd 
refinement 
step of our AMR framework (Fig.~\ref{fig:abs_repair}). The shown KMTS and the 
CTL property of interest are given as input in our prototype AMR 
implementation. 

To obtain larger models of AFS1 we have extended the original model by adding one more possible value for three model variables. Three new models are obtained with gradually increasing size of state space. 

The results of our experiments are presented in Table~\ref{table:exp_results}.  
The time needed for the AMR prototype to repair the original AFS1 model and its extensions 
is from 124 to even 836 times less than the needed time for concrete 
model repair.   
The repaired KMTS and KS for the original AFS1 model are shown in Fig.~\ref{fig:afs1_repaired_ks_kmts}.  

An interesting observation from the application of the AMR algorithm on the repair
of the AFS1 KS is that the distance $d$ (cf. Def.~\ref{def:metric_space}) of the  
repaired KS from the original KS is less than the corresponding
distance obtained from the concrete model repair algorithm in~\cite{ZD08}. This result 
demonstrates in practice the effect of the minimality of changes ordering, on which the AMR algorithm is based on (cf. Fig.~\ref{fig:order_basic_ops}). 

\begin{table}[t]
\begin{center}
    \begin{tabular}{ | p{4cm} | p{2cm} | p{2cm} | p{2cm} | p{2cm} |}
    \hline
    Models & Concrete States & Concr. Repair (Time in sec.) & AMR (Time in sec.) & Improvement (\%) \\ \hline
    $AFS1$ & $26$ & $17.4$ & $0.14$ & $124$ \\ \hline
    $AFS1 (Extension 1)$ & $30$ & $24.9$ & $0.14$ & $178$ \\ \hline
    $AFS1 (Extension 2)$ & $34$ & $35.0$ & $0.14$ & $250$ \\ \hline
    $AFS1 (Extension 3)$ & $38$ & $117.0$ & $0.14$ & $836$ \\ \hline
    \end{tabular}
\end{center}
\caption{Experimental results of AMR with respect to concrete repair}
\label{table:exp_results}
\end{table}

\section{Related Work}
\label{sec:relwork}
To the best of our knowledge this is the first work that suggests 
the use of abstraction as a means to counter the state space 
explosion in search of a Model Repair solution.  However, abstraction
and in particular abstract interpretation has been used in 
\emph{program synthesis}~\cite{VYY2010},  a different but related problem 
to the Model Repair. Program synthesis refers to the automatic generation of a program based on a given specification. Another related problem where abstraction has been
used is that of \emph{trigger querying}~\cite{AK14}: 
given a system $M$ and a formula $\phi$, find the set of scenarios that trigger 
$\phi$ in $M$.  

The related work in the area of \emph{program repair} do not consider KSs as 
the program model.  In this context, abstraction has been previously
used in the repair of data structures~\cite{ZMK13}.  The problem of repairing a 
Boolean program has been formulated in~\cite{SJB05,JGB07,GBC06,EJ12} as the 
finding of a winning strategy for a game between two players. The only exception 
is the work reported in~\cite{SDE08}.

Another line of research on program repair treats the repair
as a search problem and applies innovative evolutionary algorithms~\cite{A11}, 
\emph{behavioral programming} techniques~\cite{HKMW12} or other informal 
heuristics~\cite{WC08,AAG11,WPFSBMZ10}.   

Focusing exclusively on the area of Model Repair without the use of abstraction, 
it is worth to mention the following approaches. The first work on Model Repair 
with respect to CTL formulas was presented in~\cite{A95}.  
The authors used only the removal of transitions and showed that the problem is NP-complete.  
Another interesting early attempt to introduce the Model Repair problem for CTL properties 
is the work in~\cite{BEGL99}. The authors are based
on the AI techniques of abductive reasoning and theory revision and propose a repair 
algorithm with relatively high computational cost. A formal algorithm for Model Repair 
in the context of KSs and CTL is presented in~\cite{ZD08}.  The authors admit 
that their repair process strongly depends on the model's size and they do not attempt 
to provide a solution for handling conjunctive CTL formulas.  

In~\cite{CR09}, the authors try to render model repair applicable to large KSs by using ``table systems'', a concise representation of KSs that is implemented in the NuSMV model checker. A 
limitation of their approach is that table systems cannot represent all possible KSs.  
In~\cite{ZKZ10}, tree-like local model updates are introduced with  
the aim of making the repair process applicable to large-scale domains. However,
the proposed approach is only applicable to the universal fragment of the CTL.

A number of works attempt to ensure completeness for increasingly larger fragments of
the CTL by introducing ways of handling the constraints associated with conjunctive formulas. 
In~\cite{KPYZ10}, the authors propose the use of constraint automata for ACTL formulas, while
in~\cite{CR11} the authors introduce the use of protected models for an extension of the CTL.  
Both of the two methods are not directly applicable to formulas of the full CTL.

The Model Repair problem has been also addressed in many other contexts. In~\cite{E12}, 
the author uses a distributed algorithm and the processing power of computing clusters to fight the time and space complexity of the repair process. In~\cite{MLB11}, an extension of the Model Repair problem has been studied for Labeled Transition Systems. In~\cite{BGKRS11}, we have provided a solution for the Model Repair problem in probabilistic systems.  Another recent effort for repairing 
discrete-time probabilistic models has been proposed in~\cite{PAJTK15}.  In~\cite{BBG11}, model repair is applied to the \emph{fault recovery} of component-based models. Finally, a slightly different
but also related problem is that of Model Revision, which has been studied for UNITY properties 
in~\cite{BEK09,BK08-OPODIS} and for CTL in~\cite{GW10}. Other methods in the 
area of fault-tolerance include the work in~\cite{gr09}, which uses discrete 
controller synthesis and~\cite{fb15}, which employs SMT solving. Another 
interesting work in this direction is in~\cite{df09}, where the authors present 
a repair algorithm for fault-tolerance in a fully connected topology, with 
respect to a temporal specification.

\section{Conclusions}
\label{sec:concl}

In this paper, we have shown how abstraction can be used to cope with the state
explosion problem in Model Repair.  Our model-repair framework is based on
Kripke Structures, a 3-valued semantics for CTL, and Kripke Modal
Transition Systems, and features an abstract-model-repair algorithm for
KMTSs.  We have proved that our AMR algorithm is sound for the full CTL 
and complete for a subset of CTL.  We have also proved that our AMR 
algorithm is upper bounded by a polynomial expression in the 
size of the abstract model for a major fragment of CTL. To demonstrate 
its practical utility, we applied our framework to an Automatic Door 
Opener system and to the Andrew File System 1 protocol.  
  
As future work, we plan to apply our method to case studies with larger
state spaces, and investigate how abstract model repair can be used in
different contexts and domains.  A model repair application of high interest
is in the design of fault-tolerant systems. In~\cite{bka12}, the authors 
present an approach for the repair of a distributed algorithm such that the 
repaired one features fault-tolerance.  The input to this model repair problem 
includes a set of uncontrollable transitions such as the faults in the 
system.  The model repair algorithm used works on concrete models and it can 
therefore solve the problem only for a limited number of processes.  With this 
respect, we believe that this application could be benefited from the use of 
abstraction in our AMR framework. 

At the level of extending our AMR framework, we aim to search for ``better" abstract models, in order to either restrict failures due to refinement or ensure completeness for a larger fragment of the CTL.  We will also investigate different notions of minimality in the changes introduced by model repair and the applicability of abstraction-based model repair to probabilistic, hybrid and other
types of models.       

\section{Acknowledgment} 

This work was partially sponsored by Canada NSERC Discovery Grant 418396-2012 and NSERC Strategic Grants 430575-2012 and 463324-2014.
The research was also co-financed by the European Union (European Social Fund ESF) and Greek national funds through the Operational Program ``Education and Lifelong Learning'' of the National Strategic Reference Framework (NSRF) - Research Funding Program: Thalis Athens University of Economics and Business - SOFTWARE ENGINEERING RESEARCH PLATFORM.
 
\bibliographystyle{plain}
\bibliography{amr}

\begin{thebibliography}{10}

\bibitem{AAG11}
Thomas Ackling, Bradley Alexander, and Ian Grunert.
\newblock Evolving patches for software repair.
\newblock In {\em Proceedings of the 13th annual conference on Genetic and
  evolutionary computation}, GECCO '11, pages 1427--1434, New York, NY, USA,
  2011. ACM.

\bibitem{A95}
Marco Antoniotti.
\newblock {\em Synthesis and Verification of Discrete Controllers for Robotics
  and Manufacturing Devices with Temporal Logic and the Control-D System.}
\newblock PhD thesis, New York University, 1995.

\bibitem{A11}
Andrea Arcuri.
\newblock Evolutionary repair of faulty software.
\newblock {\em Appl. Soft Comput.}, 11:3494--3514, June 2011.

\bibitem{AK14}
Guy Avni and Orna Kupferman.
\newblock An abstraction-refinement framework for trigger querying.
\newblock {\em Formal Methods in System Design}, 44(2):149--175, 2014.

\bibitem{BK08}
Christel Baier and Joost-Pieter Katoen.
\newblock {\em Principles of Model Checking (Representation and Mind Series)}.
\newblock The MIT Press, 2008.

\bibitem{BGKRS11}
Ezio Bartocci, Radu Grosu, Panagiotis Katsaros, C.~R. Ramakrishnan, and
  Scott~A. Smolka.
\newblock Model repair for probabilistic systems.
\newblock In {\em Proceedings of the 17th international conference on Tools and
  algorithms for the construction and analysis of systems: part of the joint
  European conferences on theory and practice of software}, TACAS'11/ETAPS'11,
  pages 326--340, Berlin, Heidelberg, 2011. Springer-Verlag.

\bibitem{BBG11}
Borzoo Bonakdarpour, Marius Bozga, and Gregor Goessler.
\newblock A theory of fault recovery for component-based models.
\newblock In {\em Proceedings of the 2011 IEEE 30th International Symposium on
  Reliable Distributed Systems}, SRDS '11, pages 265--270, Washington, DC, USA,
  2011. IEEE Computer Society.

\bibitem{BEK09}
Borzoo Bonakdarpour, Ali Ebnenasir, and Sandeep~S. Kulkarni.
\newblock Complexity results in revising {UNITY} programs.
\newblock {\em ACM Trans. Auton. Adapt. Syst.}, 4:5:1--5:28, February 2009.

\bibitem{BK08-OPODIS}
Borzoo Bonakdarpour and Sandeep~S. Kulkarni.
\newblock Revising distributed {UNITY} programs is {NP}-complete.
\newblock In {\em Principles of Distributed Systems (OPODIS)}, pages 408--427,
  2008.

\bibitem{BK12}
Borzoo Bonakdarpour and Sandeep~S. Kulkarni.
\newblock Automated model repair for distributed programs.
\newblock {\em SIGACT News}, 43(2):85--107, jun 2012.

\bibitem{bka12}
Borzoo Bonakdarpour, Sandeep~S. Kulkarni, and Fuad Abujarad.
\newblock Symbolic synthesis of masking fault-tolerant programs.
\newblock {\em Springer Journal on Distributed Computing}, 25(1):83--108, March
  2012.

\bibitem{BEGL99}
Francesco Buccafurri, Thomas Eiter, Georg Gottlob, and Nicola Leone.
\newblock Enhancing model checking in verification by {AI} techniques.
\newblock {\em Artif. Intell.}, 112:57--104, August 1999.

\bibitem{CR11}
Miguel Carrillo and David Rosenblueth.
\newblock Nondeterministic update of {CTL} models by preserving satisfaction
  through protections.
\newblock In Tevfik Bultan and Pao-Ann Hsiung, editors, {\em Automated
  Technology for Verification and Analysis}, volume 6996 of {\em Lecture Notes
  in Computer Science}, pages 60--74. Springer Berlin / Heidelberg, 2011.

\bibitem{CR09}
Miguel Carrillo and David.~A. Rosenblueth.
\newblock A method for {CTL} model update, representing {K}ripke {S}tructures
  as table systems.
\newblock {\em IJPAM}, 52:401--431, January 2009.

\bibitem{GBSK12}
George Chatzieleftheriou, Borzoo Bonakdarpour, Scott~A. Smolka, and Panagiotis
  Katsaros.
\newblock Abstract model repair.
\newblock In {\em Proceedings of the 4th international conference on NASA
  Formal Methods}, NFM'12, pages 341--355, Berlin, Heidelberg, 2012.
  Springer-Verlag.

\bibitem{CES09}
Edmund~M. Clarke, E.~Allen Emerson, and Joseph Sifakis.
\newblock Model checking: Algorithmic verification and debugging.
\newblock {\em Communications of the ACM}, 52(11):74--84, 2009.

\bibitem{CGJLV00}
Edmund~M. Clarke, Orna Grumberg, Somesh Jha, Yuan Lu, and Helmut Veith.
\newblock Counterexample-guided abstraction refinement.
\newblock In {\em Proceedings of the 12th International Conference on Computer
  Aided Verification}, CAV '00, pages 154--169, London, UK, 2000.
  Springer-Verlag.

\bibitem{CGL94}
Edmund~M. Clarke, Orna Grumberg, and David~E. Long.
\newblock Model checking and abstraction.
\newblock {\em ACM Trans. Program. Lang. Syst.}, 16:1512--1542, September 1994.

\bibitem{CPR05}
Byron Cook, Andreas Podelski, and Andrey Rybalchenko.
\newblock Abstraction refinement for termination.
\newblock In Chris Hankin and Igor Siveroni, editors, {\em Static Analysis},
  volume 3672 of {\em Lecture Notes in Computer Science}, pages 87--101.
  Springer Berlin / Heidelberg, 2005.

\bibitem{CC77}
Patrick Cousot and Radhia Cousot.
\newblock Abstract interpretation: a unified lattice model for static analysis
  of programs by construction or approximation of fixpoints.
\newblock In {\em POPL '77}, pages 238--252, New York, NY, USA, 1977. ACM.

\bibitem{CC79}
Patrick Cousot and Radhia Cousot.
\newblock Systematic design of program analysis frameworks.
\newblock In {\em Proceedings of the 6th ACM SIGACT-SIGPLAN symposium on
  Principles of programming languages}, POPL '79, pages 269--282, New York, NY,
  USA, 1979. ACM.

\bibitem{CGR07}
Patrick Cousot, Pierre Ganty, and Jean-François Raskin.
\newblock Fixpoint-guided abstraction refinements.
\newblock In Hanne Nielson and Gilberto Filé, editors, {\em Static Analysis},
  volume 4634 of {\em Lecture Notes in Computer Science}, pages 333--348.
  Springer Berlin / Heidelberg, 2007.

\bibitem{DGG97}
Dennis Dams, Rob Gerth, and Orna Grumberg.
\newblock Abstract interpretation of reactive systems.
\newblock {\em ACM Trans. Program. Lang. Syst.}, 19:253--291, March 1997.

\bibitem{D96}
DR. Dams.
\newblock {\em Abstract interpretation and partition refinement for model
  checking}.
\newblock PhD thesis, Technische Universiteit Eindhoven, 1996.

\bibitem{MLB11}
Maria de~Menezes, Silvio do~Lago~Pereira, and Leliane de~Barros.
\newblock System design modification with actions.
\newblock In Antônio da~Rocha~Costa, Rosa Vicari, and Flavio Tonidandel,
  editors, {\em Advances in Artificial Intelligence – SBIA 2010}, volume 6404
  of {\em Lecture Notes in Computer Science}, pages 31--40. Springer Berlin /
  Heidelberg, 2011.

\bibitem{df09}
Rayna Dimitrova and Bernd Finkbeiner.
\newblock Synthesis of fault-tolerant distributed systems.
\newblock In {\em Proceedings of the 7th International Symposium on Automated
  Technology for Verification and Analysis (ATVA)}, pages 321--336, 2009.

\bibitem{E12}
Ali Ebnenasir.
\newblock Action-based discovery of satisfying subsets: A distributed method
  for model correction.
\newblock {\em Information and Software Technology}, 2012.

\bibitem{EH85}
E.~A. Emerson and Joseph~Y. Halpern.
\newblock Decision procedures and expressiveness in the temporal logic of
  branching time.
\newblock {\em J. Comput. Syst. Sci.}, 30:1--24, February 1985.

\bibitem{fb15}
Fathiyeh Faghih and Borzoo Bonakdarpour.
\newblock {SMT}-based synthesis of distributed self-stabilizing systems.
\newblock {\em ACM Transactions on Autonomous and Adaptive Systems (TAAS)},
  2015.
\newblock To appear.

\bibitem{gr09}
Alain Girault and {\'E}ric Rutten.
\newblock Automating the addition of fault tolerance with discrete controller
  synthesis.
\newblock {\em Formal Methods in System Design (FMSD)}, 35(2):190--225, 2009.

\bibitem{GHJ01}
Patrice Godefroid, Michael Huth, and Radha Jagadeesan.
\newblock Abstraction-based model checking using modal transition systems.
\newblock In {\em Proceedings of the 12th International Conference on
  Concurrency Theory}, CONCUR '01, pages 426--440, London, UK, 2001.
  Springer-Verlag.

\bibitem{GJ02}
Patrice Godefroid and Radha Jagadeesan.
\newblock Automatic abstraction using generalized model checking.
\newblock In {\em Proceedings of the 14th International Conference on Computer
  Aided Verification}, CAV '02, pages 137--150, London, UK, UK, 2002.
  Springer-Verlag.

\bibitem{GS97}
Susanne Graf and Hassen Saidi.
\newblock Construction of abstract state graphs with pvs.
\newblock In Orna Grumberg, editor, {\em Computer Aided Verification}, volume
  1254 of {\em Lecture Notes in Computer Science}, pages 72--83. Springer
  Berlin / Heidelberg, 1997.

\bibitem{GBC06}
Andreas Griesmayer, Roderick Bloem, and Byron Cook.
\newblock Repair of boolean programs with an application to {C}.
\newblock In Thomas Ball and Robert Jones, editors, {\em Computer Aided
  Verification}, volume 4144 of {\em Lecture Notes in Computer Science}, pages
  358--371. Springer Berlin / Heidelberg, 2006.

\bibitem{GLLS07}
Orna Grumberg, Martin Lange, Martin Leucker, and Sharon Shoham.
\newblock When not losing is better than winning: Abstraction and refinement
  for the full mu-calculus.
\newblock {\em Inf. Comput.}, 205:1130--1148, August 2007.

\bibitem{GW10}
Paulo~T. Guerra and Renata Wassermann.
\newblock Revision of {CTL} models.
\newblock In {\em Proceedings of the 12th Ibero-American Conference on Advances
  in Artificial Intelligence}, IBERAMIA'10, pages 153--162, Berlin, Heidelberg,
  2010. Springer-Verlag.

\bibitem{HKMW12}
D.~Harel, G.~Katz, A.~Marron, and G.~Weiss.
\newblock Non-intrusive repair of reactive programs.
\newblock In {\em Engineering of Complex Computer Systems (ICECCS), 2012 17th
  International Conference on}, pages 3--12, july 2012.

\bibitem{HJS01}
Michael Huth, Radha Jagadeesan, and David~A. Schmidt.
\newblock Modal transition systems: A foundation for three-valued program
  analysis.
\newblock In {\em Proceedings of the 10th European Symposium on Programming
  Languages and Systems}, ESOP '01, pages 155--169, London, UK, 2001.
  Springer-Verlag.

\bibitem{HR04}
Michael Huth and Mark Ryan.
\newblock {\em Logic in Computer Science: Modelling and Reasoning about
  Systems}.
\newblock Cambridge University Press, August 2004.

\bibitem{JGB07}
Barbara Jobstmann, Andreas Griesmayer, and Roderick Bloem.
\newblock Program repair as a game.
\newblock In Kousha Etessami and Sriram Rajamani, editors, {\em Computer Aided
  Verification}, volume 3576 of {\em Lecture Notes in Computer Science}, pages
  287--294. Springer Berlin / Heidelberg, 2005.

\bibitem{KPYZ10}
Michael Kelly, Fei Pu, Yan Zhang, and Yi~Zhou.
\newblock {ACTL} local model update with constraints.
\newblock In {\em Proceedings of the 14th international conference on
  Knowledge-based and intelligent information and engineering systems: Part
  IV}, KES'10, pages 135--144, Berlin, Heidelberg, 2010. Springer-Verlag.

\bibitem{LGSBBP95}
C.~Loiseaux, S.~Graf, J.~Sifakis, A.~Bouajjani, S.~Bensalem, and David Probst.
\newblock Property preserving abstractions for the verification of concurrent
  systems.
\newblock {\em Formal Methods in System Design}, 6:11--44, 1995.

\bibitem{ZMK13}
Razieh Nokhbeh~Zaeem, MuhammadZubair Malik, and Sarfraz Khurshid.
\newblock Repair abstractions for more efficient data structure repair.
\newblock In Axel Legay and Saddek Bensalem, editors, {\em Runtime
  Verification}, volume 8174 of {\em Lecture Notes in Computer Science}, pages
  235--250. Springer Berlin Heidelberg, 2013.

\bibitem{PAJTK15}
Shashank Pathak, Erika \'{A}brah\'{a}m, Nils Jansen, Armando Tacchella, and
  Joost-Pieter Katoen.
\newblock A greedy approach for the efficient repair of stochastic models.
\newblock In Klaus Havelund, Gerard Holzmann, and Rajeev Joshi, editors, {\em
  NASA Formal Methods}, volume 9058 of {\em Lecture Notes in Computer Science},
  pages 295--309. Springer International Publishing, 2015.

\bibitem{SDE08}
Roopsha Samanta, Jyotirmoy~V. Deshmukh, and E.~Allen Emerson.
\newblock Automatic generation of local repairs for boolean programs.
\newblock In {\em Proceedings of the 2008 International Conference on Formal
  Methods in Computer-Aided Design}, FMCAD '08, pages 27:1--27:10, Piscataway,
  NJ, USA, 2008. IEEE Press.

\bibitem{SG04}
Sharon Shoham and Orna Grumberg.
\newblock Monotonic abstraction-refinement for {CTL}.
\newblock In Kurt Jensen and Andreas Podelski, editors, {\em Tools and
  Algorithms for the Construction and Analysis of Systems}, volume 2988 of {\em
  Lecture Notes in Computer Science}, pages 546--560. Springer Berlin /
  Heidelberg, 2004.

\bibitem{SG07}
Sharon Shoham and Orna Grumberg.
\newblock A game-based framework for {CTL} counterexamples and 3-valued
  abstraction-refinement.
\newblock {\em ACM Trans. Comput. Logic}, 9, 2007.

\bibitem{SJB05}
Stefan Staber, Barbara Jobstmann, and Roderick Bloem.
\newblock Finding and fixing faults.
\newblock In Dominique Borrione and Wolfgang Paul, editors, {\em Correct
  Hardware Design and Verification Methods}, volume 3725 of {\em Lecture Notes
  in Computer Science}, pages 35--49. Springer Berlin / Heidelberg, 2005.

\bibitem{BGS07}
Jan Van~den Bussche, Dirk Van~Gucht, and Stijn Vansummeren.
\newblock Well-definedness and semantic type-checking for the nested relational
  calculus.
\newblock {\em Theor. Comput. Sci.}, 371(3):183--199, feb 2007.

\bibitem{VYY2010}
Martin Vechev, Eran Yahav, and Greta Yorsh.
\newblock Abstraction-guided synthesis of synchronization.
\newblock In {\em Proceedings of the 37th annual ACM SIGPLAN-SIGACT symposium
  on Principles of programming languages}, POPL '10, pages 327--338, New York,
  NY, USA, 2010. ACM.

\bibitem{EJ12}
Christian von Essen and Barbara Jobstmann.
\newblock Program repair revisited.
\newblock Technical Report TR-2012-4, {Verimag} Research Report, 2012.

\bibitem{WC08}
Farn Wang and Chih-Hong Cheng.
\newblock Program repair suggestions from graphical state-transition
  specifications.
\newblock In {\em Proceedings of the 28th IFIP WG 6.1 international conference
  on Formal Techniques for Networked and Distributed Systems}, FORTE '08, pages
  185--200, Berlin, Heidelberg, 2008. Springer-Verlag.

\bibitem{WPFSBMZ10}
Yi~Wei, Yu~Pei, Carlo~A. Furia, Lucas~S. Silva, Stefan Buchholz, Bertrand
  Meyer, and Andreas Zeller.
\newblock Automated fixing of programs with contracts.
\newblock In {\em Proceedings of the 19th international symposium on Software
  testing and analysis}, ISSTA '10, pages 61--72, New York, NY, USA, 2010. ACM.

\bibitem{WV95}
Jeannette~M. Wing and Mandana Vaziri-Farahani.
\newblock Model checking software systems: a case study.
\newblock In {\em Proceedings of the 3rd ACM SIGSOFT symposium on Foundations
  of software engineering}, SIGSOFT '95, pages 128--139, New York, NY, USA,
  1995. ACM.

\bibitem{ZD08}
Yan Zhang and Yulin Ding.
\newblock {CTL} model update for system modifications.
\newblock {\em J. Artif. Int. Res.}, 31:113--155, January 2008.

\bibitem{ZKZ10}
Yan Zhang, Michael Kelly, and Yi~Zhou.
\newblock Foundations of tree-like local model updates.
\newblock In {\em Proceeding of the 2010 conference on ECAI 2010: 19th European
  Conference on Artificial Intelligence}, pages 615--620, Amsterdam, The
  Netherlands, The Netherlands, 2010. IOS Press.

\end{thebibliography}

\end{document}